\documentclass[dvipsnames]{article}
\usepackage[left=3cm,right=3cm,top=3cm,bottom=3cm]{geometry}

\usepackage{amsmath} 
\usepackage{amsthm}
\usepackage{amsfonts}
\usepackage{amssymb}
\usepackage{mathrsfs}

\let\originalleft\left
\let\originalright\right
\renewcommand{\left}{\mathopen{}\mathclose\bgroup\originalleft}
\renewcommand{\right}{\aftergroup\egroup\originalright}

\theoremstyle{plain}
\newtheorem{lemm}{Lemma}[subsection]
\newtheorem{theo}[lemm]{Theorem}
\newtheorem{prop}[lemm]{Proposition}
\newtheorem{corol}[lemm]{Corollary}

\theoremstyle{definition}
\newtheorem{defi}[lemm]{Definition}
\newtheorem{example}[lemm]{Example}
\theoremstyle{remark}
\newtheorem{rema}[lemm]{Remark}

\usepackage[unicode, breaklinks, hidelinks, hyperindex=true]{hyperref}
\usepackage{makeidx}
\makeindex
\let\oldemph\emph
\renewcommand{\emph}[1]{\oldemph{#1}\index{#1}}

\usepackage{cleveref}
\crefname{theo}{theorem}{theorems}
\crefname{lemm}{lemma}{lemmas}
\crefname{defi}{definition}{definitions}

\usepackage{lineno}
\usepackage{xspace}
\usepackage{xifthen}
\usepackage{xparse}
  {\begin{itemize}\ignorespaces}%
  {\end{itemize}\ignorespacesafterend}

\newcommand           {\N}{\mathbb{N}}

\DeclareMathOperator  {\Inf}{Inf}
\DeclareMathOperator  {\Fin}{Fin}

           \newcommand{\Vertices}{V}
\newcommand           {\VerticesOf}[1]{\Vertices_{#1}}
\newcommand           {\vtx}{v}             

 \DeclareMathOperator{\Reach} {Reach}

\DeclareMathOperator  {\Parity}{Parity}
\DeclareMathOperator  {\Rabin}{Rabin}
\DeclareMathOperator  {\Streett}{Streett}

           \newcommand{\g}{\bar{g}}      

\newcommand           {\ch}{C}
\newcommand           {\prov}[1]{\prover_#1}
\newcommand           {\prover}{P}

\usepackage{soul}

\newcommand{\JF}[2][]{\ul{#1}\textcolor{NavyBlue}{#2}} 
\newcommand{\Ve}[1]{\textcolor{Magenta}{#1}}
\newcommand{\Ma}[1]{\textcolor{PineGreen}{#1}}

\usepackage[english]{babel}
\usepackage{xspace}
\usepackage{comment}

\newcommand{\NCRS}{SPE-NCRS\xspace}

\NewDocumentCommand{\pcp}{}{$P_1CP_2$\xspace}
\NewDocumentCommand{\pcpG}{}{$P_1CP_2(\mathcal{G})$\xspace}

\NewDocumentCommand{\pc}{}{$PC$\xspace}
\NewDocumentCommand{\pcG}{}{$PC(\mathcal{G})$\xspace}
\NewDocumentCommand{\pcpOG}{}{$P_1CP_2{\text -}\mathcal{O}(\mathcal{G})$\xspace}
\NewDocumentCommand{\pcOG}{}{$PC{\text -}\mathcal{O}(\mathcal{G})$\xspace}
\newcommand{\Obs}{\mathcal{O}bs}
\newcommand{\ObsAut}{\mathcal{O}}
\newcommand{\NegComponent}{$(G'\! \setminus \! G)$-component\xspace}
\newcommand{\NegComponents}{$(G'\! \setminus \! G)$-components\xspace}
\newcommand{\NegG}{V'_{(G'\! \setminus \! G)}}
\newcommand{\OComponent}{$\mathcal{O}$-component\xspace}
\newcommand{\OComponents}{$\mathcal{O}$-components\xspace}

\newcommand{\APun}{A'_{P_1}}
\newcommand{\APd}{A'_{P_2}}
\newcommand{\AC}{A'_{C}}
\newcommand{\VPun}{V'_{P_1}}
\newcommand{\VPd}{V'_{P_2}}
\newcommand{\VC}{V'_{C}}

\newcommand{\Plays}[1]{\mathsf{Plays}_{#1}}
\newcommand{\Hist}[1]{\mathsf{Hist}_{#1}}
\newcommand{\Histi}[2]{\mathsf{Hist}_{#1}^{#2}}

\newcommand{\pstable}{player-stable\xspace}
\newcommand{\Pstable}{Player-stable\xspace}
\newcommand{\pstability}{player-stability\xspace}

\newcommand{\Pstability}{Player-stability\xspace}
\newcommand{\StronlgyPstable}{strongly player-stable\xspace}
\newcommand{\stronglyPstability}{strong player-stability\xspace}
\newcommand{\StronglyPstability}{Strong player-stability\xspace}
\newcommand{\actionstable}{action-unique\xspace}
\newcommand{\actionstability}{action-unicity\xspace}
\newcommand{\visibleactionstable}{action-stable\xspace}
\newcommand{\visibleactionstability}{action-stability\xspace}
\newcommand{\Visibleactionstability}{Action-stability\xspace}

\newcommand{\ARondePd}{{\mathbf{\mathbb{A}}}_{P_2}}

\newcommand{\corresp}[1]{\varphi(#1)}
\newcommand{\correspbis}[1]{\phi(#1)}
\newcommand{\cor}[1]{\varphi(#1)}

\newcommand{\simu}[1]{\mathsf{sim}(#1)}
\newcommand{\gainSim}[1]{\mathsf{simGain}(#1)}

\newcommand{\alphabet}{\Sigma}

\newcommand{\languageOf}{L}
\newcommand{\automaton}{{\cal A}}

\usepackage{tikz}
\usetikzlibrary{shapes}
\usepackage{hyperref}
\usepackage{stmaryrd}
\usepackage{subcaption}
\usetikzlibrary{arrows,topaths}
\usetikzlibrary{fit,matrix}
\usetikzlibrary{automata, positioning}
\usetikzlibrary{shapes.multipart}
\tikzset{
  node/.style = {draw, minimum size=1em, inner sep=0pt, font=\scriptsize
    font=\sffamily}
}

\tikzset{
state1P/.style={
		rectangle,
		fill=gray!50,
		rounded corners,
       draw=black, thick,
       minimum height=2em,
       inner sep=5pt,
       text centered,
       }
}

\tikzset{
state2Pr/.style={
		rectangle split,
       rectangle split horizontal=true,
       rectangle split parts=2,
       rectangle split part fill={none,gray!30},
       draw=black, thick,
       minimum height=2em,
       inner sep=3pt,
       text centered,
       }
}

\tikzset{
state2Prc/.style={
		rectangle split,
       rectangle split horizontal=true,
       rectangle split parts=2,
       rectangle split part fill={none,gray!30},
       rounded corners,
       draw=black, thick,
       minimum height=2em,
       inner sep=3pt,
       text centered,
       }
}

\tikzset{
state2Pc/.style={
		circle split,
       draw=black, thick,
       minimum height=2em,
       inner sep=3pt,
       text centered,
       }
}

\tikzset{
state3P/.style={
		rectangle split,
       rectangle split horizontal,
       rectangle split parts=3,
       rectangle split part fill={none,gray!50,none},
       rounded corners,
       draw=black, thick,
       minimum height=2em,
       inner sep=3pt,
       text centered,
       }
}

\begin{document}
\title{The Non-Cooperative Rational Synthesis Problem for Subgame Perfect Equilibria and $\omega$-regular Objectives\footnote{This work is partially supported by Fondation ULB (https://www.fondationulb.be/en/), the Thelam fund 2024-F2150080-0021312 (Open Problems on the Decidability and Computational Complexity of Infinite Duration Games), and by FNRS under PDR Grant T.0023.22 (Rational).}}
\author{Véronique Bruyère$^1$ \and Jean-François Raskin$^2$ \and Alexis Reynouard$^2$ \and Marie Van Den Bogaard$^3$}

\date{$^1$Université de Mons, UMONS, Belgium\\ $^2$Université libre de Bruxelles, ULB, Belgium\\ $^3$Univ. Gustave Eiffel, CNRS, LIGM, Marne-la-Vallée, France}

\maketitle

\begin{abstract}
 This paper studies the rational synthesis  problem for multi-player games played on graphs when rational players are following subgame perfect equilibria. In these games, one player, the system, declares his strategy upfront, and the other players, composing the environment, then rationally respond by playing  strategies forming a subgame perfect equilibrium. We study the complexity of the rational synthesis problem when the players have $\omega$-regular objectives encoded as parity objectives. Our algorithm is based on an encoding into a three-player game with imperfect information, showing that the problem is in 2ExpTime. When the number of environment players is fixed, the problem is in ExpTime and is NP- and coNP-hard. Moreover, for a fixed number of players and reachability objectives, we get a polynomial algorithm.
\end{abstract}

\section{Introduction}

Studying non zero-sum games played on graphs of infinite duration with multiple players~\cite{DBLP:conf/lata/BrenguierCHPRRS16,DBLP:conf/dlt/Bruyere17} poses both theoretical and algorithmic challenges. This paper primarily addresses the (non-cooperative) rational synthesis problem for $n$-player non zero-sum games featuring {\em $\omega$-regular} objectives.
In this context, the goal is to algorithmically determine the existence of a strategy $\sigma_0$ for the system (also called player~0) to enforce his objective against any {\em rational} response from the environment (players $1,\dots,n$). So, {\em rational synthesis} supports the automatic synthesis of systems wherein the environment consists of multiple agents having each their {\em own} objective. These agents are assumed to act rationally towards their own objective rather than being {\em strictly antagonistic} (to the system). This approach contrasts with the simpler scenario of zero-sum two-player game graphs, the fully antagonistic setting, a framework extensively explored in earlier reactive synthesis research, see~\cite{DBLP:reference/mc/BloemCJ18} and the numerous references therein.

While the computational complexity of rational synthesis, where {\em rationality} is defined by the concept of {\em Nash equilibrium} (NE), has been explored in~\cite{DBLP:conf/icalp/ConduracheFGR16}, this paper revisits the rational synthesis problem using the more encompassing definition of {\em subgame perfect equilibrium} (SPE) to formalize rationality. Nash equilibria (NEs) have a known limitation in {\em sequential} games, including the infinite-duration games on graphs that we consider here: they are prone to {\em non-credible threats}. Such threats involve decisions within subgames (potentially reached after a deviation from the equilibrium) that, while not rational, are intended to coerce other players into specific behaviors. To address this limitation, the concept of subgame-perfect equilibrium (SPE) was introduced, as discussed in~\cite{Osborne94}. An SPE is a profile of strategies that form Nash equilibria in every subgame, thereby preventing non-rational and thus non-credible threats. Although SPEs align more intuitively with sequential games, their algorithmic treatment in the context of infinite-duration graph games remains underdeveloped. This gap persists primarily because SPEs require more complex algorithmic techniques compared to NEs. Moreover, the standard backward induction method used in {\em finite} duration sequential games~\cite{kuhn53} cannot be directly applied to {\em infinite} duration games due to the non-terminating nature of these interactions.

Kupferman et al. introduced rational synthesis in two distinct forms. The first approach, dubbed {\em cooperative} rational synthesis~\cite{DBLP:conf/tacas/FismanKL10}, considers an environment working collaboratively with the system to determine whether a specific SPE that ensures the specification of the system is met. So, in this model, the agents of the environment engage in an SPE that guarantees a win for player~0, provided such an equilibrium exists. At the opposite, the second approach, termed {\em non-cooperative} rational synthesis (NCRS)~\cite{DBLP:journals/amai/KupfermanPV16}, grants the environment greater flexibility. Here, the system first selects a fixed strategy $\sigma_0$ and then the environmental agents respond by selecting any strategy that form an SPE with the fixed strategy $\sigma_0$ of the system. The central algorithmic challenge is determining whether there exists a strategy $\sigma_0$ for the system such that every resulting $\sigma_0$-fixed SPE ensures the specification of the system is upheld (\NCRS problem). 

The computational complexity of decision problems for the cooperative synthesis problem with subgame-perfect equilibria (SPE) is now well understood. While the decidability of this problem was first established in~\cite{DBLP:conf/fsttcs/Ummels06}, its exact complexity was resolved in~\cite{DBLP:conf/csl/BriceRB22}, where the problem was shown to be NP-complete for parity objectives.

In contrast, the computational complexity of the \NCRS problem remains less thoroughly investigated. Although the decidability of this problem can be determined through an encoding in Strategy Logic~\cite{DBLP:conf/concur/MogaveroMPV12}, such an encoding does not provide clear insights into the effective construction of the strategy $\sigma_0$ and is suboptimal from an algorithmic standpoint. For example, for rational environment behaviors modeled by Nash equilibria (NEs) instead of SPEs, the NE-NCRS problem can be solved in 3ExpTime using a Strategy Logic encoding~\cite{DBLP:journals/amai/KupfermanPV16}, while it can be solved in 2ExpTime through the use of tree automata~\cite{DBLP:conf/tacas/KupfermanS22}. This shows that using a reduction to Strategy Logic does not deliver optimal worst-case complexity. Addressing this gap, our contribution is twofold. First, we introduce an \emph{innovative} algorithm that transforms the \NCRS problem into a three-player imperfect information game. Such games were analyzed in~\cite{DBLP:conf/icalp/Chatterjee014}, where computationally optimal algorithms for their analysis were presented. Second, our reduction offers a clear advantage over the Strategy Logic encoding, providing improved complexity (double exponential time as opposed to triple exponential time for LTL objectives in~\cite{DBLP:journals/amai/KupfermanPV16})). It also enables a precise analysis of the algorithmic complexity when the number of players is \emph{fixed}, a consideration that is practically relevant in cases where the environment consists of a limited number of players.

\begin{figure} 
\begin{center}
	\begin{tikzpicture}[->, >=latex,shorten >=1pt, scale=1, every node/.style={scale=1, align=center}]	
	\node[draw, rounded corners] (SPE-NCRSp) at (5, 11) {\NCRS problem on multi-player \\ non zero-sum games \\ with parity objectives};
	\node[draw, rounded corners] (pr-g-ii) at (5, 9) {Three-player \pcp game \\ with imperfect information};
	\node[draw, rounded corners, dashed] (pr-g-parity) at (10, 7.5) {Three-player \pcp game \\ with imperfect information \\ and a Rabin objective};
	\node[draw, rounded corners] (2-pl-pr-g-ii) at (5, 6) {Two-player zero-sum \pc game \\ with imperfect information \\ and a Rabin objective};
	\node[draw, rounded corners] (pr-g-pi) at (5, 3) {Two-player zero-sum game \\ with perfect information \\ and a parity objective};
	
	\path[<->] (SPE-NCRSp) edge[double] node[right] {Thm.~\ref{thm:equivalence}} (pr-g-ii);
	\path[<->] (pr-g-ii) edge[double] (2-pl-pr-g-ii);
	\path[<->] (pr-g-ii) edge node[above right] {Sec.~\ref{subsec:pr-g-parity}} (pr-g-parity);
	\path[<->] (pr-g-parity) edge node[below right] {Thm.~\ref{thm:2Provers->1Prover}} (2-pl-pr-g-ii);
	\path[<->] (2-pl-pr-g-ii) edge[double]  node[right] {Thm.~\ref{thm:PerfectInfo}} (pr-g-pi);
	\path[<->] (SPE-NCRSp) edge[double, bend right=90]  node[left] {Thm.~\ref{thm:main}} (pr-g-pi);

	\end{tikzpicture}
\end{center}
\caption{Structure of the article}
\label{fig_structure}
\end{figure}
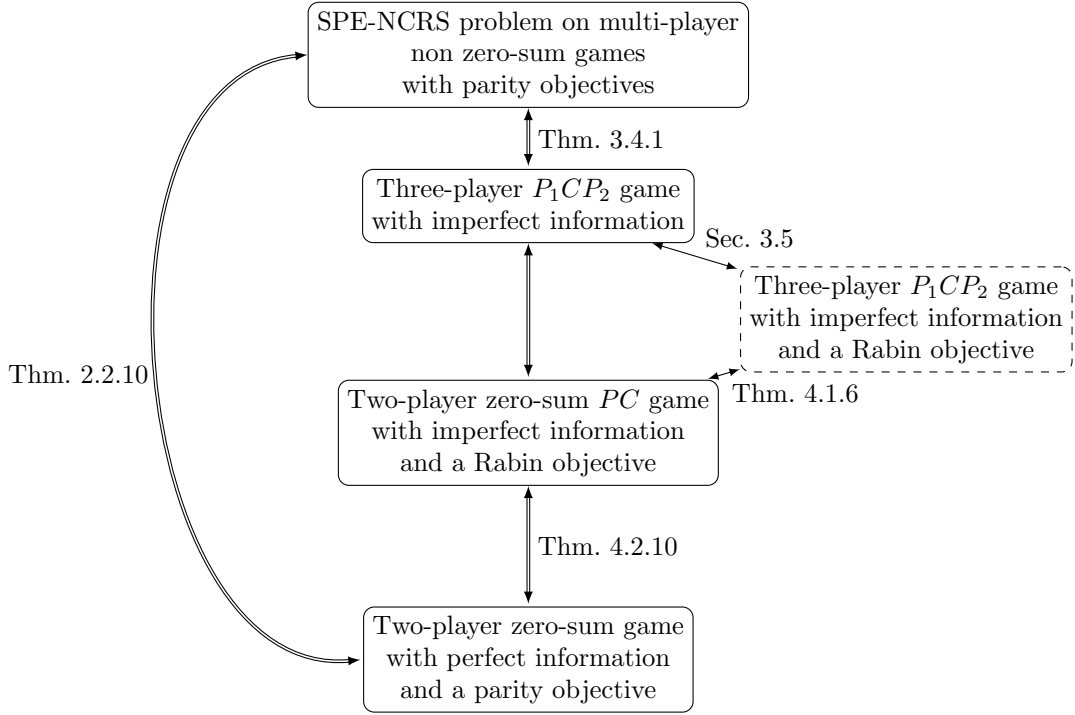

\paragraph{Technical contributions and sketch of our solution}
The summary of our technical contributions is depicted in Figure~\ref{fig_structure}. Our main result shows how to transform the \NCRS problem into a two-player zero-sum parity game with perfect information (Theorem~\ref{thm:main}), a well-studied class of games for which algorithms are available. (Thus, our algorithm supports LTL specifications via their translation into deterministic parity automata, yielding parity
conditions on the game graph). The transformation is structured in several non-trivial steps. 

First, to solve the \NCRS problem, we propose to use the {\em Prover-Challenger} framework initially introduced for the development of algorithms capable of determining the presence of a simulation relation between transition systems (for more details and extensions of this concept, see e.g.~\cite{DBLP:conf/concur/AlurHKV98,DBLP:conf/concur/HenzingerKR97,DBLP:conf/ijcai/Milner71}). However, in our context, we need to use \emph{two Provers} $P_1$ and $P_2$ with the Challenger $C$: $P_1$ aims to demonstrate the existence of a solution, i.e., a strategy $\sigma_0$ for the system, to the \NCRS problem, while $C$ seeks to counter this assertion, i.e., with a subgame perfect response $\bar{\sigma}_{-0}$ that results in an unfavorable outcome for player~0. Then, $P_2$ endeavors to demonstrate that the combined profile $(\sigma_0,\bar{\sigma}_{-0})$ is either not a $0$-fixed SPE or its outcome favors player~$0$. To ensure that $\sigma_0$ is fixed and cannot be modified in subgames, we prevent $P_1$ from adjusting his strategy based on the interactions between $C$ and $P_2$ by imposing to him {\em imperfect information} of the game. More intuition and the formal definition of this \pcp{} game are given in Section~\ref{sec:SPE-NCRS_PCP} together with a proof of correctness (Theorem~\ref{thm:equivalence}).

Second, we detail our method for solving the \pcp game. Given that this game involves three players with imperfect information, specialized techniques are essential for its resolution as multi-player games with imperfect information are undecidable in general (see e.g.~\cite{DBLP:conf/focs/PetersonR79}). We employ a solution specifically adapted for our context, derived from a transformation introduced in~\cite{DBLP:conf/icalp/Chatterjee014}. This transformation was originally proposed for addressing similar types of three-player games with imperfect information. Due to the intricate winning condition present in our Prover-Challenger reduction, we detail its translation into an explicit Rabin objective (this is done in Section~\ref{subsec:pr-g-parity}). After that transformation of the winning condition, the aforementioned techniques for three-player games with imperfect information can be adapted. As a result of this transformation, we obtain a more conventional two-player zero-sum game with imperfect information with a Rabin objective (Theorem~\ref{thm:2Provers->1Prover}).

Third, techniques to remove imperfect information, see e.g.~\cite{DBLP:journals/lmcs/RaskinCDH07,DBLP:journals/jcss/Reif84}, can then be used to obtain the desired two-player zero-sum parity game with perfect information (Theorem~\ref{thm:PerfectInfo}). Therefore, solving the \NCRS problem reduces to solving this two-player zero-sum parity game for which algorithms are well-known.

In Section~\ref{sec:complexity}, we provide a detailed complexity analysis of each step of our construction. This analysis enables us to derive detailed complexity results: solving the \NCRS problem for parity objectives is exponential in the size of the graph and the number of priorities used by the parity objectives, and double-exponential in the number of players (a PSpace lower bound can be deduced from~\cite{DBLP:conf/icalp/ConduracheFGR16}). Consequently, when the number of players is fixed, we achieve membership in ExpTime. Furthermore, by adapting proofs from the NE-NCRS problem~\cite{DBLP:conf/concur/ConduracheOT18}, we establish NP- and co-NP-hardness (the NE-NCRS problem is in PSpace and both NP- and co-NP-hard). Finally, for the specific case of reachability objectives, we obtain a polynomial-time algorithm when the number of players is fixed, as is the case for the NE-NCRS problem~\cite{DBLP:conf/concur/ConduracheOT18}.

\paragraph{\bf Related work}
Recent literature, such as surveys in \cite{DBLP:conf/lata/BrenguierCHPRRS16,DBLP:conf/dlt/Bruyere17,DBLP:journals/siglog/Bruyere21}, underscores a growing interest in non zero-sum games for synthesis. 

Algorithms have been developed for reasoning about NEs in graph games, both for $\omega$-regular \cite{DBLP:conf/fossacs/Ummels08} and quantitative objectives \cite{DBLP:conf/lfcs/BrihayePS13}, or even in a concurrent setting~\cite{DBLP:journals/corr/BouyerBMU15}. The concept of secure equilibrium, a refinement of NE, was introduced in \cite{DBLP:journals/tcs/ChatterjeeHJ06} and its synthesis potential was demonstrated in later studies \cite{DBLP:conf/tacas/ChatterjeeH07}. Similarly, doomsday equilibria, an expansion of secure equilibria for $n$-player games, is elaborated upon in \cite{DBLP:journals/iandc/ChatterjeeDFR17}. All those results are related to the notion of NE.

Algorithms to reason on SPEs are more recent and usually more intricate. For that reason, progresses on weaker notions like weak SPE~\cite{DBLP:journals/iandc/BrihayeBGR21,DBLP:conf/csl/BrihayeBMR15,DBLP:journals/iandc/BruyereRPR21} have been needed before establishing tight complexity bounds for SPEs. For SPEs, the exact complexity for the constrained existence for reachability games was established in~\cite{DBLP:journals/lmcs/BrihayeBGRB20}, for parity games in~\cite{DBLP:conf/csl/BriceRB22}, and for mean-payoff games in~\cite{DBLP:conf/icalp/BriceRB22}. All those works introduce new algorithms that are substantially more sophisticated than those needed to study the notion of NE. 
As mentioned previously, the SPE-NCRS problem with LTL objectives can be addressed by reducing it to the model-checking problem of Strategy Logic, as shown in~\cite{DBLP:journals/amai/KupfermanPV16}. However, this encoding results in a triple-exponential complexity, even when the number of players is fixed. Consequently, this approach does not allow for a fine-grained complexity analysis when the number of players is treated as a fixed parameter, and so it cannot yield the ExpTime complexity we have achieved in this case, nor the PTime complexity for a fixed number of players with reachability objectives.

Cooperative rational synthesis was first introduced in~\cite{DBLP:conf/tacas/FismanKL10}. The adversarial version was later introduced in~\cite{DBLP:journals/amai/KupfermanPV16}. In both cases, as mentioned above, the decidability results were given through a reduction to Strategy Logic~\cite{DBLP:conf/concur/MogaveroMPV12}. A more detailed analysis of the complexity for a large variety of $\omega$-regular objectives and for rationality expressed as NEs was given in~\cite{DBLP:conf/icalp/ConduracheFGR16} for turn-based games and in~\cite{DBLP:conf/concur/ConduracheOT18} for concurrent games, for LTL objectives in~\cite{DBLP:conf/tacas/KupfermanS22}, and for two players and for mean-payoff and discounted-sum objectives in~\cite{NoncoopSynth_Meanpayoff_2players}. Those results do not cover SPEs as we do in this paper. 

\emph{Rational verification} (instead of rational synthesis)
studies the problem of verifying that a given strategy $\sigma_0$ for the system is a solution to the NE/SPE-NCRS problem. The complexity of this simpler problem has been studied for several kinds of objectives in \cite{DBLP:conf/mfcs/BriceRB23,DBLP:journals/ai/GutierrezNPW20,DBLP:journals/amai/GutierrezNPW23}.

Another notion of rational environment behavior  treats the environment as a single agent but with multiple, sometimes conflicting, goals, aiming for behaviors that achieve a \emph{Pareto-optimal} balance among these objectives. Both rational synthesis and verification have been very recently studied for this concept of rationality~\cite{gaspard-pareto-quanti-reach,DBLP:journals/tocl/BruyereFRT24,DBLP:conf/concur/BruyereGR24,Pareto-Rational-Verification}.

\paragraph{\bf Structure of the paper}
In Section~\ref{sec:prel}, we introduce the necessary preliminaries on games played on graphs, the notion of subgame perfection for rationality, and the \NCRS problem that we study in this paper. In Section~\ref{sec:SPE-NCRS_PCP}, we show how to reduce this problem to a Prover-Challenger game with two Provers, called the \pcp game. In Section~\ref{sec:SolvingPCPGame}, we provide an algorithm to solve the \pcp game. This algorithm is structured in two important steps. The first step transforms the game with three players with imperfect information into a game with two players with imperfect information. The second step eliminates the imperfect information and leads to a two-player zero-sum game with a parity objective. As this kind of game is well-studied, we are then equipped to present in Section~\ref{sec:complexity} a detailed complexity analysis of our algorithm, including an analysis when the number of players in the environment is fixed. In Section~\ref{sec:conclusion}, we end the article with a conclusion.

\section{Preliminaries} \label{sec:prel}

In this section, we recall the necessary notions and concepts underpinning this work: First, we specify the model of games on graphs that we study, along with the corresponding definitions of plays, (winning) strategies, and objectives. Second, we present the solution concepts relevant to our non-zero-sum model and specific problem, ranging from Nash equilibria to $0$-fixed subgame-perfect equilibria. 
Finally, we provide a precise statement of the \NCRS problem (see Definition~\ref{def:SPE-NCRSp}) and our main results (see Theorems~\ref{thm:main} and \ref{thm:mainReach}).

\subsection{Games on Graphs}

\begin{defi}[Game structure]\label{def:game_structure}
A \emph{game structure} is a tuple $G = (V, A, \Pi, \delta, v_0
)$, where:
\begin{itemize}
    \item $\Pi = \lbrace 0, \dots, n \rbrace$ is a finite set of \emph{players},
    \item $A = \bigcup_{i=0}^n A_i$ is the set of \emph{actions} of the players, such that $A_i$ is the action set of player $i$, $i \in \Pi$, and $A_0 \cap (\bigcup_{i\neq 0} A_i) = \varnothing$, 
    \item $V = \bigcup_{i=0}^n V_i$ is the set of \emph{states}, such that $V_i$ is the state set of player $i$, $i \in \Pi$, and 
    $V_i \cap V_j = \varnothing$ for all $i \neq j$,
    \item $v_0 \in V$ is the \emph{initial state},
    \item $\delta: \bigcup_{i=0}^n (V_i \times A_i) \rightarrow V$ is a partial function called the \emph{transition function}, such that:
    \begin{enumerate}
        \item $G$ is \emph{deadlock-free}: for every state $v \in V_i$, there exists $a \in A_i$ an action of player~$i$, such that $\delta(v,a)$ is defined,
        \item $G$ is \emph{\actionstable}: for every state $v \in V_i$, and for all  $a,b \in A_i$ actions of player~$i$, we have 
        $\delta(v,a) = \delta(v,b) \Longleftrightarrow a=b$.
    \end{enumerate}
\end{itemize}
The \emph{size} of a game structure is given by the numbers $|V|$, $|A|$, and $|\Pi|$ of its vertices, actions, and players respectively.
\end{defi}

We say that a state $v \in V_i$ is \emph{controlled} or \emph{owned} by player~$i$.
Note that the condition $A_0 \cap (A \setminus A_0) = \varnothing$ means that one knows from the actions when player~$0$ is the one that is playing.  This property of the action set of player~$0$ is not classical nor necessary but will reveal itself useful in the remaining of the paper. Indeed, in the setting we study here, player~$0$ has a distinguished role compared to the other players, see Section~\ref{sec:SPE-NCRS_PCP}.
Condition~$1$ on the transition function ensures that in every state, there is always a possible action to play. Condition~$2$ requires that a transition between two states can only be achieved via a unique action.

\begin{defi}[Play and history]
A \emph{play} in a game structure $G$ is an infinite sequence $\rho \in (VA)^\omega$ of states and actions, such that $\rho$ is of the form $v_0a_0v_1a_1 \cdots$ with $v_0$ the initial state and where for every $k\in \N$, we have $\delta(v_k,a_k) = v_{k+1}$.
The set of all plays in a game structure $G$ is denoted $\Plays{G}$. 
A \emph{history} is a finite prefix $h \in (VA)^*V$ of a play ending in a state of $G$, that we describe as $h = v_0 a_0 v_1 a_1 \dots v_k$. The set of all histories in a game structure $G$ is denoted $\Hist{G}$, while for $i\in \Pi$, the set of histories ending in a state controlled by player~$i$ is denoted $\Histi{G}{i}$.
We write $h \sqsubset \rho$ if the history $h$ is a prefix of the play $\rho$. We also use both notations $\sqsubset$ and $\sqsubseteq$ for two histories.
\end{defi}

\begin{defi}[Strategy]\label{def:strategy}
A \emph{strategy} for player $i \in \Pi$ is a function $\sigma: \Histi{G}{i} \rightarrow A_i$, that prescribes an action $\sigma(h)$ for player~$i$ to choose for every history where it is his turn to play. The set of all strategies $\sigma$ for player~$i$ is denoted by $\Sigma_i$. 
A collection $\bar{\sigma} = (\sigma_i)_{i\in \Pi}$ of strategies, one for each player, is called a \emph{profile} of strategies.
A play $\rho = v_0 a_0 v_1 a_1 \dots$ is \emph{compatible} with a strategy $\sigma_i$ of player~$i$ if for every $k \in \N$  such that $v_k \in V_i$, we have $a_k = \sigma_i(v_0 a_0 \dots v_k)$. Histories compatible with a strategy are defined similarly. A strategy $\sigma$ is said to be \emph{memoryless} when the prescribed action only depends on the last visited state, that is, for every $hav, h'a'v \in \Histi{G}{i}$, we have $\sigma(hav) = \sigma(h'a'v)$.
 
Given a profile of strategies $\bar{\sigma} = (\sigma_i)_{i\in \Pi}$, there is a unique play starting in $v_0$ that is compatible with every strategy of the profile, that we call the \emph{outcome} of $\bar{\sigma}$ and denote it by $\langle \bar{\sigma} \rangle_{v_0}$. Given a strict subset of players $\Pi' \subset \Pi$, we write $\bar{\sigma}_{-\Pi'}$ to refer to a partial profile of strategies that contains a strategy for each player except the ones in $\Pi'$. In particular, we will often focus on the strategy of one player, say $\Pi' = \{i\}$, and the profile of strategies of the rest of the players, and use the notation $(\sigma_i, \bar{\sigma}_{-i})$ to denote the complete profile of strategies. 
\end{defi}

\begin{defi}[Objective and game]
A \emph{winning condition}, or \emph{objective} for player~$i$ is a subset $W_i \subseteq \Plays{G}$ of plays in the game structure $G$. 
We say that a play $\rho$ is \emph{winning} for player~$i$ or \emph{satisfies} his objective if $\rho \in W_i$. Otherwise, we say that $\rho$ is \emph{losing} for player~$i$.
A \emph{game} is a pair $\mathcal{G} = (G, (W_i)_{i \in \Pi})$ consisting in a game structure $G$, together with a profile of objectives for the players. When the context is clear, we often only write $\mathcal{G}$ to designate a game. Given a strategy profile $\bar \sigma$, its \emph{gain profile} (or simply \emph{gain}) is the Boolean vector $\bar g$ such that for all $i \in \Pi$, we have $g_i = 1$ if $\langle \bar{\sigma} \rangle_{v_0} \in W_i$, and $g_i = 0$ if $\langle \bar{\sigma} \rangle_{v_0} \not\in W_i$. We also say that $\bar g$ is the gain profile of the outcome of $\bar \sigma$.
\end{defi}

In this paper, we consider the concept of \emph{parity} objective (that encodes all classical $\omega$-regular objectives), as well as the particular case of \emph{reachability} objective. Both concepts are defined in the following way:

\begin{defi}[Reachability objective] \label{def:reach-objective}
For each $i \in \Pi$, let $T_i \subseteq V$ be a \emph{target set}.
The reachability objective for player~$i$ is then $W_i = \lbrace \rho = v_0a_0v_1a_1 \ldots \in \Plays{G} ~\vert~ \exists k \in \N, v_k \in T_i \rbrace$ that we also denote by $\Reach(T_i)$. Given such a profile of target sets, we say that $(G,(\Reach(T_i))_{i \in \Pi})$ is a \emph{reachability game}.
\end{defi}

The plays of a reachability objective $W_i$ are those visiting a state of the target set $T_i$. 
On the other hand, the parity objective has constraints on states that appear infinitely often in a play.
For a play $\rho = v_0 a_0 v_1 a_1 \dots$, let $\mathrm{Inf}(\rho) = \lbrace v \in V ~\vert~ \exists^{\infty}~k,~ v_k = v \rbrace$, that is, the set of states that appear infinitely often in $\rho$. 
The plays of a parity objective are those whose smallest priority visited infinitely often is even. 

\begin{defi}[Parity objective] \label{def:parity-objective}
For each $i \in \Pi$, let $C_i = \{0,1, \ldots, d_i\}$ be a finite set of \emph{priorities} and $\alpha_i: V \rightarrow C_i$ be a \emph{priority function}, that is, a function that assigns a priority to each state of the game.
The parity objective for player~$i$ is then $W_i = \lbrace \rho \in \Plays{G} ~\vert~ \min_{v \in \mathrm{Inf}(\rho)}(\alpha_i(v)) \text{ is even} \rbrace$ that we also denote by $\Parity(\alpha_i)$.
Given such a profile of priority functions, we say that $(G,(\Parity(\alpha_i))_{i \in \Pi})$ is a \emph{parity game}.
The \emph{size} of each parity objective, denoted by $|\alpha_i|$, is the maximum priority $d_i$.
\end{defi}

\begin{rema}
Note that the objectives we consider in this work only put constraints on the states of the plays (and not on the actions).
\end{rema}

\begin{defi}[Winning strategy]\label{def:win_strat}
Let $G$ be a game structure and $W_i$ be an objective for player~$i$.
A strategy $\sigma_i$ of player~$i$ is \emph{winning} for $W_i$, if for every profile $\bar{\sigma}_{-i}$, the outcome $\rho$ of the profile $(\sigma_i,\bar{\sigma}_{-i})$ is winning for player~$i$, that is, $\rho \in W_i$.
\end{defi}

This definition focuses on $W_i$ only. 
A winning strategy $\sigma_i$ ensures that player~$i$ satisfies his objective $W_i$ against any strategy profile $\bar{\sigma}_{-i}$ of the other players.
In particular, it ensures that player~$i$ wins even if the other players are \emph{strictly antagonistic}, or \emph{adversarial}, that is, their objective is to make player~$i$ lose.
This context corresponds to the classical \emph{zero-sum} setting, and we use notation $\mathcal{G} = (G,W_i)$ for the game, since the objective $W_j$, for all players $j \ne i$, is implicitly defined as $\Plays{G} \setminus W_i$.

\begin{figure}
        \centering
    	\begin{tikzpicture}[->,>=latex,shorten >=1pt, initial text={}, scale=0.65, every node/.style={scale=0.65}]
    	\node[initial above, state, rectangle] (a) at (0, 0) {$v_0$};
    	\node[state] (b) at (-2, -2) {$v_1$};
    	\node[state] (c) at (2, -2) {$v_2$};
    	\node[state, double] (d) at (-3, -4) {$v_3$};
    	\node[state] (e) at (-1, -4) {$v_4$};
    	\node[state] (f) at (1, -4) {$v_5$};
    	\node[state] (g) at (3, -4) {$v_6$};
    	\path[->, blue] (a) edge node[above left] {$\ell$} (b);
    	\path[->, red] (a) edge node[above right] {$r$} (c);
    	\path[->, red] (b) edge node[above right] {$r'$} (e);
    	\path[->, blue] (b) edge node[above left] {$\ell'$} (d);
    	\path[->, blue] (c) edge node[above left] {$\ell'$} (f);
    	\path[->, red] (c) edge node[above right] {$r'$} (g);
    	\path (d) edge [loop below] node[below] {$\ell'$} (d);
    	\path (e) edge [loop below] node[below] {$\ell'$} (e);
    	\path (f) edge [loop below] node[below] {$\ell'$} (f);
    	\path (g) edge [loop below] node[below] {$\ell'$} (g);
    	\end{tikzpicture}
    	\caption{Two NEs and one SPE}
    	\label{fig_ne_spe}
\end{figure}
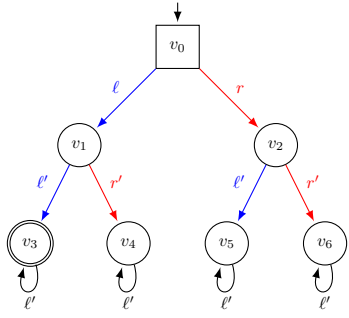

\begin{example}[A simple reachability game] 
Consider the game structure $G$ pictured in Figure~\ref{fig_ne_spe}. 
Its initial state is $v_0$ and there are two players, player $0$, who owns the circle states, and player $1$, who owns the square state. 
Transitions are represented by the arrows between states, such that there exists an arrow from a state $v$ to another state $v'$ if there exists an action $a$ (either left or right here) such that $\delta(v, a) = v'$. The set of actions $A$ is thus partitioned into $A_0 = \{\ell',r'\}$ and $A_1 = \{\ell,r\}$.

Both players have the same reachability objective, with $T_0 = T_1 = \lbrace v_3 \rbrace$, denoted in the figure by the fact that state $v_3$ is double-circled.
None of the two players have a winning strategy: player $0$ can prevent player~$1$ to ever reach state $v_3$ by choosing to go to $v_4$, $v_5$ or $v_6$, depending on the first action of player~$1$, while player~$1$ can prevent player $0$ from winning by going to state $v_2$.
\end{example}

\subsection{Solution Concepts and Studied Problem}

Within the classical \emph{zero-sum} context, the \emph{synthesis problem} asks if there exists, for a particular player, a winning strategy (see~\cite{DBLP:reference/mc/BloemCJ18} for an introduction).
As recalled in the preliminaries (cf.\ Definition~\ref{def:win_strat}), such a strategy ensures that his player wins against all possible strategies of the other players.
However, when we depart from this adversarial hypothesis and consider games that are \emph{non zero-sum}, that is, where each player has his own objective, which may overlap with the objectives of the others, this solution concept shows its limits.
Hence the call to solution concepts such as \emph{Nash equilibrium} or \emph{subgame-perfect equilibrium}, on which we focus in this work.

\begin{defi}[Nash equilibrium~\cite{Osborne94}] 
Let $\mathcal{G} = (G,(W_i)_{i \in \Pi})$ be a game and $\bar{\sigma} = (\sigma_i)_{i \in \Pi}$ be a strategy profile. We say that $\bar{\sigma}$ is a \emph{Nash equilibrium} (NE for short) if for every player $i \in \Pi$ and strategy $\sigma'_i \in \Sigma_i$, we have $\langle \sigma'_i, \bar{\sigma}_{-i} \rangle_{v_0} \in W_i \implies \langle \sigma_i, \bar{\sigma}_{-i} \rangle_{v_0} \in W_i$.
\end{defi}

In other words, no player has an incentive to \emph{deviate} unilaterally from its fixed strategy $\sigma_i$ in the profile $\bar{\sigma}$: if he does so, the resulting outcome will not satisfy his objective if the outcome of $\bar{\sigma}$  was not already doing so.

\begin{example}[A game with two NEs]
Let us come back to the reachability game of Figure~\ref{fig_ne_spe}. Two strategy profiles are represented in Figure~\ref{fig_ne_spe}: the first one, $\bar \sigma$, by the red transitions between states and the second one, $ \bar \sigma'$, by the blue ones.
Both profiles are NEs.
Indeed, for the red profile with outcome $v_0 r v_2 r'(v_6 \ell')^\omega$, if player $0$ deviates from $\sigma_0$, the resulting play $v_0 r v_2 \ell'(v_5 \ell')^\omega$ is still losing for him, since player~$1$ chooses to go to state $v_2$. On the other hand, if player~$1$ deviates from $\sigma_1$, the resulting play is $v_0 \ell v_1 r' (v_4 \ell')^\omega$, since player~$0$ chooses to go to state $v_4$, and this play is losing for player~$1$.
Thus, none of the players have a profitable deviation from ${\bar\sigma}$. Notice that the gain profile of the red profile is equal to $(0,0)$.
Similarly, one can easily check that the blue profile is also an NE with a gain profile equal to $(1,1)$.
\end{example}

In the red profile $\bar \sigma$ of the previous example, the possible choice of going from $v_1$ to $v_4$ for player~$0$ is irrational: after all, his objective is to reach state $v_3$. 
Notice, however, that this behavior is part of the NE $\bar \sigma$. 
These completely adversarial choices of actions are called \emph{non-credible threats}, and can appear once a deviation has already occurred in the outcome of the NE. 
In other words, NEs allow players to forget about their own objectives once the equilibrium outcome has been left.
These non-credible threats are one important limitation of NEs as a solution concept that captures rationality of the players in sequential games.
In order to enforce rationality in every scenario, even the ones that stem from deviations, one needs to monitor what the strategy profile prescribes after \emph{every} history.
This is exactly what \emph{subgame perfect equilibria}  (defined below) do~\cite{Osborne94}.

Let $\mathcal{G} = (G, (W_i)_{i \in\Pi})$ be a game. To each history $hv \in \Hist{G}$, with $v \in V$, corresponds a \emph{subgame} $\mathcal{G}_{\upharpoonright hv}$ that is the game $\mathcal{G}$ starting \emph{after the history $hv$}: its plays $\rho'$ start at the initial state $v$ and for all $i \in \Pi$, $\rho'$ is winning for player~$i$ if, and only if, $h\rho' \in W_i$. Given a strategy $\sigma_i$ for player~$i$, this strategy in $\mathcal{G}_{\upharpoonright hv}$ is defined as $\sigma_{i,\upharpoonright hv}(h') = \sigma_i(hh')$ for all histories $h' \in \Hist{G}$ starting with the initial vertex $v$. Given a strategy profile $\bar\sigma$, we denote by $\bar\sigma_{\upharpoonright hv}$ the profile $(\sigma_{i,\upharpoonright hv})_{i\in \Pi}$ (note that its outcome starts in~$v$). 

\begin{defi}[Subgame perfect equilibrium]
Let $\mathcal{G} = (G, (W_i)_{i \in\Pi})$ be a game. A \emph{subgame perfect equilibrium} (SPE for short) is a profile $\bar{\sigma}$ of strategies such that $\bar{\sigma}_{\upharpoonright hv}$ is an NE in each subgame $\mathcal{G}_{\upharpoonright hv}$ of $G$.
\end{defi}

\begin{example}\label{ex:diff_NE_SPE}
Consider again the game in Figure~\ref{fig_ne_spe}.
In the subgame $\mathcal{G}_{\upharpoonright v_0 \ell v_1}$, starting after the history $v_0 \ell v_1$, the red profile is not an NE.
Indeed, choosing to go to state  $v_4$ is a non-credible threat from player~$0$, as his target set $\{v_3\}$ is accessible from state $v_1$, which he owns. 
The player~$0$ strategy that goes to $v_3$ from $v_1$ is thus a profitable deviation in the subgame $\mathcal{G}_{\upharpoonright v_0 \ell v_1}$.
Therefore, the red profile is not an SPE in $\mathcal{G}$.
On the other hand, one can check that the blue profile is an SPE in $\mathcal{G}$, as it is an NE in every of its subgames.
\end{example}

In this work, we focus on \emph{synthesizing} a good strategy for a specific player (see Definition~\ref{def:SPE-NCRSp}), that we assume to be trustworthy, thus not prone to deviate from the strategy we prescribe.
As stated in the beginning of this section, we consider the \emph{non zero-sum} setting, where each player has his own objective, which may overlap with the objective of some of the other players. 
Hence, we want a solution concept that encompasses both the fact that the specific 
player will not deviate, and that other players may have behaviors that are \emph{adversarial} to the specific player, as long as it does not compromise their respective objectives, as they are assumed to be rational~\cite{DBLP:journals/amai/KupfermanPV16}.
We next define \emph{$0$-fixed} equilibria, which fit the requirements.

\begin{defi}[0-Fixed NE]
Let $\mathcal G$ be a game. A strategy profile $\bar{\sigma} = (\sigma_0, \bar{\sigma}_{-0})$ is a \emph{$0$-fixed NE}, if for every player~$i$, where $i \neq 0$, and every strategy $\sigma'_i \in \Sigma_i$, we have $\langle \sigma'_i, \bar{\sigma}_{-i} \rangle_{v_0} \in W_i \implies \langle \sigma_i, \bar{\sigma}_{-i} \rangle_{v_0} \in W_i$.
\end{defi}

In other words, if the strategy $\sigma_0$ of player~$0$ is fixed, that is, if we assume player~$0$ \emph{will stick} to his strategy, no other player has an incentive to deviate unilaterally. In this paper, we focus on 0-fixed SPEs: profiles $\bar{\sigma}$ of strategies that are 0-fixed NE in each subgame of $\mathcal G$ compatible with the fixed strategy $\sigma_0$ of player~$0$. Formally:

\begin{defi}[0-Fixed SPE]
Let $\mathcal G$ be a game. A strategy profile $\bar{\sigma} = (\sigma_0, \bar{\sigma}_{-0})$ is a \emph{$0$-fixed SPE} if for every history $hv$ compatible with $\sigma_0$, the profile $\bar{\sigma}_{\upharpoonright hv}$ is a $\sigma_0$-fixed NE is the subgame  $\mathcal{G}_{\upharpoonright hv}$.
\end{defi}

We also write that a strategy profile is a $\sigma_0$-fixed SPE to make clear that the profile is a $0$-fixed SPE where player $0$ has strategy $\sigma_0$. 
Furthermore, given a strategy $\sigma_0$, we say that the strategy profile $\bar{\sigma}_{-0}$ is a \emph{subgame-perfect response} to $\sigma_0$ if the complete profile $(\sigma_0, \bar{\sigma}_{-0})$ is a $0$-fixed SPE.
The next theorem guarantees the existence of an SPE in every reachability or parity game. This result also holds for $0$-fixed SPEs{~\cite{DBLP:conf/fsttcs/Ummels06}.

\begin{theo}[Existence of (0-fixed) SPEs]\label{th:SPE-existence}
Given a game $\mathcal{G} = (G, (W_i)_{i \in \Pi})$ that is a parity game or a reachability game, 
\begin{itemize}
    \item there always exists an SPE in $\mathcal G$,
    \item for every strategy $\sigma_0$ of player~$0$, there always exists a $\sigma_0$-fixed SPE in $\mathcal{G}$.
\end{itemize}
\end{theo}

In this paper, we study the following synthesis problem. Intuitively, the \emph{non-cooperative rational synthesis problem} asks, for a distinguished player, if there exists a strategy that is \emph{as close as possible} to a winning strategy, taking into account the fact that the others are behaving rationally.
The solutions seeked in this problem are indeed strategies that ensure winning for \emph{every} possible way to complete the profile in a rational manner (i.e., other players do not cooperate but take care of their own objectives first).

\begin{defi}[Non-cooperative rational synthesis problem]\label{def:SPE-NCRSp}
The \emph{non-cooperative rational synthesis problem} (\NCRS problem for short), asks, given a game $\mathcal{G} = (G, (W_i)_{i \in \Pi}) $, if there exists a strategy $\sigma_0$ of player $0$, such that, for every subgame-perfect response $\bar{\sigma}_{-0}$, the resulting outcome $\langle (\sigma_0, \bar{\sigma}_{-0}) \rangle_{v_0} $ is winning for player~$0$. 
\end{defi}

\begin{example}
The answer to the \NCRS problem for the game in Figure~\ref{fig_ne_spe} is positive. Indeed, consider the strategy $\sigma_0$ of player~$0$ that chooses action $\ell'$ in both states $v_1$ and $v_2$. The unique subgame-perfect response $\bar{\sigma}_{-0}$ for player~$1$ is to choose action $\ell$ in $v_0$. The resulting strategy profile is thus the blue one whose outcome is winning for player~$0$. 
\end{example}

We now state our main result on parity games. 

\begin{theo}[Complexity of the \NCRS problem for parity games]\label{thm:main}
\begin{enumerate}
\item The \NCRS problem for parity games is in 2ExpTime and PSpace-hard. Given ${\mathcal G} = (G,(\Parity(\alpha_i))_{i\in\Pi})$, this problem is solvable in time exponential in $|V|$ 
and each $|\alpha_i|$, $i \in \Pi$, and double-exponential in $|\Pi|$.

\item If the number of players is fixed, the problem is in ExpTime, and it is NP-hard and co-NP-hard.
\end{enumerate}
\end{theo}

When the number of players is fixed, the \NCRS problem becomes polynomial in the particular case of reachability games.




\begin{theo}[Complexity of the \NCRS problem for reachability games]\label{thm:mainReach}
If the number of players is fixed, the \NCRS problem is solvable in time polynomial in $|V|$ 
for reachability games. 
\end{theo}

\section{\NCRS Problem and \texorpdfstring{\pcp} Game }\label{sec:SPE-NCRS_PCP}

\subsection{The approach}\label{sec:description_SPE-NCRS}
In order to answer the question asked by the \NCRS problem, one may notice the particular role \emph{player~$0$} has in the setting of the problem: one needs to find, if it exists, a strategy $\sigma_0$ of player $0$ that ensures him to win \emph{for all} possible partial profiles $\bar{\sigma}_{-0}$ such that $(\sigma_0,\bar{\sigma}_{-0})$ is a $0$-fixed SPE. 
Our approach is to construct, in the spirit of \cite{Meunier16}, a zero-sum two-player \emph{Prover-Challenger} game that will, once solved, deliver the answer. 
The key idea in this approach is that Prover is set to prove that there exists a solution $\sigma_0$ to the \NCRS problem, while Challenger wants to disprove this claim. 
However, the special role of player $0$ in the original game $\mathcal G$ must be taken into account. 
To this end, we \emph{split} Prover into a coalition of two Provers, \emph{Prover $1$} and \emph{Prover $2$} (we later show, in Section~\ref{sec:SolvingPCPGame}, how to come back to the two-player model). 
The goal of this coalition of Provers is still to show that there exists a solution to the \NCRS problem.
The efforts towards this goal are distributed between the two as follows:
The role of Prover~$1$ (noted $P_1$) is to exhibit the solution, a strategy $\sigma_0$ of player $0$ in the game $\mathcal G$.
The role of Prover $2$ (noted $P_2$) is, roughly, to show that for any profile $\bar{\sigma}_{-0}$, the complete profile $(\sigma_0,\bar{\sigma}_{-0})$ is either \emph{not} a $0$-fixed SPE, or its outcome is \emph{winning} for player~$0$.
The objective of Challenger (noted $C$) is to show that there does not exist a solution to the \NCRS problem, in other words, to prove that, for each strategy $\sigma_0$ of player~$0$, there exists a subgame perfect response $\bar{\sigma}_{-0}$ such that the outcome of $(\sigma_0,\bar{\sigma}_{-0})$ is \emph{losing} for player~$0$. 
This game is called \emph{\pcp game} and is denoted by \pcpG.

Before proceeding to the formal definitions, let us give some intuition about this \pcp game for the fixed game $\mathcal G$. 
The three players $P_1$, $C$, and $P_2$, proceed to construct, step by step and with some additional interactions, a play in \pcpG, that \emph{simulates} a play in $\mathcal G$ (see Definition~\ref{def:sim_play_in_G} below). Each player has a specific part to play in the construction of this simulating play. As mentioned previously, player~$P_1$ has the particular task to exhibit a candidate solution strategy $\sigma_0$ to the \NCRS problem for $\mathcal G$.
This is done by simulating player $0$ playing $\sigma_0$ and not deviating from it.
Thus, along the play in \pcpG being constructed, whenever the corresponding state in the simulated play of $\mathcal G$ belongs to player $0$, it is then up to $P_1$ to choose the next state.
One of the tasks of player~$C$, is to try to complete $\sigma_0$ with a subgame perfect response $\bar{\sigma}_{-0}$ from the other players in $\mathcal G$.
Thus, whenever the corresponding state in the simulated play of $\mathcal G$ belongs to a player~$i \neq 0$, it is player~$C$'s turn to play in \pcpG: he \emph{proposes} to play an action according to some strategy $\sigma_i$ of player~$i$ in $\mathcal G$ and in addition \emph{predicts} the gain profile of what will be the outcome of the profile $(\sigma_0,\bar\sigma_{-0})$, that is, of the simulated play being constructed. However, since players other than player $0$ in $\mathcal G$ can deviate, whenever $C$ has made a proposal for a next state, he has to let player~$P_2$ have a say: $P_2$ can either accept this proposal or refuse it and \emph{deviate} on behalf of player~$i$ by choosing another state as the next one in the simulating play. This phase is called the \emph{decision phase}. If this phase results in a deviation, the game proceeds to the \emph{adjusting} phase: since the current subgame has changed, Challenger has to predict the gain profile of $(\sigma_0,\bar\sigma_{-0})$ in this subgame, that shows the deviation of player~$i$ was not a profitable one.
Then, the play in \pcpG resumes to the construction phase again.

How do we ensure that $\sigma_0$ is fixed? In other words, what restricts $P_1$ from adapting his strategy to the way $C$ and $P_2$ interact, but only according to the so far constructed simulated history in $G$?
The solution we choose is to make use of \emph{imperfect information}: 
in the \pcp game, $P_1$ cannot, in fact, distinguish between all the states. 
To model the partial view $P_1$ has on the states of the game, we speak of \emph{observations}. 
Informally, to each state (and action) is associated an observation, and  $P_1$ has to take action upon sequences of observations only (that we call \emph{observed histories}).
In particular, it is assumed that the strategies should \emph{respect} the observations, in the sense that two different histories yielding the same sequence of observations should trigger the same action of $P_1$. 
Thus, we only consider these \emph{observation-based strategies} as possible behaviors of $P_1$.

We now proceed to the formal definition of the \pcp game. In Section~\ref{sec:SolvingPCPGame}, we show that solving the \NCRS problem is equivalent to solving this game in a sense stated in Theorem~\ref{thm:equivalence}.

\subsection{Observation Functions} \label{subset:Obs}

We begin by introducing the concepts of observation and of observation-based strategy.

\begin{defi}[Observation function]\label{def:observations}
Let $G =  (V, A, \Pi, \delta, v_0)$ be a game structure and player~$i$ be some player in $G$. 
Let $\mathcal{O}$ be a partition of $V$, that defines an \emph{observation function} $\Obs: V \rightarrow \mathcal{O}$ for player~$i$, that is, for every state $v$, we have $\Obs(v) = o \in \mathcal{O}$.
Let then $\bar{\mathcal{O}}$ be a partition of $A$, that extends the previous function $\Obs$ to actions, that is, for every action $a$, we have $\Obs(a) = \bar{o} \in \bar{\mathcal{O}}$. 
\end{defi}

The function $\Obs$ extends to histories and plays in the straightforward way. We say that player~$i$ \emph{observes G through $\Obs$} if he cannot distinguish between states (resp. actions, histories, plays) that yield the same observation via the function $\Obs$. 

For those states or actions that are the unique element of their observation set, we say that they are \emph{visible}.
We say that $\Obs$ is the \emph{identity function} if $\Obs(v) = \{v\}$ for all $v \in V$ and $\Obs(a) = \{a\}$ for all $a \in A$.
Whenever player~$i$ observes the game structure through the identity function, we say that player~$i$ has \emph{perfect information}. If it is through any other function, we say that he has \emph{imperfect information}.
We will always assume that a player \emph{sees his own actions perfectly} (that is, his observation function restricted to the domain of his action set corresponds to the identity function).
Furthermore, we assume all players have \emph{perfect recall}: they remember the full sequence of observations they witnessed from the start of the play.

\begin{example}\label{ex:info_imparfaite}
In Figure~\ref{fig_ex_imperfect_info}, a two-player game structure is pictured, together with the observations of player~$0$, who owns the circle states (we suppose that player~$1$ has perfect information).
Each state is divided in two sections: on the left, with a white background, the name of the state is displayed, while on the right side, with a gray background, the observation of player~$0$ is displayed. 
Similarly, the actions $\ell$ and $r$ of player~$1$, who owns the square state, are accompanied by a $\#$ on their right side, which is an abuse of notation to mean $\lbrace \ell, r \rbrace$, that is, player~$0$ cannot distinguish between the two possible actions of player~$1$.
Notice that from states owned by player~$0$, the actions do not have an observation attached to them, as player~$0$ knows his own actions, by hypothesis.

Let us now look at which histories player~$0$ can or cannot distinguish.
From $v_0$, player~$0$ cannot distinguish between the actions of player~$1$, and both states $v_1$ and $v_2$ have the same observation set $\lbrace v_1, v_2 \rbrace$. Therefore, the two histories $v_0 \ell v_1$ and $v_0 r v_2$ are indistinguishable for player~$0$, i.e., $\Obs(v_0 \ell v_1) = \Obs(v_0 r v_2)$.
Similarly, we have that $\Obs(v_0 \ell v_1 \ell' v_3) = \Obs(v_0 r v_2 \ell' v_5)$.
\end{example}

\begin{figure}
        \centering
    	\begin{tikzpicture}[->,>=latex,shorten >=1pt, initial text={}, scale=0.65, every node/.style={scale=0.65}]
    	\node[initial above, state2Pr] (a) at (0, 0) {$v_0$ \nodepart{two}$\lbrace v_0 \rbrace$};
    	\node[state2Prc] (b) at (-4, -2) {$v_1$ \nodepart{two} $\lbrace v_1, v_2 \rbrace$};
    	\node[state2Prc] (c) at (4, -2) {$v_2$ \nodepart{two} $\lbrace v_1, v_2 \rbrace$};
    	\node[state2Prc, double] (d) at (-6, -4) {$v_3$ \nodepart{two} $\lbrace v_3, v_5 \rbrace$};
    	\node[state2Prc] (e) at (-2, -4) {$v_4$ \nodepart{two} $\lbrace v_4, v_6 \rbrace$};
    	\node[state2Prc] (f) at (2, -4) {$v_5$ \nodepart{two} $\lbrace v_3, v_5 \rbrace$};
    	\node[state2Prc] (g) at (6, -4) {$v_6$ \nodepart{two} $\lbrace v_4, v_6 \rbrace$};
    	\path[->, blue] (a) edge node[above left] {$\ell~\vert~\sharp$} (b);
    	\path[->, red] (a) edge node[above right] {$r~\vert~\sharp $} (c);
    	\path[->, red] (b) edge node[above right] {$r'$} (e);
    	\path[->, blue] (b) edge node[above left] {$\ell'$} (d);
    	\path[->, blue] (c) edge node[above left] {$\ell'$} (f);
    	\path[->, red] (c) edge node[above right] {$r'$} (g);
     
    	\path (d) edge [loop below] node[below] {$\ell'$} (d);
    	\path (e) edge [loop below] node[below] {$\ell'$} (e);
    	\path (f) edge [loop below] node[below] {$\ell'$} (f);
    	\path (g) edge [loop below] node[below] {$\ell'$} (g);
     
    	\end{tikzpicture}
    	\caption{Imperfect information in a reachability game}
    	\label{fig_ex_imperfect_info}
\end{figure}
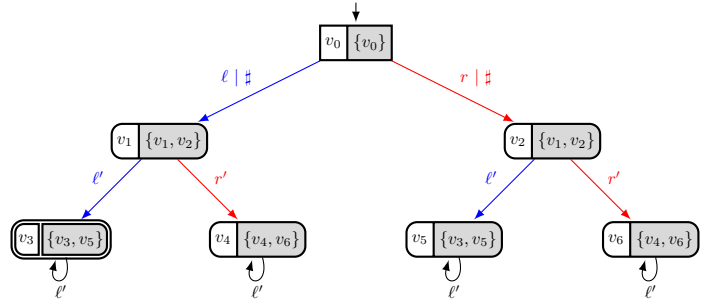

In the sequel, we only consider \emph{\pstable} observation functions such that the last states of the histories observed similarly all belong to the same player.

\begin{defi}[\Pstable observation function] \label{def:pstable}
An observation function $\Obs$ for player~$i$ is \emph{\pstable} if for every two histories $h = v_0a_0v_1a_1 \ldots v_k$ and $g = u_0b_0u_1b_1 \ldots u_k$ such that $\Obs(h) = \Obs(g)$, then the states $v_\ell,u_\ell$ are controlled by the same player, for all $\ell \in \{0,\ldots, k\}$. 
\end{defi}

The identity observation function is trivially \pstable. This is also the case for the observation function of Example~\ref{ex:info_imparfaite} due to the form of its game structure. A game $\mathcal G$ with objectives $W_i$ for all its players $i \in \Pi$ and a \pstable observation function $\Obs$ for one of its players is denoted $\mathcal{G} = (G,(W_i)_{i \in \Pi},\Obs)$. When $\Obs$ is the identity observation function, we do not mention it in this notation.

Notice that with a \pstable observation function for player~$i$, we have that, given a history $h \in \Histi{G}{i}$, all histories $h'$ with the same observation as $h$ also belong to $\Histi{G}{i}$. It is therefore natural to ask that a strategy $\sigma_i$ for player~$i$ is \emph{observation-based}, that is, the same action is played by player~$i$ after all histories observed similarly.

\begin{defi}[Observation-based strategy]\label{def:obs-based-strat}
Given a \pstable observation function $\Obs$, a strategy $\sigma_i$ of player~$i$ is \emph{observation-based} if for every pair of histories $h, h'$, if $h \in \Histi{G}{i}$ and $\Obs(h) = \Obs(h')$, then $\sigma_i(h)=\sigma_i(h')$. 
\end{defi}

\begin{example}\label{ex:info_imparfaite_2}
Let us come back to the game structure of Figure~\ref{fig_ex_imperfect_info}. Suppose that the objective of player~$0$ is to reach the target set $\lbrace v_3 \rbrace$. One can verify that player~$0$ does not have an observation-based winning strategy.
Indeed, as noted in Example~\ref{ex:info_imparfaite}, player~$0$ cannot distinguish between histories $h_1 = v_0 \ell v_1$ and $h_2 = v_0 r v_2$.
By definition, any observation-based strategy must prescribe the same action from both $h_1$ and $h_2$.
If this action is $r'$, then both resulting plays are losing for player~$0$.
If it is $\ell'$, then one of the resulting play is winning, while the other is losing.
Thus, no observation-based strategy ensures player~$0$ to reach $\{v_3\}$.
\end{example}

\subsection{Definition of the \texorpdfstring{\pcp} Game} \label{subsec:DefPCPgame}

In this section, let us fix a game $\mathcal{G} = (G,(W_i)_{i \in \Pi})$ with a game structure $G = (V, A, \Pi, \delta, v_0)$ and objectives $(W_i)_{i \in \Pi}$. Later, this game will be a reachability game or a parity game.
We here formally define the corresponding \pcp game, that was informally presented in Section~\ref{sec:description_SPE-NCRS}.
Recall that the role of player~$P_1$ is to simulate a strategy $\sigma_0$ of player $0$ in $\mathcal G$, while player~$C$ has to propose a subgame perfect response to $\sigma_0$ such that $(\sigma_0,\bar{\sigma}_{-0})$ is losing for player~$0$ in $\mathcal G$}, response to which player~$P_2$ can react by choosing some deviations (decision phase). The interactions between the three players result in a play in \pcpG that simulates a play in $\mathcal G$. Challenger has to predict a gain profile for this simulated play, which he can update upon deviations from~$P_2$ (adjusting phase). 

We begin by defining the game structure of the new game \pcpG, in Definition~\ref{def:pcp-game-structure-states} for its states and in Definition~\ref{def:pcp-game-structure-transitions} for its transition function.\footnote{Note that to avoid confusion with the players of $G$, and without loss of generality, the three players of the \pcp game have been given explicit names rather than numbers.} 
We then define the observation function of $P_1$ in Definition~\ref{def:Obs-of-the-PCP-game}; and we finally define the objectives of the three players in Definition~\ref{def:objectivesPCP}. Figure~\ref{fig:prover-game-play} should help the reader for understanding those definitions.

The states (except its initial state $v'_0$) of the \pcp game structure have several components:
\begin{itemize}
    \item all of them have a state of $G$ as their first component, that we call \emph{$G$-component},
    \item similarly, all of them have a \emph{gain-component}, which consists of a gain profile $\bar g \in \{0,1\}^{|\Pi|}$ for the players of $G$, where in each component $g_i$ of $\bar g$, $0$ symbolizes a loss and $1$ a win with respect to objective $W_i$.
\end{itemize} 

The states that have only these two components are called \emph{$G$-states} and belong either to $P_1$ or $C$.
A projection on the $G$-components of these states determines the simulated play in $G$ (see Definition~\ref{def:sim_play_in_G}).

As already stated in Section~\ref{sec:description_SPE-NCRS}, whenever the current state of the simulated play belongs to player $0$ in $G$, the next state of the simulated play is chosen directly by $P_1$.
However, when it is not the case, $C$ has to make a proposition, which should then be validated or changed by $P_2$ (decision phase).
Thus, to reflect this proposal while remaining hidden to $P_1$, the successor states of $G$-states belonging to $C$ have an extra component, called \emph{action-component}, which records the action of $G$ proposed by $C$.
From such states called \emph{action-states}, $P_2$ has to react by confirming or changing the action that was proposed by $C$. This is modeled by a third kind of states, called \emph{$player$-states}.
Such states have a \emph{player-component}, which consists of a player from $G$ different from player~$0$ to signal to $C$ this player has indeed deviated  (or the empty set, to signal acceptance by Prover~$2$ of Challenger's proposal).

\begin{defi}[States of the \pcp game] \label{def:pcp-game-structure-states}
Given a game $\mathcal G$, the game structure of \pcpG is a game structure $G' = (V', A', \lbrace P_1, C, P_2 \rbrace, \delta', v'_0)$, where $P_1$, $C$, and $P_2$ are the three players and the set of states $V' = \VPun \cup \VC \cup \VPd$ is as follows:\footnote{The action set and the transition function are defined in Definition~\ref{def:pcp-game-structure-transitions}.}
\begin{itemize}
    \item $\VPun = \lbrace (v, \bar{g})~\vert~ v \in V_0,~ \bar{g}\in \lbrace 0, 1 \rbrace^{\vert \Pi \vert} \rbrace$,
    \item $\VC  = \lbrace v'_0 \rbrace \cup \lbrace (v, \bar{g})~\vert~ v \in V \setminus V_0,~ \bar{g}\in \lbrace 0, 1 \rbrace^{\vert \Pi \vert} \rbrace \cup \lbrace (v,i, \bar{g})~\vert~ v \in V,~ i\in \Pi \setminus \lbrace 0 \rbrace \cup \lbrace \varnothing \rbrace,~ \bar{g}\in \lbrace 0, 1 \rbrace^{\vert \Pi \vert} \rbrace$,
    \item $\VPd = \lbrace (v, a, \bar{g}) ~\vert~v \in V \setminus V_0, ~ a \in A \setminus A_0 \text{ and } \delta(v, a) \text{ is defined},~ \bar{g}\in \lbrace 0, 1 \rbrace^{\vert \Pi \vert} \}$,
    \item $v'_0$ is the initial state.
\end{itemize}
Among those states, $(v,\bar{g})$ are $G$-states, $(v,a,\bar{g})$ are action-states, and $(v,i,\bar{g})$ are player-states. Moreover, $v$ is a $G$-component, $\bar{g}$ is a gain-component, $a$ is an action-component, and $i$ is a player-component. Notice that the sets $\VPun, \VC, \VPd$ are pairwise disjoint.
\end{defi}

After the definition of the states of the \pcp game structure, let us proceed to the definition of its transitions. We already began to explain some actions of the players. Let us continue these explanations, see also Figure~\ref{fig:prover-game-play}. Recall that the goal of Challenger is to prove the existence of a $\sigma_0$-fixed SPE losing for player $0$ in $\mathcal G$. 
To be able to verify this claim in \pcpG, Challenger has to predict the gain of the simulated play that is being constructed.
To do so, the gain-components of the states of \pcpG are used as follows:
\begin{itemize}
    \item first, Challenger owns the initial state $v'_0$, and has to choose a gain profile $\bar{g}$, losing for player $0$\footnote{As an outcome of a subgame perfect response to $\sigma_0$ losing for player $0$ is exactly what Challenger wants to exhibit, for any $\sigma_0$ simulated by $P_1$.}, to start the construction of the simulated play in the $G$-state $(v_0, \bar{g})$,
    \item then, whenever $P_2$ chooses to make some player~$i$ deviate, before reaching a $G$-state,  Challenger has to respond by choosing an adjusted gain profile with a lower or equal gain for this player~$i$ (adjusting phase). This new gain expresses to $P_2$ the absence of any profitable deviation for player~$i$.
\end{itemize}

Note that in the following definition of the \pcp game structure, two modeling decisions have been made to help handle the subsequent developments. 
First, to ensure a certain regularity of the plays' shape in the \pcp game and that $P_1$ cannot infer whether $P_2$ actually made a deviation, the adjusting phase is played regardless of a deviation occurring or not: if there was no deviation, Challenger has no choice but to play the same gain profile.
Second, to avoid confusion on who, among Challenger and Prover~$2$, is currently playing some action $a \in A \setminus A_0$, this action is of the form $(a, i)$ for $C$ and $a$ for $A$, that is, we additionally recall the player~$i$ from $G$ performing the action $a$ when Challenger is playing.

\begin{defi}[Transitions of the \pcp game] \label{def:pcp-game-structure-transitions}
The set of actions $A' = \APun \cup \AC \cup \APd$ of the \pcpG game structure is as follows:
\begin{itemize}
    \item $\APun = A_0$,
    \item $\AC  = (A \setminus A_0 \times \Pi \setminus \lbrace 0 \rbrace) \cup \lbrace 0, 1 \rbrace^{\vert \Pi \vert}$,
    \item $\APd = A \setminus A_0$,
\end{itemize}
Notice that the sets $\APun, \AC, \APd$ are pairwise distinct.
The transition function $\delta'$ of the \pcpG game structure is defined as follows:
\begin{itemize}
\item For player~$P_1$:
\begin{itemize}
    \item for every state $(v, \bar{g}) \in \VPun$, action $a \in \APun$ and state $u \in V$ such that $\delta(v,a) = u$, we have $\delta'((v, \bar{g}), a) = (u, \bar{g})$,
\end{itemize}
\item For player~$C$: 
\begin{itemize}
    \item for every gain profile $\bar{g} \in \lbrace 0, 1 \rbrace^{\vert \Pi \vert}$ such that $g_0 = 0$, we have $\delta'(v'_0, \bar{g}) = (v_0, \bar{g})$,
    \item for every state $(v, \bar{g}) \in \VC$ with $v \in V_i$, and action $a \in A_i$, we have $\delta'((v, \bar{g}), (a, i)) =(v, a ,\bar{g})$,
    \item for every state $(v, i, \bar{g}) \in \VC$ with $i \neq \varnothing$, and gain $\bar{g}'$ such that $g'_i \leq g_i$, we have $\delta'((v, i, \bar{g}), \bar{g}') =(v, \bar{g}')$,
    \item for every state $(v, i, \bar{g}) \in \VC$ with $i = \varnothing$, we have $\delta'((v, i, \bar{g}), \bar{g}) =(v, \bar{g})$,
\end{itemize}
The last two cases are the \emph{adjusting phase}.
\item For player~$P_2$:
\begin{itemize}
    \item for every state $(v, a, \bar{g}) \in \VPd$ and state $u \in V$ such that $\delta(v, a) = u$, we have $\delta'((v, a, \bar{g}), a) = (u, \varnothing, \bar{g})$,
    \item for every state $(v, a, \bar{g}) \in \VPd$ with $v \in V_i$, action $b \in A_i$ such that $b \neq a$, and state $u \in V$ such that $\delta(v, b) = u$, we have $\delta'((v, a, \bar{g}), b) = (u, i, \bar{g})$.
\end{itemize}
These two cases are the \emph{decision phase}.
\end{itemize}
\end{defi} 

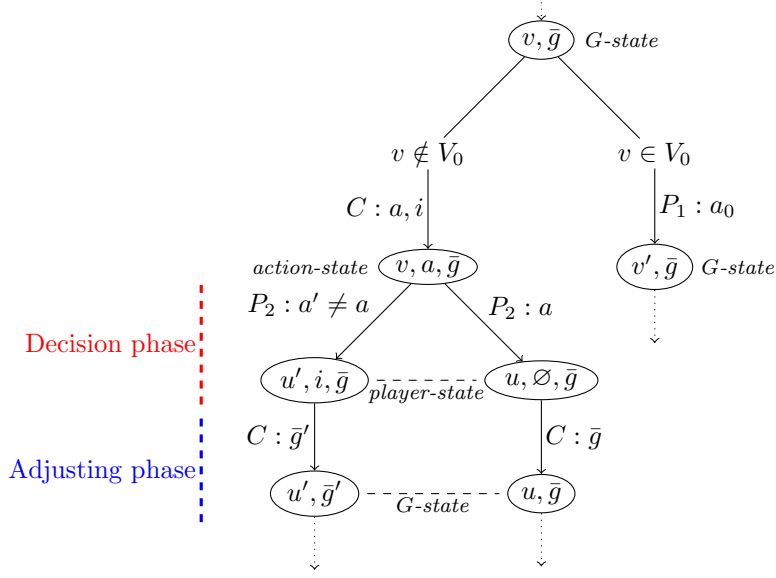
\begin{figure}
    \centering
\begin{tikzpicture}[->,on grid,node distance=1.5cm and 1.5cm,
    vertex/.style={ellipse, draw}, inner sep = 0.3ex]

  \node [vertex] (vg0) {$\vtx, \g$} node[right=+15pt]{\footnotesize \emph{$G$-state}};
  
    \node (vg0C) [below left=of vg0] {$\vtx \notin \VerticesOf{0}$};
      \node [vertex, inner sep = 0.2ex] (vvg) [below=of vg0C] {$v, a, \g$} ;
        \node [vertex] (vpg) [below left=of vvg] {$u', i, \g$};
          \node [vertex] (vg1) [below=of vpg] {$u', \g'$};
        \node [vertex] (vng) [below right=of vvg] {$u, \varnothing, \g$};
          \node [vertex] (vg2) [below=of vng] {$u, \g$};
    \node (vg0P) [below right=of vg0] {$\vtx \in \VerticesOf{0}$};
      \node [vertex] (vg3) [below=of vg0P] {$\vtx', \g$} ;

  \draw [dotted] (0,0.5)  --  (vg0.north) ;  
  \draw [dotted] (vg1.south)  --  ([shift=({0cm,-0.7cm})]vg1.south) ;
    \draw [dotted] (vg2.south)  --  ([shift=({0cm,-0.7cm})]vg2.south) ;
    \draw [dotted] (vg3.south)  --  ([shift=({0cm,-0.7cm})]vg3.south) ;
\draw node at ([shift=({-0.9cm,-0cm})]vvg.west) {\emph{\footnotesize action-state}};
\draw node at ([shift=({0.6cm,-0cm})]vg3.east) {\emph{\footnotesize $G$-state}};
\draw [dashed, very thick, red, - ] ([shift=({-3cm,-0cm})]vvg.south)  --  node[left] {Decision phase} ([shift=({-3cm,-1.6cm})]vvg.south) ;
 \draw [dashed, very thick, blue, - ] ([shift=({-1.5cm,-0.2cm})]vpg.south)  --  node[left] {Adjusting phase} ([shift=({-1.5cm,-1.6cm})]vpg.south) ;

  \path (vg0)  edge [-] (vg0C)
               edge [-] (vg0P)
        (vg0C) edge node[left]       {$\ch : a, i $}       (vvg)
        (vvg)  edge node[above left] {${\prov2} : a' \neq a$} (vpg)
               edge node[above right]{${\prov2} : a$}      (vng)
        (vpg)  edge node[left]       {$\ch : \g'$}            (vg1)
        (vng)  edge node[right]      {$\ch : \g$}           (vg2)
        (vg0P) edge node[right]      {${\prov1} : a_0$} (vg3);
        
        \draw [dashed, -, shorten <=3, shorten >=3 ] (vg1)   --  node[below] {\footnotesize \emph{$G$-state}}  (vg2) ; 
        \draw [dashed, -, shorten <=3, shorten >=3 ] (vpg)   --  node[below] {\footnotesize \emph{player-state}}  (vng) ;

\end{tikzpicture}
\caption{General structure of a play in \pcpG}
    \label{fig:prover-game-play}
\end{figure} 

As already mentioned in Section~\ref{sec:description_SPE-NCRS}, given a play in \pcpG, there exists a unique corresponding play in $G$ that is being simulated by the interactions of the Provers and Challenger.

\begin{defi}[Simulated play and gain in $G$]\label{def:sim_play_in_G}
 Let $\rho = v'_0a'_0v'_1a'_1v'_2a'_2\dots$ be a play in the game structure of \pcpG. The \emph{simulated play} of $\rho$ is the play $\simu{\rho} \in \Plays{G}$ being the projection of $\rho$ on the $G$-component of $G$-states (which belong either to $P_1$ or $C$) and on actions of $P_1$ and $P_2$.\footnote{This projection does not take into account the action- and player-states and the actions of Challenger.}
The definition extends naturally to histories.
The \emph{simulated gain} of $\rho$ is a Boolean vector $\gainSim{\rho} \in \lbrace 0, 1 \rbrace^{\vert \Pi \vert}$, such that 
$\gainSim{\rho}_i = 0$ if $\simu{\rho}$ is losing for player~$i$, and $\gainSim{\rho}_i = 1$ if $\simu{\rho}$ is winning for player~$i$.
\end{defi}

\begin{example}\label{ex:simu_play_gain} 
Consider the fictional history 
$$h = v'_0~\bar{g}~(v_0,\bar{g})~a_0~(v_1,\bar{g})~(a_1,i)~(v_1, a_1, \bar{g})~a'_1~(v_2, i, \bar{g})~\bar{g}'~(v_2, \bar{g}')$$
in the \pcp game of some game $\mathcal G$. 
It starts in the initial state $v'_0$, where $C$ chooses a gain profile $\bar{g}$ where $g_0 = 0$. Thus the second state of the history is the $G$-state $(v_0, \bar{g})$. 
Looking at the form of the third state $(v_1,\bar{g})$ in $h$, we can deduce that $(v_0, \bar{g})$ belongs to Prover~$1$, that is, $v_0$ belongs to player~$0$ in $G$. Moreover $\delta(v_0,a_0) = v_1$.
Then, the fourth state $(v_1, a_1, \bar{g})$ in $h$ is an action-state, which means that the previous state $(v_1,\bar{g})$ belongs to Challenger, such that $v_1$ belongs to player~$i \neq 0$ in $G$, and $a_1$ is the action from $G$ proposed by $C$ in this scenario. 
One can see that $P_2$ then chooses the action $a'_1$, which is different from $a_1$ since the fifth state $(v_2, i, \bar{g})$ of $h$ is a state whose player-component $i$ is not equal to $\varnothing$.
This indicates that $P_2$ makes player~$i$ deviate in $G$ and that $\delta(v_1, a'_1) = v_2$.
By now, Challenger has to choose a new gain profile $\bar{g}'$, yielding the last state $(v_2, \bar{g}')$ of $h$, which is a $G$-state.
By projecting $h$ on the $G$-components of its $G$-states and on the actions of the Provers, one can check that the simulated history of $h$ in $G$ is $\simu{h} = v_0a_0v_1a'_1v_2$.

Consider now the fictional play 
$$\rho = h ~\bigl((a_2,j)~(v_2, a_2, \bar{g}')~a_2~(v_2, \varnothing, \bar{g}')~\bar{g}'~(v_2, \bar{g}')\bigr)^\omega$$
in the \pcp game.
Looking at $\rho$, one can deduce that from $(v_2, \bar{g}')$, the last state of $h$, Challenger proposes action $a_2$ from $G$ for player~$j \neq 0$ and that $P_2$ accepts this proposal.
Furthermore, Challenger cannot adjust the gain profile, leading to state $(v_2, \bar{g}')$.
One can see that this behavior repeats indefinitely, thus the corresponding simulated play is $\simu{\rho} = \simu{h} (a_2 v_2)^\omega$. The simulated gain $\gainSim{\rho}$ of $\rho$ is a Boolean vector deduced from $\simu{\rho}$, such that its $i$-th component is equal to $1$ if, and only if, $\simu{\rho}$ is winning for  player~$i$.
\end{example}

\begin{rema}\label{rema:several_plays_in_pcp_for_one_in_G}
  Note that while there exists a unique simulated play $\simu{\rho}$ in $\mathcal{G}$ for every play $\rho$ in \pcpG, the converse does not hold. Indeed, several different sequences of interactions between $P_2$ and $C$ yield the same simulated play in $\mathcal{G}$. For instance, Challenger can propose different actions that $P_2$ can refuse. We will come back to the remark in Lemma~\ref{lem:tau_C_histG->histG'} below.
\end{rema}

We have seen that each state and action of $G$ appear in several contexts in the game structure $G'$ of \pcpG. Let us now state how the size of $G'$ depends on the size of $G$. The proof of this lemma directly follows from the definition of the \pcp game.

\begin{lemm} \label{lem:sizeGamePCP}
Given the size $|V|$, $|A|$, and $|\Pi|$ of the game structure of $\mathcal G$, the game structure of \pcpG has 3 players and a size 
\begin{itemize}
\item $|V'|$ linear in $|V|$ and $|A|$, and exponential in $|\Pi|$, 
\item $|A'|$ linear in $|A|$ and exponential in $|\Pi|$.
\end{itemize}
\end{lemm}

In the \pcp game, recall that players~$C$ and $P_2$ have perfect information whereas player~$P_1$ has imperfect information to ensure that $P_1$ cannot adapt his strategy from witnessing the interactions between $C$ and $P_2$. The observation function $\Obs$ for player~$P_1$ is defined below on $V'$ and $A'$. For each state $v'$ of $V'$, he is only able to observe the $G$-component of $v'$. Concerning the actions of $A'$, those of $\APun \cup \APd$ are all visible whereas those of $\AC$ are not visible at all. 
In the next definition, there is an abuse of notation in $\mathcal{O}$ and $\bar{\mathcal{O}}$, as already done in Example~\ref{ex:info_imparfaite}.

\begin{defi}[Information in the \pcp game] \label{def:Obs-of-the-PCP-game}
Given the game structure $G'$ of \pcpG, players~$C$ and $P_2$ have perfect information and the observation function $\Obs$ of player~$P_1$ is defined as follows.
\begin{itemize}
    \item We have $\Obs : V' \rightarrow \mathcal{O} = \{v'_0\} \cup \{v \mid v \in V\}$ such that:
    \begin{itemize}
        \item $\Obs(v,\bar{g}) = v$ for all $(v,\bar{g}) \in \VPun$,
        \item $\Obs(v'_0) =  v'_0$,
        \item $\Obs(v, \bar{g}) = v$ for all $(v, \bar{g}) \in \VC$,
        \item $\Obs(v,i,\bar{g}) = v$ for all $(v,i,\bar{g}) \in \VC$,
        \item $\Obs(v,a,\bar{g}) = v$ for all $(v,a,\bar{g}) \in \VPd$.
    \end{itemize}
    \item We have $\Obs : A' \rightarrow \bar{\mathcal{O}} = \{\sharp\} \cup \{a \mid a \in A\}$ such that:
    \begin{itemize}
        \item $\Obs(a') = a'$ for all $a' \in \APun$,
        \item $\Obs(a') =  \sharp$ for all $a' \in \AC$,
        \item $\Obs(a') = a'$ for all $a' \in \APd$.
    \end{itemize}
\end{itemize}
\end{defi}

To finalize the definition of the \pcp game, it remains to define the objectives of the three players. 
Recall that the objective of Challenger is to show that for each strategy $\sigma_0$ of player~$0$ in $\mathcal G$, there exists a subgame perfect response $\bar{\sigma}_{-0}$ such that the outcome of $(\sigma_0,\bar{\sigma}_{-0})$ is losing for player~$0$. This objective $W_C$ of Challenger is detailed below in Definition~\ref{def:objectivesPCP}. Let us first give some intuition on what is a play $\rho$ in the \pcp game that is winning for $C$. 
Three winning situations may occur for $C$ along $\rho$: 
\begin{itemize}
    \item[$(iC)$] eventually, that is, after reaching some subgame, $P_2$ always accepts the action proposals of $C$ and the gain predicted by $C$ in the subgame is correct (it is equal to the simulated gain in this subgame),
    \item[$(iiC)$] eventually, $P_2$ keeps making one unique player~$i$ deviate in the decision phase, but $C$ is able to adjust the gain to show that this deviation is not profitable for player~$i$,
    \item[$(iiiC)$] $P_2$ keeps making at least two different players deviate, essentially conceding the play, as the only deviations that can be considered within the scope of ($0$-fixed) SPEs are \emph{unilateral} deviations.
\end{itemize}

In the \pcp game, the two Provers have the same objective $W_P$ that is opposed to $W_C$. 
Indeed, recall that their objective is to exhibit a strategy $\sigma_0$ for player~$0$ in $G$, such that for every subgame-perfect response $\bar{\sigma}_{-0}$ to $\sigma_0$, the outcome of the resulting profile $(\sigma_0,\bar{\sigma}_{-0})$ is winning for player~$0$.
Let us give some intuition on what is a winning play for the Provers, see Definition~\ref{def:objectivesPCP} for the formal definition. 
Two winning situations may occur for the Provers along a play $\rho$ in the \pcp game: 
\begin{itemize}
    \item[$(iP)$] eventually, after reaching some subgame, $P_2$ always accepts the action proposals of $C$ and the gain predicted by $C$ is incorrect. 
    \item[$(iiP)$] eventually, $P_2$ keeps making one unique player~$i$ deviate in the decision phase, and $C$ is not able to adjust the gain to show that this deviation is not profitable for player~$i$.
\end{itemize}

\begin{defi}[Objectives of the \pcp game] \label{def:objectivesPCP}
Let $\rho = v'_0a'_0v'_1a'_1v'_2a'_2\dots$ be a play in the game structure of \pcpG. Let $\gainSim{\rho}$ be its simulated gain.

\begin{itemize}
\item The play $\rho$ is \emph{winning for Challenger} if one of the following conditions is satisfied: \begin{itemize}
    \item[$(iC)$] there exist $n \in \N$ and $g \in \lbrace 0, 1 \rbrace^{\vert \Pi \vert}$, such that for every state $v'_k$ with $k > n$, if $v'_k$ is a player-state, then it is of the form $(v_k, i_k, \bar{g}_k)$ such that its player-component $i_k$ equals $\varnothing$ and its gain-component satisfies
    $$\bar{g}_k = \gainSim{\rho},$$ 
    \item[$(iiC)$]
    there exists a player $i \neq 0$ such that 
    \begin{itemize}
        \item for every $n \in \N$, there exists $k > n$ for which $v'_{k}$ is a player-state with its player-component being equal to $i$,
        \item there exists $n \in \N$ such that for every state $v'_k$ with $k > n$, if $v'_{k} = (v_k, i_k, \bar{g}_k)$ is a player-state with $i_k \neq \varnothing$, then its player-component $i_k$ equals $i$ and the $i$-th component $\bar{g}_{k,i}$ of its gain-component $\bar{g}_k$ satisfies 
        $$\bar{g}_{k,i} \geq \gainSim{\rho}_i,$$
    \end{itemize}

    \item[$(iiiC)$] there exist two distinct players $i, j  \neq 0$ such that for every $n \in \N$, there exist $k, \ell > n$ for which $v'_{k}$ and $v'_{\ell}$ are player-states, with their player-component being respectively $i$ and $j$.
\end{itemize}
The set of plays satisfying one of these conditions is denoted $W_C$.

\item The play $\rho$ is \emph{winning for Prover~$1$ and Prover~$2$} if one of the following conditions is satisfied:
\begin{itemize}
    \item[$(iP)$]\label{cond_iP} there exist $n \in \N$ and $g \in \lbrace 0, 1 \rbrace^{\vert \Pi \vert}$, such that for every state $v'_k$ with $k > n$, if $v'_k$ is a player-state, then it is of the form $(v_k, i_k, \bar{g}_k)$ such that its player-component $i_k$ equals $\varnothing$ and its gain-component satisfies
    $$\bar{g}_k \neq \gainSim{\rho},$$ 
    \item[$(iiP)$]\label{cond_iiP} there exists a player $i \neq 0$ such that 
    \begin{itemize}
        \item for every $n \in \N$, there exists $k > n$ for which $v'_{k}$ is a player-state with its player-component being equal to $i$,
        \item there exists $n \in \N$ such that for every state $v'_k$ with $k > n$, if $v'_{k} = (v_k, i_k, \bar{g}_k)$ is a player-state with $i_k \neq \varnothing$, then its player-component $i_k$ equals $i$ and the $i$-th component $\bar{g}_{k,i}$ of its gain-component $\bar{g}_k$ satisfies 
        $$\bar{g}_{k,i} < \gainSim{\rho}_i,$$
    \end{itemize}

    \end{itemize}
The set of plays satisfying one of these conditions is denoted $W_P$.
\end{itemize}
\end{defi}

\begin{rema}
We have $W_{P_1}= W_{P_2} = W_P$ and $W_P = \Plays{G'} \setminus W_C$.
\end{rema}

Given a game $\mathcal G$, the definition of its \pcp game is now completed: we have \pcpG $= (G',(W_{P_1},W_C,W_{P_2}),\Obs)$, with the game structure $G'$ defined in Definitions~\ref{def:pcp-game-structure-states} and~\ref{def:pcp-game-structure-transitions}, the observation function $\Obs$ of $P_1$ defined in Definition~\ref{def:Obs-of-the-PCP-game}, and the objectives $W_{P_1},W_C,W_{P_2}$ of the three players defined in Definition~\ref{def:objectivesPCP}. To lighten notations, the resulting game \pcpG will often be denoted $(G',W_P,\Obs)$, where $W_P$ is the objective just defined.

\subsection{Equivalence Between \NCRS Problem and \texorpdfstring{\pcp} Game}

In this section, we prove that solving the \NCRS problem for a game $\mathcal G$ is equivalent to solving its \pcp game, as stated in the following theorem.  

\begin{theo}[Equivalence theorem between $\mathcal G$ and \pcpG]\label{thm:equivalence}
Let $\mathcal{G} = (G, (W_i)_{i \in \Pi})$ be a game, and \pcpG $= (G',W_P, \Obs)$ be its \pcp game. There exists in $\mathcal{G}$ a strategy $\sigma_0$ of player $0$, such that, 
for every subgame-perfect response $\bar{\sigma}_{-0}$, the play $\langle (\sigma_0, \bar{\sigma}_{-0}) \rangle_{v_0}$ is winning for player~$0$ if, and only if, there exists in \pcpG an observation-based strategy $\tau_{P_1}$ of $P_1$ such that for all strategies $\tau_C$ of $C$, there exists a strategy $\tau_{P_2}$ of $P_2$ such that the play $\langle \tau_{P_1}, \tau_C, \tau_{P_2}\rangle_{v'_0}$ belongs to $W_P$. 
\end{theo}

In the following, we assume that both games $\mathcal{G} = (G,(W_i)_{i \in \Pi})$ and \pcpG $= (G', W_P, \Obs)$ are fixed. Before proving the equivalence theorem, we are going to show that the observation function $\Obs$ is \pstable (see Definition~\ref{def:pstable}). So it makes sense to speak about observation-based strategies for player~$P_1$ as in Theorem~\ref{thm:equivalence}.

\begin{lemm}[\Pstability of $\Obs$] \label{lem:playerStable}
In \pcpG, the observation function $\Obs$ is \pstable for player~$P_1$.
\end{lemm}

The proof of this lemma is based on the next useful property. The proof of both lemmas is given in Appendix~\ref{app:Action-stability}.

\begin{lemm}[\Visibleactionstability of $\Obs$] \label{lem:actionStable}
The observation function $\Obs$ is \emph{\visibleactionstable} for player~$P_1$, that is, in $G'$, for every four states $v'_1, v'_2, u'_1, u'_2 \in V'$ and two actions $a'_1, a'_2 \in A'$ such that $\delta'(v'_1, a'_1) = u'_1$, $\delta'(v'_2, a'_2) = u'_2$, and $\Obs(v'_1) = \Obs(v'_2)$,
    \begin{itemize}
        \item if $a'_1 = a'_2$, then $\Obs(u'_1) = \Obs(u'_2)$,
        \item if $a'_1, a'_2$ are visible actions and $\Obs(u'_1) = \Obs(u'_2)$, then $a'_1 = a'_2$. 
    \end{itemize}
\end{lemm}

The \visibleactionstability is the counter-part of the \actionstability of a game structure (see Definition~\ref{def:game_structure}) in presence of imperfect information. Informally, this property states that the observation function $\Obs$ is able to distinguish between states reached via visibly different actions.

We now proceed to the proof of Theorem~\ref{thm:equivalence}. Towards proving this theorem, we need to describe how strategy profiles relate in both games. Specifically, we want to highlight how the game structure of the \pcp game enables to embed both the outcome of a strategy profile in $G$, the possible deviations from this profile, and the evaluation of whether they are profitable or not. We refer the reader to Figure~\ref{fig:prover-game-play} to remember the game structure of the \pcp game.

We start by describing a one-to-one correspondence between strategies of player~$0$ in $\mathcal G$ and observation-based strategies of $P_1$ in \pcpG. This correspondence uses the concept of simulated history defined in Definition~\ref{def:sim_play_in_G}.

\begin{defi}[Simulation of player~$0$]\label{defi:sim_of_p_0}
Let $\sigma_0$ be a strategy of player~$0$ in $\mathcal G$.
Let $\tau_{P_1}$ be the strategy of $P_1$ in \pcpG such that $\tau_{P_1}(h') = \sigma_0(\simu{h'})$ for all $h' \in \Histi{G'}{P_1}$. We call $\tau_{P_1}$ the \emph{simulation} of $\sigma_0$ and we say that $\sigma_0$ is \emph{simulated by} $\tau_{P_1}$. 
\end{defi}

\begin{lemm}[Player $0$ - $P_1$ strategy correspondence] \label{lem:correspondanceStrat}
Let $\sigma_0$ be a strategy of player~$0$ in $\mathcal G$. Then the strategy $\tau_{P_1}$ that simulates $\sigma_0$ is observation-based. Conversely, given an observation-based strategy $\tau_{P_1}$ of $P_1$ in \pcpG, there exists a strategy $\sigma_0$ of player~$0$ in $G$ such that $\tau_{P_1}$ is the simulation of $\sigma_0$.
\end{lemm}

\begin{proof}
The proof follows from the fact that there is a one-to-one correspondence between histories in $G$ and observed histories for $P_1$ in $G'$:
Indeed, $P_1$ only observes the $G$-component of each state in $G'$, which consists of states of $G$, and for each transition in $G'$ from a state not belonging to $P_1$, this $G$-component is repeated during the interaction phase between $C$ and $P_2$. 
Furthermore, the actions available to $P_1$ at an observed history, which corresponds to a unique history in $G$, are exactly the same as to player~$0$ at this corresponding history in $G$.
\end{proof}

We now turn to the transfer of strategies of $C$ in \pcpG to strategy profiles in $\mathcal G$. As already stated in Remark~\ref{rema:several_plays_in_pcp_for_one_in_G}, there may exist several histories in $G'$ that simulate the same history in $G$, because $P_2$ is the player that determines the next state of the simulating play, regardless of the proposal of $C$.
However, given the simulated history in $G$ and a strategy of $C$, we can recover the actual corresponding history in \pcpG as stated in the next lemma. From this lemma, we will derive how to transfer strategies of $C$ in Definition~\ref{def:simulated_profile}.

\begin{lemm}[Simulating history compatible with a Challenger strategy]\label{lem:tau_C_histG->histG'}
Let $\tau_C$ be a strategy of $C$ in \pcpG and $h$ be a history in $\mathcal{G}$. Then there exists a unique history $h'$ compatible with $\tau_C$, ending in a $G$-state, and such that $h = \simu{h'}$.
\end{lemm}

\begin{proof}
We proceed by induction on the length of the history $h$.
If $h = v_0$, then the unique corresponding history is $h' = v'_0 ~\bar g ~(v_0, \bar g)$ with $\bar g = \tau_C(v'_0)$. Assume now that Lemma~\ref{lem:tau_C_histG->histG'} is true for history $h$, and $h'$ is its unique corresponding history in \pcpG. Let us prove the lemma for history $hav$. We proceed as follows:
\begin{itemize}
    \item Suppose the last state $u$ of $h$ belongs to player~$0$. 
    Then the last state $u' = (u, \bar{g}) $ of $h'$ belongs to $P_1$. 
    Thus the only action possible for $P_1$ to continue simulating $hav$ is action $a$ followed by the $G$-state $(v,\bar g)$. The required history is therefore equal to $h' ~a ~(v,\bar g)$ that ends with a $G$-state. 
    
    \item Suppose the last state $u$ of $h$ belongs to player~$i$, with $i \neq 0$. 
    Then, by induction hypothesis, the last state of $h'$ is a $G$-state of the form $(u, \bar{g})$, and belongs to Challenger.
    The next action $(a', i)$ is thus determined by $\tau_C$, which in turn determines the next action-state $(u,a',\bar{g})$.
    \begin{itemize}
    \item If $a' = a$, we are in the case where $P_2$ accepts Challenger's proposal, thus plays $a'$, yielding the successor state $(v, \varnothing, \bar{g})$. This state is followed by the action $\bar{g}$ and the state $(v, \bar{g})$, by construction of \pcpG. The required history is then equal to $h'~(a, i)~(u, a, \bar{g})~a~ (v, \varnothing, \bar{g})~ \bar g ~(v, \bar{g})$. It ends with a $G$-state and simulates $hav$.
    
    \item If $a' \neq a$, we are in the case where $P_2$ refuses Challenger's proposal, thus plays $a$, yielding the successor state $(v, i, \bar{g})$, followed by the action $\bar{g}'$ determined by $\tau_C$ and then the state $(v, \bar{g}')$. This leads to the history $h'~(a', i)~(u,a', \bar{g})~a~(v, i, \bar{g})~ \bar{g}'~ (v, \bar{g}')$, which ends with a $G$-state and simulates $hav$.
    \end{itemize}
\end{itemize}
\end{proof}

\begin{defi}[Simulation of the players $i$, $i \neq 0$]\label{def:simulated_profile}
Let $\tau_C$ be a strategy of $C$. For every player~$i \neq 0$ in $G$, and every history $h \in \Histi{G}{i}$, let $h'$ be the unique history of \pcpG compatible with $\tau_C$, that ends in a $G$-state, and such that $h = \simu{h'}$. Then we define $\sigma_i(h) = a$ such that $\tau_C(h') = (a,i)$. We say that $\bar{\sigma}_{-0} = (\sigma_i)_{i \in \Pi\setminus\{0\}}$ is \emph{simulated by} $\tau_C$. 
\end{defi}

Let us finally propose a useful particular strategy for $P_2$, where he accepts every action proposal of $C$ and never deviates.

\begin{defi}[Accepting strategy of $P_2$]\label{def:accepting_strategy} 
Let $\tau_{\mathrm{acc}}$ be the strategy of $P_2$ such that for every history $h' \in \Histi{G'}{P_2}$, if its last state is $(v, a, \bar{g})$, then $\tau_{\mathrm{acc}} (h') = a$. We call $\tau_{\mathrm{acc}}$ the \emph{accepting strategy} of $P_2$. 
\end{defi}

We are now ready to prove the first direction of Theorem~\ref{thm:equivalence}. 

\begin{lemm}\label{lem:SPE-NCRSp->pcp}
Let $\sigma_0$ be a solution to the \NCRS problem of $\mathcal G$.
Then, in \pcpG, there exists a strategy $\tau_{P_1}$ of $P_1$, such that for every strategy $\tau_C$ of $C$, there exists a strategy $\tau_{P_2}$ of $P_2$, such that 
$\langle \tau_{P_1}, \tau_C, \tau_{P_2} \rangle_{v'_0} \in W_P$.
\end{lemm}

\begin{proof}
Let $\tau_{P_1}$ be the observation-based strategy of $P_1$ that simulates $\sigma_0$, as defined in Definition~\ref{defi:sim_of_p_0}.
Let $\tau_C$ be a strategy of $C$.
Let $\bar{\sigma}_{-0}$ be the strategy profile obtained from $\tau_C$, as defined in Definition~\ref{def:simulated_profile}.
There are two possibilities, either $(\sigma_0, \bar{\sigma}_{-0})$ is a 0-fixed SPE in $\mathcal{G}$, or
it is not a 0-fixed SPE.
Let us show that in both cases, there exists $\tau_{P_2}$ such that $\langle \tau_{P_1}, \tau_C, \tau_{P_2} \rangle_{v'_0} \in W_P$.

\begin{enumerate}
    \item Suppose $(\sigma_0, \bar{\sigma}_{-0})$ is a 0-fixed SPE in $\mathcal{G}$. 
    Let $\tau_{P_2}$ be $\tau_{\mathrm{acc}}$, the accepting strategy of $P_2$.
    By definition of $\tau_{P_2}$, the simulated play of $\langle \tau_{P_1}, \tau_C, \tau_{P_2} \rangle_{v'_0}$ is $\langle \sigma_0, \bar{\sigma}_{-0}\rangle_{v_0}$.
    In that case, since $\sigma_0$ is a solution to the \NCRS problem of $\mathcal{G}$, we know that $\langle \sigma_0, \bar{\sigma}_{-0}\rangle_{v_0} \in W_0$.
    However, by construction of \pcpG and the fact that $\tau_{P_2}$ is the accepting strategy of $P_2$, we know that $C$ has predicted a loss for player~$0$ from the start, and has never had to adjust his prediction. Thus, we are in the case $(iP)$ of the winning condition for the Provers.
    
    \item Suppose now $(\sigma_0, \bar{\sigma}_{-0})$ is not a 0-fixed SPE in $\mathcal{G}$. 
    This means that it is not a 0-fixed NE in some subgame of $\mathcal G$.
    Let $h$ be a history compatible with $\sigma_0$ such that $(\sigma_0, \bar{\sigma}_{-0})_{\upharpoonright h}$ is not an NE.
    Let player~$i \neq 0$ be a player in $\mathcal{G}$ that has a profitable deviation $\sigma'_i$ in the subgame $\mathcal{G}_{\upharpoonright h}$ starting after $h$. 
    In \pcpG, by Lemma~\ref{lem:tau_C_histG->histG'}, there exists a unique history $h'$ ending in a $G$-state, compatible with $\tau_C$, such that $\simu{h'} = h$. Since $h$ is compatible with $\sigma_0$, the history $h'$ is also compatible with $\tau_{P_1}$ by Definition~\ref{defi:sim_of_p_0}.
    We distinguish two cases, depending on Challenger's prediction for player~$i$ at~$h'$.
    \begin{enumerate}
        \item Suppose $C$ predicts a loss for player~$i$ at $h'$. 
        In this case, the strategy $\tau_{P_2}$ is defined as the strategy that first makes $h'$ be a prefix of the play, and then simulates $\sigma'_i$, the profitable deviation of player~$i$ (thus accepting all propositions of $C$ for players other than player~$i$). 
        Formally, let $h'_1 \in \Histi{G'}{P_2}$ and let $h_1 = \simu{h'_1}$. If $h_1 \sqsubset h$, we define $\tau_{P_2}(h'_1) = a$ such that $\simu{h'_1a} \sqsubseteq h$. 
        If $h \sqsubseteq h_1$ and $h'_1$ ends with the action-state $(v,a,\bar g)$ with $v \in V_j$ for some $j \neq 0$, we define $\tau_{P_2}(h'_1) = a$ if $j \neq i$ and $\tau_{P_2}(h'_1) = \sigma'_i(\simu{h'_1})$ if $j = i$. Otherwise, $\tau_{P_2}(h'_1)$ is defined arbitrarily. Recall that $C$ predicts a loss for player~$i$ at $h'$. By construction of \pcpG, the only possible gain prediction that Challenger can propose when $P_2$ refuses his proposal is a loss for player~$i$, as it is not possible to strictly decrease a gain of~$0$.
        Since $P_2$ does not make any other player deviate, Challenger must thus remain predicting a loss for player~$i$ forever.
        However, since $\sigma'_i$ is a profitable deviation, this means that the outcome of the profile $( \sigma_0, \bar{\sigma}_{- \lbrace0,i \rbrace}, \sigma'_i )_{\upharpoonright h}$ is winning for player~$i$ in the subgame ${\mathcal G}_{\upharpoonright h}$.
        If $\sigma'_i$ has a finite number of deviation points from $\sigma_i$, then the winning condition $(iP)$ of Definition~\ref{cond_iP} is satisfied.
        If $\sigma'_i$ has an infinite number of deviation points from $\sigma_i$, then the winning condition $(iiP)$ of Definition~\ref{cond_iiP} is satisfied. 
        
        \item Suppose $C$ predicts a win for player~$i$ at $h'$.
        In that case, the strategy $\tau_{P_2}$ is defined as the strategy that first makes $h'$ be a prefix of the play, and then switches to $\tau_{\mathrm{acc}}$.
        We skip the formal definition of $\tau_{P_2}$.
        The resulting simulated play in the subgame $\mathcal{G}_{\upharpoonright h}$ is the outcome of $(\sigma_0, \bar{\sigma}_{-0})_{\upharpoonright h}$ which we know is losing for player~$i$ (otherwise there would not exist any profitable deviation from $h$). 
        As $C$ predicts a win for player~$i$, the  winning condition $(iP)$ of Definition~\ref{cond_iP} is satisfied.
    \end{enumerate}
\end{enumerate}
\end{proof}

Let us now turn to the other direction of Theorem~\ref{thm:equivalence}.

\begin{lemm}\label{lem:pcp->SPE-NCRSp}
Assume there is no solution to the \NCRS problem in $\mathcal G$.
Then, in \pcpG, for every strategy $\tau_{P_1}$ of $P_1$, there exists a strategy $\tau_C$ of $C$, such that for every strategy $\tau_{P_2}$ of $P_2$, we have $\langle \tau_{P_1}, \tau_C, \tau_{P_2} \rangle_{v'_0} \in W_C$.
\end{lemm}

\begin{proof} 
Let $\tau_{P_1}$ be a strategy of $P_1$ in \pcpG.
Let $\sigma_0$ be the strategy of player~$0$ in $\mathcal G$ such that $\tau_{P_1}$ is its simulation, which exists and is unique by Lemma~\ref{lem:correspondanceStrat}.
By Theorem~\ref{th:SPE-existence}, there exists a $\sigma_0$-fixed SPE in $\mathcal{G}$. Among all those $\sigma_0$-fixed SPEs, one must have an outcome losing for player~$0$ as, by assumption, $\sigma_0$ is not a solution of the \NCRS problem in $\mathcal G$. Let $\bar{\sigma} = (\sigma_0, \bar{\sigma}_{-0})$ be such a $\sigma_0$-fixed SPE.
Let us now come back to \pcpG and define the following strategy $\tau_C$ for $C$:
\begin{enumerate}
    \item At history $v'_0$, we have $\tau_C(v'_0) = \bar{g}$, where $\bar{g}$ is the gain profile of $\bar{\sigma}$. Recall that $g_0 = 0$ by choice of the $0$-fixed SPE.
    \item At any history $h'$ that ends in a $G$-state, we define $\tau_C(h') = (a,i)$ such that $h = \simu{h'}$ ends in a state controlled by player~$i$ and $\sigma_i(h) = a$.
    \item At any history $h'$ that ends in a player-state, we define $\tau_C(h') = \bar{g}'$ such that $\bar g'$ is the gain profile of $\bar{\sigma}_{\upharpoonright \simu{h'}}$.
\end{enumerate}

Let us show that the strategy $\tau_C$ is well-defined with respect to the gains. First, at history $v'_0$, it respects the structural constraint of \pcpG that the predicted gain for player~$0$ must be $0$. Second, at any history $h'$ that ends in a player-state $(u,i,\bar{g})$ with $i \neq \varnothing$, the gain $\bar{g}' = \tau_C(h')$ must satisfy $g'_i \leq g_i$ by definition of the \pcp game. In this case, $h'$ is of the form 
$$h' = h'_1 ~(a,i)~(v,a,\bar{g})~ a'~ (u,i,\bar{g})$$
(Challenger proposes action $(a,i)$ and Prover~$2$ refuses this proposal by playing $a' \neq a$). Let $h = \simu{h'_1} = \simu{h'_1 ~(a,i)~ (v,a,\bar{g})}$, then $ha'u = \simu{h'}$. Moreover, by definition of $\tau_C$, $\bar{g}$ is the gain of  $\langle \bar{\sigma} \rangle_{\upharpoonright h}$ and $\bar{g}'$ is the gain of $\langle \bar{\sigma}\rangle_{\upharpoonright ha'u}$. Due to the deviation of $P_2$, and as $\bar{\sigma}$ is a $0$-fixed SPE, $g'_i \leq g_i$ as expected. 

Let now $\tau_{P_2}$ be a strategy of $P_2$.
Consider the play $\rho = \langle \tau_{P_1}, \tau_C, \tau_{P_2} \rangle_{v'_0}$.
If there are infinitely many deviations by two different players of $\mathcal G$ in the play $\rho$, then $\rho$ satisfies condition $(iiiC)$ of Definition~\ref{def:objectivesPCP} and is winning for Challenger, and $\rho \in W_C$.
Assume now that there is at most one player that $P_2$ makes deviate.
There are two cases: 
\begin{enumerate}
\item Suppose that $\tau_{P_2}$ prescribes a finite (possibly null) number of deviations along $\rho$, and then switches to the accepting strategy $\tau_{acc}$. Let $h'$ be the prefix of $\rho$ such that it ends with a player-state that is the last one with the player-component different from $\varnothing$ (or it ends at $(v_0,\bar{g})$ if there is no deviation at all). Then in the subgame of $\mathcal G$ starting after $\simu{h'}$, Challenger will always predict a gain profile equal to $\langle \bar{\sigma} \rangle_{\upharpoonright \simu{h'}}$ and Prover~$2$ will always accept the proposal. As this is also the gain profile of $\simu{\rho}$, $\rho$ satisfies condition $(iC)$ and is winning for Challenger.

\item Suppose that $\tau_{P_2}$ prescribes an infinite number of deviations for the same player~$i$ along $\rho$. Let $\rho = h'\rho_1$ be such that $h'$ ends in a player-state $(u,i,\bar{g})$  and that every player-state in $\rho_1$ has a player-component either equal to $\varnothing$ or $i$ and the gain-component $g_i$ for player~$i$ is stable along $\rho_1$. Such a suffix $\rho_1$ of $\rho$ exists as there are only two different gain values, only player~$i$ deviates in $\rho_1$, and there is no opportunity for $g_i$ to increase. 
If the stable gain-component $g_i$ of $\rho_1$ is $1$, we clearly have $g_i = 1 \geq \gainSim{\rho}_i$. Thus condition $(iiC)$ is satisfied, and Challenger wins. If this gain-component is $0$, recall that $\bar{\sigma}$ is a $0$-fixed SPE, thus it is an NE in particular in the subgame $\mathcal{G}_{\upharpoonright_{\simu{h'}}}$.
The outcome of ${\bar \sigma}_{\upharpoonright \simu{h'}}$ is losing for player~$i$:
By definition of $\tau_C$ (item $3.$), we also know that Challenger predicts a gain of $0$ for a player at $(u,i,\bar{g})$ only if the outcome of $\bar{\sigma}_{\upharpoonright \simu{h'}}$ is losing for this player. 
Any deviation from this outcome is thus also losing for player~$i$, including $\simu{\rho}$. 
Thus, condition $(iiC)$ is again satisfied, and Challenger wins. 

\end{enumerate}
\end{proof}

\subsection{The \texorpdfstring{\pcp} Game as a Rabin Game}\label{subsec:pr-g-parity}

From now on, we \emph{focus on parity games}. Recall that all $\omega$-regular objectives can be encoded as a parity objective. The particular case of reachability games is reported to the end of the paper (in Section~\ref{sec:reach_case}) and treated with a more fine-tuned approach in a way to get the complexity result announced in Theorem~\ref{thm:mainReach}. Given a parity game $\mathcal G$, we show that its corresponding \pcp game can be seen as a three-player game with the objective $W_P$ for the two Provers translated into a \emph{Rabin} objective. The approach is to use a deterministic automaton $\ObsAut$ that \emph{observes} the states of \pcpG. Then the \emph{synchronized product} of the game structure of \pcpG with this observer automaton $\ObsAut$ is equipped with a Rabin objective translating $W_P$ and thus leads to the announced three-player Rabin game.

We recall the notion of Rabin objective and its dual notion of Streett objective for a game structure $G'$ like the \pcpG \cite{DBLP:conf/dagstuhl/2001automata}.

\begin{defi}[Rabin/Streett objectives] \label{def:rabin-objective}
Let $(E_j,F_j)_{j \in J}$ be a set of pairs $(E_j,F_j)$ with  $E_j, F_j \subseteq V'$.
\begin{itemize}
\item The \emph{Rabin} objective $\Rabin((E_j,F_j)_j)$ is the set of plays $\rho \in \Plays{G'}$ 
such that there exists $j \in J$ with $\Inf(\rho) \cap E_j = \varnothing$ and $\Inf(\rho) \cap F_j \neq \varnothing$.
\item The \emph{Streett} objective $\Streett((E_j,F_j)_j)$ is the set of plays $\rho \in \Plays{G'}$ 
such that for all $j \in J$,  $\Inf(\rho) \cap E_j \neq \varnothing$ or $\Inf(\rho) \cap F_j = \varnothing$.
\end{itemize}
The \emph{size} of a Rabin/Streett objective is the number $|J|$ of its pairs.
\end{defi}

To determine who among the Provers and Challenger wins a play, one needs to know the following: 
first, whether there are more than one player in the original game that is made to deviate by $P_2$,
second, whether deviations are indeed profitable,
and finally, whether Challenger predicts gains accurately.
This can be done by keeping track in the observer automaton of the last deviation point enforced by $P_2$ and the associated deviating player, and the current status of alternation of deviating players. 
To then be able to decide whether Challenger's gain predictions are correct, one additionally needs to extract the gain of the corresponding simulated play in the original game.
In order to monitor the simulated play, one needs to keep track of which priorities are seen along the play in the parity game~$\mathcal G$. 

In the next proposition, we denote by \pcpOG the synchronized product of \pcpG with the observer automaton $\ObsAut$, and its game structure by $G' \times \ObsAut$.

\begin{prop} \label{prop:PCPParityRabin}
Let $\mathcal{G} = (G,(W_i)_{i \in \Pi})$ be a parity game such that each $W_i$ is a parity objective using the priority function $\alpha_i$. Let $|V|$, $|A|$, and $|\Pi|$ be the size of the game structure $G$, and $|\alpha_i|$ be the size of each $W_i$.
Let $G'$ be the game structure of its \pcp game (according to Section~\ref{subsec:DefPCPgame}).
Then, there exists an observer automaton $\mathcal O$ such that the \pcp game of $\mathcal{G}$ is equivalent to a three-player game \pcpOG $= (G'\times \ObsAut,\Rabin((E_j,F_j)_{j \in J}) ,\Obs)$ with a Rabin objective $\Rabin((E_j,F_j)_{j \in J})$ for the Provers. The game structure $G' \times \ObsAut$ of \pcpOG has a size
\begin{itemize}
\item $|V'|$ linear in $|V|$ and $|A|$, and exponential in $|\Pi|$, 
\item $|A'|$ linear in $|A|$ and exponential in $|\Pi|$,
\end{itemize}
and its Rabin objective has a size $|J|$ linear in $|\Pi|$ and each $|\alpha_i|$, $i \in \Pi$.
\end{prop}

\begin{proof}
Let us first briefly recall the conditions defining the objective $W_C$ of Challenger and the opposed objective $W_P$ of the Provers in \pcpG (see Definition \ref{def:objectivesPCP}). 
A play $\rho$ is winning for $C$ if one of the following cases occurs:
\begin{itemize}
    \item[$(iC)$] eventually, no player deviates ($P_2$ always accepts the action proposals of $C$) and the gain predicted by $C$ is correct (it is equal to $\gainSim{\rho}$),
    \item[$(iiC)$] one unique player~$i$ infinitely often deviates, and the gain $\bar g$ predicted by $C$ shows that this deviation is not profitable for this player (the component $g_i$ of $\bar g$ is greater than or equal to $\gainSim{\rho}_i$),
    \item[$(iiiC)$] at least two different players infinitely often deviate.
\end{itemize}
A play $\rho$ is winning for $P_1$ and $P_2$ if $\rho \not\in W_C$, that is,
\begin{itemize}
    \item[$(iP)$] eventually, no player deviates and the gain predicted by $C$ is not correct,
    \item[$(iiP)$] one unique player~$i$ infinitely often deviates, and the gain $\bar g$ predicted by $C$ shows that this deviation is profitable for this player.
\end{itemize}

Before proving the proposition, let us also recall that any parity objective is a special case of Rabin objective. For instance, in the given game $\mathcal G$, the parity objective $\alpha_i$ of player~$i$ using the priorities of $\{0,\dots,d_i\}$ can be translated into a Rabin objective $\Rabin((E_p,F_p)_{p \text{ even}})$ with the following pairs $(E_p,F_p)$ for each even priority $p \in \{0,\dots,d_i\}$:
$$E_p = \{v \in V \mid \alpha_i(v) < p\} \text{ and } F_p = \{v \in V \mid \alpha_i(v) = p\}.$$
Indeed, with $F_p$, we impose to see the even priority $p$ infinitely often, and with $E_p$, we impose to eventually see no smaller priority. Notice that when $p=0$, we have $E_p = \varnothing$.

We are going to construct a deterministic automaton $\ObsAut$ that observes the states of \pcpG in a way to detect the last player who deviates, and whether two players infinitely often deviate.
That is, the states of $\mathcal O$ are of the form $(j,d,f)$ where $j$ is the last player seen in a player-state, $d$ is the last deviating player, and $f$ indicates whether the last two deviating players are different or not.
This deterministic observer is defined as follows (Figure~\ref{fig:prover-game-play} should again be helpful):
\begin{itemize}
    \item The states of $\ObsAut$ are tuples $(j,d,f) \in (\Pi \setminus \{0\} \cup \{\varnothing\}) \times (\Pi \setminus \{0\}) \times \{0,1\}$. The component $j$ is the last player seen in a player-state (including $j = \varnothing$), $d$ is the last deviating player (with $d \neq \varnothing$), and $f$ oscillates between $0$ and $1$ ($f=1$ meaning that the last two deviating players are different). 
    \item The initial state of $\ObsAut$ is the tuple $(0,0,0)$. (As this state will be synchronized with the initial state $v'_0$ of $G'$, we can arbitrarily choose the values of the tuple).
    \item The transitions $(j,d,f),v',(j',d',f'))$ of $\ObsAut$ labeled with states $v' \in V'$ of $G'$ are such that:
    \begin{itemize}
        \item if $v'$ is not a player-state, then $(j',d',f') = (j,d,f)$.
        \item if $v' = (v,i,\bar{g})$ is a player-state with $i = \varnothing$, then $j' = \varnothing$ and $(d',f') = (d,f)$,
        \item if $v' = (v,i,\bar{g})$ is a player-state with $i \neq \varnothing$, then $j' = d' = i$ and $f' = 1 \Leftrightarrow d' \neq d$,
    \end{itemize}
\end{itemize}

We then construct the synchronized product $G'\times \ObsAut$ of the game structure $G'$ of \pcpG with the observer $\mathcal O$ and define its observation $\Obs$ in the following way. The resulting game structure is composed of the states $v'$ of Definition~\ref{def:pcp-game-structure-states} extended with the observed information $(j,d,f)$. Moreover, the observation function of the extended states is naturally defined such that the observation of $(v',j,d,f)$ is equal to $\Obs(v')$, while we keep the same observation function for the actions. The resulting function is still denoted by $\Obs$. It remains to describe the Rabin objective of the synchronized product. This objective uses the following pairs that we explain below:
\begin{itemize}
\item for each player $i \in \Pi$ and each even priority $p \in \{0,\dots , d_i\}$, a pair $(E_{i,p},F_{i,p})$ such that 
$$E_{i,p} = \{(v',j,d,f) \mid j \neq \varnothing, \alpha_i(v) < p \text{ where $v$ is the $G$-component of $v'$}\},$$ 
$$F_{i,p} = \{(v',j,d,f) \mid g_i = 0, \alpha_i(v) = p  \text{ where $v,\bar{g}$ are the $G$- and gain-components of $v'$}\},$$
\item another set of similar pairs $(E_{i,p},F_{i,p})$ where $p$ is now odd, and $g_i = 1$ in the definition of $F_{i,p}$,
\item for each player $j \in \Pi \setminus \{0\}$ and each even priority $p \in \{0,\dots , d_j\}$, a pair $(E'_{j,p},F'_{j,p})$ such that 
$$E'_{j,p} = \{(v',j,d,f) \mid f = 1, \alpha_j(v) < p \text{ where $v$ is the $G$-component of $v'$}  \},$$ 
$$F'_{j,p} = \{(v',j,d,f) \mid g_j = 0, \alpha_j(v) = p  \text{ where $v,\bar{g}$ are the $G$- and gain-components of $v'$}\}.$$  
\end{itemize}
Let us explain these pairs. The first two sets of pairs replace the pair $(E_0,F_0)$ of the previous proof, and correspond to condition $(iP)$ of the Provers. Indeed, the pair $(E_{i,p},F_{i,p})$ in the first set of pairs is such that eventually, no player deviates, and the gain $\bar g$ predicted by $C$ is not correct because $g_i = 0$ and the lowest priority seen infinitely often with respect to $\alpha_i$ is $p$ which is even. Similarly, the pair $(E_{i,p},F_{i,p)})$ in the second set of pairs is such that $g_i = 1$ and the lowest priority seen infinitely often with respect to $\alpha_i$ is $p$ which is odd. The last set of pairs corresponds to condition $(iiP)$ of the Provers. Indeed, the pair $(E'_{j,p},F'_{j,p})$ in the last set of pairs is such that $j \in \Pi \setminus \{0\}$ is the unique infinitely deviating player and his deviation is profitable as $g_j = 0$ and $p$ is the lowest (even) priority seen infinitely often with respect to $\alpha_j$.

In this way, given a parity game $\mathcal{G}$, its \pcp game is equivalent to \pcpOG with a Rabin objective for the Provers having a number of pairs that is linear in $|\Pi|$ and each $|\alpha_i|$, $i \in \Pi$. Moreover, the set of states $|V'|$ of \pcpOG is linear in $|V|$ and $|A|$, and exponential in $|\Pi|$, 
and its set of actions $|A'|$ is linear in $|A|$ and exponential in $|\Pi|$. 
\end{proof}

In the following section, we unfold the process of solving the game \pcpOG, which we call simply \pcp game in the sequel (to maintain readability).

\section{Solving the \texorpdfstring{\pcp} Game} \label{sec:SolvingPCPGame}

Thanks to Theorem~\ref{thm:equivalence}, we have an equivalence between the existence of a solution of the \NCRS problem for a game and the fact that the Provers are able to win the associated \pcp game.
More precisely, the situation where $P_1$ has \emph{a} strategy, such that for \emph{every} strategy of Challenger, $P_2$ has \emph{a} strategy to make the Provers win.
The remaining question is how to determine whether this is the case? 
In other words, how can one \emph{solve}, in this particular sense, this three-player game with imperfect information?
In this more technical section, we answer this question in two steps:
First, we get rid of the three-player setting, by eliminating one Prover (see Section~\ref{subsec:2-pl-pr-g-ii}).
Second, we get rid of the imperfect information and work with a parity objective instead of a Rabin one (see Section~\ref{subsec:eliminating}).
This allows us to obtain an equivalent \emph{two-player zero-sum game parity game with perfect information}, which is effectively solvable.

\subsection{From two Provers to one Prover} \label{subsec:2-pl-pr-g-ii}
In this section, we show how to transform a three-player Rabin game with imperfect information into an equivalent\footnote{in a sense that is given in Theorem~\ref{thm:2Provers->1Prover} below.} 
two-player zero-sum Rabin game with imperfect information. 
The technique used here is inspired from~\cite{DBLP:conf/icalp/Chatterjee014}.
However, their setting is more general than ours (as they allow more than one player to have imperfect information, as long as it respects some hierarchical constraint on the level of information).
Thus, to obtain finer complexity measures, we adapt their technique to our setting and describe extensively our construction.
Note that in this section, the original game on which one wishes to solve the \NCRS problem is irrelevant, as the construction can be studied independently and works not only for a \pcp game, but for any three-player Rabin game with similar information distribution.
Thus, for readability, we forget about the underlying original game and set out from a generic game that has the same properties as the ones obtained in Section~\ref{subsec:DefPCPgame}.
Notation-wise, this enables us to name by $\mathcal{G}$ the \pcp-like games and to get rid of the $'$ when referring to them.

From the previous Section~\ref{subsec:pr-g-parity}, we consider a three-player game $\mathcal{G} = (G,\Rabin((E_j,F_j)_{j \in J}),\Obs)$ with 
\begin{itemize}
    \item a game structure $G = (V, A, \lbrace P_1, C, P_2 \rbrace, \delta, v_0)$, 
    \item $V = V_{P_1} \cup V_C \cup V_{P_2}$ and $A = A_{P_1} \cup A_C \cup A_{P_2}$ where $V_{P_1}, V_C, V_{P_2}$ are pairwise distinct sets as well as $A_{P_1}, A_C, A_{P_2}$,
    \item a Rabin objective $\Rabin((E_j,F_j)_{j \in J})$ for players~$P_1$ and~$P_2$, the opposite Streett objective $\Streett((E_j,F_j)_{j \in J})$ for player~$C$,
    \item an observation function $\Obs$ for player~$P_1$ such that 
    \begin{itemize}
        \item all actions $a \in A_{P_1} \cup A_{P_2}$ are visible,
        \item $\Obs(a) = \sharp$ for all $a \in A_C$,
        \item $\Obs$ is player- and \visibleactionstable.\footnote{An observation function that is \pstable is defined in Definition~\ref{def:pstable} and that is \visibleactionstable in Lemma~\ref{lem:actionStable}.}
    \end{itemize}
\end{itemize}
In this section, we often call this game a three-player game.

Let us explain how to go from two Provers to one Prover, while assuming the same roles. The main idea is to use imperfect information to ensure merging the two Provers does not grant too much knowledge to the new single Prover. 
Indeed, if the new Prover had perfect information, he could not simulate the first Prover truthfully.
Thus we let the new Prover have the same level of information as Prover~$1$ in $\mathcal G$.
However, in order to let the new Prover have as much actions available as Prover~$2$ in $\mathcal G$ and stay observation-based, we modify the action set to include all \emph{functions} from states of $P_2$ to actions of $P_2$.
We define the corresponding \pc game, denoted \pcG, as follows.

\begin{defi}[\pc game] \label{def:OneProver}
Let $\mathcal{G} = (G,\Rabin((E_j,F_j)_{j \in J}), \Obs)$ be a three-player game. 
The \pc game associated with $\mathcal G$ is a two-player game \pcG $= (G',\Rabin((E_j,F_j)_{j \in J}) ,\Obs')$ defined as follows.
Its game structure is $G' = (V', A', \lbrace P, C \rbrace, \delta', v'_0)$, where:
\begin{itemize}
    \item $P$ and $C$ are the two players,
    \item $V' = V$ with $V' = V'_P \cup V'_C$ such that $V'_P = V_{P_1} \cup V_{P_2}$ and $V'_C = V_C$,
    \item $v'_0 = v_0$ is the initial state,
    \item $A' = A'_P \cup A'_C$ with $A'_C = A_C$ and $A'_P = A_{P_1} \cup \ARondePd$ where  
$$\ARondePd = \{ f: V_{P_2} \rightarrow A_{P_2} \mid f \mbox{ partial function} \}.$$ 
    \item the transition function $\delta'$ is defined as follows\footnote{Note that this game departs slightly from Definition~\ref{def:game_structure}, as $G'$ is not action-unique: two different partial functions can send an element of the domain to the same image.}:
    \begin{itemize}
        \item $\delta'(v',a') = \delta(v',a')$ for every $(v',a') \in V_{P_1} \times A_{P_1} \,\cup\, V_C \times A_C$,
        \item $\delta'(v',a') = \delta(v',a'(v'))$ for every $(v',a') \in V_{P_2} \times \ARondePd$.
    \end{itemize}
\end{itemize}

The objective $\Rabin((E_j,F_j)_{j \in J})$ for player~$P$ uses the same Rabin pairs as in $\mathcal G$. Player~$C$ has a Streett objective that is the opposite of the Rabin objective.

Player~$C$ has perfect information and the observation function $\Obs'$ of player~$P$ is defined as follows:
    \begin{itemize}
        \item $\Obs' = \Obs$ on $V'= V$ and on $A_{P_1} \cup A_C$,
        \item $\Obs'(a') = a'$ for every $a' \in \ARondePd$.
     \end{itemize}
\end{defi}

Let us make some comments. First, the actions $a'$ of $\ARondePd \subset A'_P$ are partial functions $f: V_{P_2} \rightarrow A_{P_2}$ such that $a'(v') \in A_{P_2}$ for every $v'$ in the domain of $f$. Second, $\Obs'$ is such that all actions $a' \in A_{P_1} \cup \ARondePd$ are visible. In particular, player~$P$ is able to distinguish between the actions of $A_{P_1}$ and $\ARondePd$ as those sets of actions are disjoint. 

We describe below the size of \pcG, given the size of $\mathcal{G}$.

\begin{lemm} \label{lem:sizePCgame} 
Let $|V|, |A|$ be the size of the game structure of $\mathcal G$ and $|J|$ be the size of its Rabin objective. Then the associated \pc game  
\begin{itemize}
    \item has a game structure with size $|V'|$ linear in $|V|$, and $|A'|$ exponential in $|V|$ and $\log(|A|)$, and
    \item it keeps the size $|J|$ of $\mathcal G$ for its Rabin objective.  
\end{itemize}
\end{lemm}

\begin{proof}
The proof is immediate from the definition of \pcG~if we notice that the size of $\ARondePd$ is in $O(|A|^{|V|})) = O(2^{|V|\cdot \log(|A|)})$.
\end{proof}

Before presenting the equivalence of $\mathcal G$ and \pcG in Theorem~\ref{thm:2Provers->1Prover} below, we need to study the player-stability of the observation function $\Obs'$ in $G'$. Indeed, this property is required to speak about observation-based strategies for player $P$ (see Section~\ref{subset:Obs}). The next lemma relates the observability of histories between the two games.

\begin{lemm}[Observed correspondence of histories] 
\label{lem:corresp}
 \hfill
\begin{enumerate}
    \item Let $h' = v'_0a'_0v'_1a'_1 \ldots v'_k \in \Hist{G'}$ be a history in $G'$. Then, to $h'$ corresponds a unique history $h = v_0a_0v_1a_1 \ldots v_k \in \Hist{G}$, denoted by $h = \corresp{h'}$, such that 
    \begin{itemize}
        \item $v_\ell = v'_\ell$ for all $v_\ell$,
        \item $a_\ell = a'_\ell$ for all $a'_\ell \in A_{P_1} \cup A_C$,
        \item $a_\ell = a'_\ell(v'_\ell)$ for all $a'_\ell \in \ARondePd$.
    \end{itemize}
    The same property holds for plays $\rho' \in \Plays{G'}$.
    \item Let $h', g' \in \Hist{G'}$ be two histories in $G'$ such that $\Obs'(h') = \Obs'(g')$. Then for their corresponding histories $h = \corresp{h'}$ and $g = \corresp{g'}$ in $\Hist{G}$, we have that $\Obs(h) = \Obs(g)$.
\end{enumerate}
\end{lemm}

\begin{proof}
The first statement is a direct consequence of the definition of the \pc game. Let us prove the second one. Let $h' = v'_0a'_0v'_1a'_1 \ldots v'_k$ and $g' = u'_0b'_0u'_1b'_1 \ldots u'_k$ be two histories in $G'$ such that $\Obs'(h') = \Obs'(g')$. Let $h = \corresp{h'} = v_0a_0v_1a_1 \ldots v_k$ and $g = \corresp{g'} = u_0b_0u_1b_1 \ldots u_k$ be the corresponding histories in $G$. By definition of $\Obs'$, we have that $\Obs(v_\ell) = \Obs(u_\ell)$ for all $\ell \in \{0,\ldots,k\}$. From his observation of the actions through $\Obs'$, player~$P$ knows that each pair $a'_\ell, b'_\ell$ is composed of actions from the same subset of actions: either $A_C$ or $A_{P_1}$ or $\ARondePd$. 
Moreover, by definition of $\Obs'$, we get that $\Obs(a_\ell) = \Obs(b_\ell)$ for all pairs $a_\ell, b_\ell$ of actions both in either $A_{P_1}$ or in $A_C$. It remains to consider each pair $a'_\ell, b'_\ell \in \ARondePd$ and the corresponding pair $a_\ell = a'_\ell(v'_\ell), b_\ell = b'_\ell(u'_\ell)$ with $a_\ell, b_\ell \in A_{P_2}$. We have to prove that $\Obs(a_\ell) = \Obs(b_\ell)$, that is, $a_\ell = b_\ell$. As $\Obs(v_\ell) = \Obs(u_\ell), \Obs(v_{\ell+1}) = \Obs(u_{\ell+1})$, and $a_\ell, b_\ell$ are visible actions, we have that $a_\ell = b_\ell$ because $\Obs$ is \visibleactionstable in $G$ (see the second part of Lemma~\ref{lem:actionStable}). It follows that $\Obs(h) = \Obs(g)$.
\end{proof}

From the previous lemma and Lemma~\ref{lem:playerStable}, we derive the \pstability of $\Obs'$.

\begin{corol} \label{cor:playerStable}
The observation function $\Obs'$ is \pstable for player~$P$.
\end{corol}

The observation function $\Obs'$ satisfies a stronger property\footnote{Note that the \stronglyPstability of $\Obs'$ implies its \pstability. 
} that we call \emph{\stronglyPstability}. This property will be useful in Section~\ref{subsec:eliminating}. It is proved in Appendix~\ref{app:StrongPStable}.

\begin{lemm}[\StronglyPstability of $\Obs'$] \label{lem:strongBis} 
The observation function $\Obs'$ is \emph{\StronlgyPstable} for player~$P$, that is, for every two histories $h' = h'_1v'$ and $g' = g'_1u'$ such that $\Obs'(h'_1) = \Obs'(g'_1)$, then 
the states $v',u'$ are controlled by the same player and $\Obs'(v') = \Obs'(u')$.
\end{lemm}

Both games $G$ and \pcG are \pstable (by hypothesis for the first one and by Corollary~\ref{cor:playerStable} for the second one). We can therefore study the observation-based strategies of $P_1$ in $\mathcal G$ and of $P$ in \pcG. The next theorem states how the game \pcG is equivalent to $\mathcal G$.

\begin{theo}[Equivalence theorem between $\mathcal G$ and \pcG] \label{thm:2Provers->1Prover}
    Let  \pcG be the \pc game associated with a three-player game $\mathcal G = (G,\Rabin((E_j,F_j)_{j\in J}),\Obs)$. Then, in $\mathcal G$, there exists an observation-based strategy $\sigma_{P_1}$ of $P_1$ such that for all strategies $\sigma_C$ of $C$, there exists a strategy $\sigma_{P_2}$ of $P_2$ such that $\langle \sigma_{P_1}, \sigma_C, \sigma_{P_2}\rangle_{v_0} \in \Rabin((E_j,F_j)_{j\in J})$ if, and only if, in \pcG, there exists an observation-based strategy $\sigma'_P$ of $P$ such that for all strategies $\sigma'_C$ of $C$ such that $\langle \sigma'_P, \sigma'_C\rangle_{v'_0} \in \Rabin((E_j,F_j)_{j\in J})$.
\end{theo}

Let us prove the first direction of Theorem~\ref{thm:2Provers->1Prover}.
\begin{lemm}
In $\mathcal G$, let $\sigma_{P_1}$ be an observation-based strategy of $P_1$ such that for all strategies $\sigma_C$ of $C$, there exists a strategy $\sigma_{P_2}$ of $P_2$ such that $\langle \sigma_{P_1}, \sigma_C, \sigma_{P_2}\rangle_{v_0} \in \Rabin((E_j,F_j)_{j\in J})$. Then in \pcG, there exists an observation-based $\sigma'_P$ of $P$ such that for all strategies $\sigma'_C$ of $C$, we have $\langle \sigma'_P, \sigma'_C\rangle_{v'_0} \in \Rabin((E_j,F_j)_{j\in J})$.
\end{lemm}

\begin{proof}
Let us consider an observation-based strategy $\sigma_{P_1}$ for player~$P_1$. From $\sigma_{P_1}$ and $G$, we define the infinite game structure $G[\sigma_{P_1}]$ enriched with observed histories, with action set $A$, in the following way:
\begin{itemize}
    \item its states are of the form $(v,\Obs(h))$ with $h \in \Hist{G}$ and $v$ the last state of $h$,
    \item its initial state is $(v_0, \Obs(v_0))$,
    \item for its transition function $\Delta$, we have $\Delta((v,\Obs(h)), a) = (u,\Obs(hau))$ such that $u = \delta(v,a)$, and whenever $v \in V_{P_1}$, $\Delta((v,\Obs(h)), a)$ is only defined for $a = \sigma_{P_1}(h)$.
\end{itemize}
Note that in $G[\sigma_{P_1}]$, the second components of the states derive from the first components seen along the plays. From the Rabin objective of $\mathcal G$, we define the Rabin objective $\Rabin((\bar E_j,\bar F_j)_{j \in J}$ of $G[\sigma_{P_1}]$  such that $\bar E_j = \{(v,\Obs(h)) \mid v \in E_j\}$ and $\bar F_j = \{(v,\Obs(h)) \mid v \in F_j\}$ for all $j \in J$.
Hence we have a two-player zero-sum game $(G[\sigma_{P_1}],\Rabin((\bar E_j,\bar F_j)_{j \in J})$, with players $C$ and $P_2$, such that strategies of those players in $G$ can be viewed as strategies in $G[\sigma_{P_1}]$.

Suppose that for all strategies $\sigma_C$ of $C$, there exists a strategy $\sigma_{P_2}$ of $P_2$ such that the play $\langle \sigma_{P_1}, \sigma_C, \sigma_{P_2}\rangle_{v_0}$ belongs to $\Rabin((E_j,F_j)_{j \in J}$ in $G$. With $\sigma_C$ and $\sigma_{P_2}$ seen as strategies in $G[\sigma_{P_1}]$, it holds that the play $\langle \sigma_C, \sigma_{P_2}\rangle_{(v_0,\Obs(v_0))} \in \Rabin((\bar E_j, \bar F_j)_{j \in J})$ in $G[\sigma_{P_1}]$. By determinacy of perfect-information (infinite) Rabin  games \cite{DBLP:journals/tcs/Zielonka98}, in $G[\sigma_{P_1}]$, there exists a strategy $\tau_{P_2}$ of~$P_2$ that is memoryless\footnote{In case of Rabin objective, we can suppose that winning strategies are memoryless.}, such that for all strategies $\tau_{C}$ of~$C$, we have $\langle \tau_C, \tau_{P_2} \rangle_{(v_0,\Obs(v_0))} \in \Rabin((\bar E_j, \bar F_j)_{j \in J})$ in $G[\sigma_{P_1}]$. The strategy $\tau_{P_2}$ is a partial\footnote{It is a partial function as it is only defined for pairs $(v,\Obs(h))$ such that $h$ ends with $v$.} function $\tau_{P_2} : V_{P_2} \times \Obs(\Hist{G}) \rightarrow A_{P_2}$. 

From strategies $\sigma_{P_1}$ and $\tau_{P_2}$, we are going to define a strategy $\sigma'_{P}$ for player~$P$ in $G'$ as follows. For every $h' \in \Histi{G'}{P}$ ending with state $v' \in V'_P$, let $h = \corresp{h'}$ be the history as defined in Lemma~\ref{lem:corresp}: 
\begin{itemize}
    \item if $v' \in V_{P_1}$, then $\sigma'_P(h') = \sigma_{P_1}(h)$,
    \item if $v' \in V_{P_2}$, then $\sigma'_P(h') = a' \in \ARondePd$ such that $a'(w) = \tau_{P_2}(w,\Obs(h))$ for all $w \in V_{P_2}$ such that $(w,\Obs(h))$ belongs to the domain of $\tau_{P_2}$.\footnote{Note that $(v',\Obs(h))$ belongs to the domain of $\tau_{P_2}$.}
\end{itemize}

Let us prove that $\sigma'_P$ is observation-based. Let $h',g' \in \Histi{G'}{P}$ be such that $\Obs'(h') = \Obs'(g')$. By Lemma~\ref{lem:corresp} with $h = \corresp{h'}, g = \corresp{g'} \in \Hist{G}$, we have that $\Obs(h) = \Obs(g)$. Let $v$ (resp. $u$) be the last state of $h$ (resp. $g$). As $\Obs$ is \pstable by hypothesis on $G$, if $v \in V_{P_1}$ (resp. $v \in V_{P_2}$), then $u \in V_{P_1}$ (resp. $u \in V_{P_2}$). Therefore
\begin{itemize}
    \item if $v,u \in V_{P_1}$, then $\sigma_{P_1}(h) = \sigma_{P_1}(g)$ as $\sigma_{P_1}$ is observation-based, and then $\sigma'_P(h') = \sigma'_P(g')$ by definition of $\sigma'_P$,
    \item if $v,u \in V_{P_2}$, then by definition of $\sigma'_P$, we have $\sigma'_P(h') = a'$ such that $a'(w) = \tau_{P_2}(w,\Obs(h))$ for all $(w,\Obs(h))$ in the domain of $\tau_{P_2}$. As $\Obs(h) = \Obs(g)$, it follows that $\sigma'_P(g') = a'$ with the same partial function $a'$.
\end{itemize}

It remains to prove that, with the defined strategy $\sigma'_P$, every play $\rho'$ compatible with $\sigma'_P$ belongs to $\Rabin((E_j,F_j)_{j\in J})$ in $G'$. Fix such a play $\rho' = v'_0a'_0v'_1a'_1  \ldots$. By Lemma~\ref{lem:corresp}, consider the play $\rho = \corresp{\rho'} = v_0a_0v_1a_1  \ldots$ (compatible with $\sigma_{P_1}$) in $G$ and the associated play $\bar \rho = (v_0,\Obs(v_0)) ~ a_0 ~ (v_1,\Obs(v_0a_0v_1)) ~a_1 ~ \ldots$ in $G[\sigma_{P_1}]$. By definition of $\sigma'_P$, the play $\bar \rho$ is compatible with $\tau_{P_2}$. Hence by definition of $\tau_{P_2}$, $\bar \rho$ belongs to $\Rabin((\bar E_j,\bar F_j)_{j\in J})$ in $G[\sigma_{P_1}]$ and thus $\rho$ belongs to $\Rabin((E_j,F_j)_{j\in J})$ in $G$. Consequently, we have $\rho' \in \Rabin((E_j,F_j)_{j\in J})$ in $G'$ as its sequence of states is the same as for $\rho$.
\end{proof}

Let us now prove the second direction of Theorem~\ref{thm:2Provers->1Prover}.

\begin{lemm}
In \pcG, let $\sigma'_P$ be an observation-based strategy of $P$ such that for all strategies $\sigma'_C$ of $C$, we have $\langle \sigma'_P, \sigma'_C\rangle_{v'_0} \in \Rabin((E_j,F_j)_{j\in J})$. Then in $\mathcal G$, there exists an observation-based strategy $\sigma_{P_1}$ of $P_1$ such that for all strategies $\sigma_C$ of $C$, there exists a strategy $\sigma_{P_2}$ of $P_2$ such that $\langle \sigma_{P_1}, \sigma_C, \sigma_{P_2}\rangle_{v_0} \in \Rabin((E_j,F_j)_{j\in J})$.
\end{lemm}

\begin{proof}
Let us consider an observation-based strategy $\sigma'_{P}$ for player~$P$ in \pcG. From $\sigma'_{P}$ and $G'$, we define the infinite game structure $H[\sigma'_{P}]$, with action set $A$, in the following way:
\begin{itemize}
    \item its states are of the form $\Obs'(h')$ with $h' \in \Hist{G'}$,
    \item its initial state is $\Obs'(v'_0)$,
    \item for its transition function $\Delta$, let $h' \in \Hist{G'}$ with $v'$ being its last state, 
    \begin{enumerate}
        \item if $v' \in V_{C}$, for all $a' \in A_C$, we define $\Delta(\Obs'(h'), a') = \Obs'(h'a'u')$ where $u' = \delta'(v',a')$,
        \item if $v' \in V_{P_1}$, we proceed similarly but with the action $a' = \sigma'_P(h') \in A_{P_1}$ only,
        \item if $v' \in V_{P_2}$, we again consider the only action $a' = \sigma'_P(h') \in \ARondePd$. In this case, let $u' = \delta'(v',a') = \delta(v',a'(v'))$. We then define $\Delta(\Obs'(h'), a'(v')) = \Obs'(h'a'u')$ for the action $a'(v') \in A_2$.\footnote{Recall that $H[\sigma'_{P}]$ has its action set equal to $A$ (and not $A'$).}
    \end{enumerate}
\end{itemize}
Notice that in Case~3 above, there may exist several outgoing transitions from state $\Obs'(h')$, depending on the image of the partial function $a' = \sigma'_P(h')$ and on the histories $g'$ such that $\Obs'(h') = \Obs'(g')$. 

Let us make two comments. First, let us explain why the transition function $\Delta$ is well-defined: 
\begin{itemize}
    \item Recall that if $\Obs'(h') = \Obs'(g')$, then $a' = \sigma'_P(h') = \sigma'_P(g')$ as $\sigma'_P$ is observation-based. That is, for $h'$ and $g'$, an identical action $a'$ is proposed in Cases~2 and~3 above.
    \item Assume by contradiction that there exist two histories $h',g' \in \Hist{G'}$ with $v'_1$ (resp. $v'_2$) being the last state of $h'$ (resp. $g'$), and an action $a \in A$ such that $\Obs'(h') = \Obs'(g')$ but $\Delta(\Obs'(h'),a) \neq \Delta(\Obs'(g'),a)$. In Cases~1 and~2 above, we have $a = a' \in A_C \cup A_{P_1}$, and in Case~3, for the action $a' = \sigma'_P(h') \in \ARondePd$, we have $a = a'(v'_1) = a'(v'_2)$. By the way $\delta'$ is defined from $\delta$, and as $\Obs$ is action-stable in $\mathcal G$ (first part of Lemma~\ref{lem:actionStable}), it follows that $\Obs'(h'a'u'_1) = \Obs'(g'a'u'_2)$ such that $u'_1 = \delta'(v'_1,a')$ and $u'_2 = \delta'(v'_2,a')$. This is in contradiction with $\Delta(\Obs'(h'),a) \neq \Delta(\Obs'(g'),a)$.
\end{itemize}

The following second comment derives from the definition of $H[\sigma'_P]$. If $h = v_0a_0v_1a_1 \ldots v_k \in \Hist{G}$ is a history in $G$ such that the sequence of actions $a_0a_1 \ldots a_{k-1}$ read from the initial state $\Obs'(v'_0)$ of $H[\sigma'_{P}]$ leads to a path in $H[\sigma'_{P}]$, then the last state of this path is $\Obs'(h')$ for some history $h'$ compatible with $\sigma'_P$ such that $\Obs(h) = \Obs(\corresp{h'})$. 

We are now going to define a strategy $\sigma_{P_1}$ for player~$P_1$ in $\mathcal G$ as follows. Let $h = v_0a_0v_1a_1 \ldots v_k \in \Histi{G}{P_1}$ ending with $v_k \in V_{P_1}$. By reading the sequence of actions $a_0a_1 \ldots a_{k-1}$ from the initial state of $H[\sigma'_{P}]$, one can follow a path in $H[\sigma'_{P}]$, if it exists. 
\begin{itemize}
    \item We begin with all histories $h$ such that such a path exists. Let $\Obs'(h')$ be the last state of this path (for some history $h'$ compatible with $\sigma'_P$). We define $\sigma_{P_1}(h) = a$ where $a$ is the label of the unique transition outgoing $\Obs'(h')$ in $H[\sigma'_P]$. Recall that $a = \sigma'_P(h')$ (see Case 2). We extend this definition to all histories $g$ such that $\Obs(g) = \Obs(h)$ by defining $\sigma_{P_1}(g) = a$. 
    \item For the remaining histories $h$ for which $\sigma_{P_1}(h)$ has not been defined yet, we define $\sigma_{P_1}(h)$ arbitrarily such that for all $g$ with $\Obs(g) = \Obs(h)$, then $\sigma_{P_1}(g) = \sigma_{P_1}(h)$.
\end{itemize}
By its definition, the strategy $\sigma_{P_1}$ is observation-based. Notice that it is well-defined. Indeed, we are going to show that if $h,g \in \Histi{G}{P_1}$ are two histories such that $\Obs(h) = \Obs(g)$ and there exists for each of them a path in $H[\sigma'_P]$ ending in states $\Obs'(h')$ and $\Obs'(g')$ respectively, then these states are equal. Assume by contradiction that $\Obs'(h') \neq \Obs'(g')$, and let $h' = v'_0a'_0v'_1a'_1 \ldots v'_k$, $g' = u'_0b'_0u'_1b'_1 \ldots u'_k$. As $\Obs(\corresp{h'}) = \Obs(h)$ and $\Obs(\corresp{g'}) = \Obs(g)$, we get that $\Obs(\corresp{h'}) = \Obs(\corresp{g'})$. Hence, knowing that $\Obs'(h') \neq \Obs'(g')$, by definition of $\Obs'$, there exists a smallest index $\ell < k$ such that $\Obs'(v'_0a'_0 \ldots v'_\ell) = \Obs'(u'_0b'_0 \ldots u'_\ell)$, $v'_\ell, u'_\ell \in V_{P_2}$, and $a'_{\ell} \neq b'_{\ell}$. 
By definition of $H[\sigma'_{P}]$, we have that $\sigma'_P(v'_0a'_0 \ldots v'_\ell) = a'_{\ell}$ and $\sigma'_P(u'_0b'_0 \ldots u'_\ell) = b'_{\ell}$. Moreover, as $\sigma'_P$ is observation-based, we get that  $a'_{\ell} = b'_{\ell}$, which is a contradiction. Hence, the strategy $\sigma_{P_1}$ is well-defined.

Now, suppose that in \pcG, for all strategies $\sigma'_C$ of $C$, the play $\langle \sigma'_P, \sigma'_C\rangle_{v'_0}$  belongs to $\Rabin((E_j,F_j)_{j \in J})$. Let $\sigma_C$ be a strategy of player~$C$ in $\mathcal G$. We have to prove that there exists a strategy $\sigma_{P_2}$ of player~$P_2$ such that $\rho = \langle \sigma_{P_1}, \sigma_C, \sigma_{P_2}\rangle_{v_0}$  belongs to $\Rabin((E_j,F_j)_{j \in J})$ in \pcpG.  We are going to define the play $\rho$ and some related play $\rho'$ in $G'$, prefix by prefix\footnote{Hence we partially define $\sigma_{P_2}$ to obtain $\rho$ (and not for all histories in $\Histi{G}{P_2}$).}, such that: 
\begin{itemize}
    \item the sequence of actions of $\rho$ read from the initial state of $H[\sigma'_{P}]$ leads to an infinite path in $H[\sigma'_{P}]$,
    \item the play $\rho'$ is compatible with $\sigma'_P$ and satisfies $\corresp{\rho'} = \rho$. 
\end{itemize} 
We begin with the prefix $h = v_0$ (resp. $h' = v_0$) for $\rho$ (resp. for $\rho'$) and the initial vertex $\Obs'(v_0)$ of $H[\sigma'_P]$. By induction, let $h \in \Hist{G}$ be a history compatible with $\sigma_{P_1}, \sigma_{C}$ and with $\sigma_{P_2}$ as already defined, and let $h' \in \Hist{G'}$ be the related history compatible with $\sigma'_P$ such that $\corresp{h'} = h$ ($h$ and $h'$ are the currently defined prefixes of $\rho$ and $\rho'$). By induction, let also $\Obs'(h')$ be the last state of the path in $H[\sigma'_{P}]$ induced by the reading of the actions of $h$ in $H[\sigma'_{P}]$. Let us explain how to extend $h, h'$ into two longer prefixes $h_1, h'_1$ satisfying the same properties as $h, h'$. Let $v'$ be the last state of $h'$ (that is also the last state of $h$). 
\begin{itemize}
    \item If $v' \in V_{P_1} \cup V_C$, then we define $h_1 = hau$ compatible with either $\sigma_{P_1}$ or $\sigma_C$, and we also consider $h'_1 = h'au$ such that $\corresp{h'au} = hau$. By definition of $H[\sigma'_{P}]$ and $\sigma_{P_1}$, 
    we have the transition $\Delta(\Obs'(h'),a) = \Obs'(h'au)$ in $H[\sigma'_{P}]$ and $h'au$ is compatible with $\sigma'_P$. 
    \item If $v' \in V_{P_2}$, let $a' = \sigma'_P(h')$ and $a = a'(v')$. Then, we define $\sigma_{P_2}(h) = a$. In this way we get the history $h_1 = hau$ such that $u = \delta(v',a)$ with a corresponding path in $H[\sigma'_{P}]$ ending with the state $\Obs'(h'a'u')$ with $u' = \delta'(v',a') = u$. We thus define $h'_1 = h'a'u'$. Notice that $\corresp{h'a'u'} = hau$.
\end{itemize}
 In this way, we finally get the play $\rho = \langle \sigma_{P_1}, \sigma_C, \sigma_{P_2}\rangle_{v_0}$ and the related play $\rho'$ such that $\corresp{\rho'} = \rho$ and $\rho'$ is compatible with $\sigma'_P$. As for all strategies $\sigma'_C$ of $C$ in \pcG, we have $\langle \sigma'_P, \sigma'_C\rangle_{v'_0} \in \Rabin((E_j,F_j)_{j \in J})$, in particular it follows that $\rho'$ belongs to $\Rabin((E_j,F_j)_{j \in J})$. This is also the case for $\rho$ as its sequence of states is the same as for $\rho'$.
\end{proof}

\subsection{Eliminating Imperfect Information} \label{subsec:eliminating}

In the previous section, we showed how to go from a three-player game with imperfect information to a two-player game with imperfect information. 
To solve this two-player game, we will get rid of the imperfect information to apply standard game-theoretic techniques.
In the literature \cite{lpar/ChatterjeeD10,DBLP:journals/lmcs/RaskinCDH07,DBLP:journals/jcss/Reif84}, this is usually done in two steps : first, make the objective \emph{visible}, that is, such that any two similarly observed plays agree on the winning condition, second, apply the \emph{subset construction} to recall the set of possible visited states, and letting them be observed. In this work, we proceed in only one step, in Definition~\ref{def:last_game}, by simultaneously modifying the game structure to both entail the subset construction on the states of the $\pc$ game, and the product with an automaton\footnote{We suppose that the reader is familiar with the concept of automaton with an $\omega$-regular acceptance condition like Streett or parity (see e.g. \cite{DBLP:conf/dagstuhl/2001automata}).} that monitors the winning condition along the plays (see Definition~\ref{def:parityAut}).

From the previous Section~\ref{subsec:2-pl-pr-g-ii}, we consider a two-player zero-sum game $\mathcal{G} = (G,\Rabin((E_j,F_j)_{j \in J}),\Obs)$\footnote{Again in this section, we get rid of the $'$ notation to ease readability.} with 
\begin{itemize}
    \item a game structure $G = (V, A, \lbrace P, C\rbrace, \delta, v_0)$, 
    \item $V = V_{P} \cup V_C$ and $A = A_{P} \cup A_C$ where $V_{P}, V_C$ are pairwise distinct sets as well as $A_{P}, A_C$,
    \item a Rabin objective $\Rabin((E_j,F_j)_{j \in J})$ for player~$P$, the opposite objective for player~$C$,
    \item an observation function $\Obs$ for player~$P$ such that 
    \begin{itemize}
        \item all actions $a \in A_P$ are visible,
        \item for all $a \in A_C$, $\Obs(a) = \sharp$,
        \item $\Obs$ is \StronlgyPstable (see Lemma~\ref{lem:strongBis}), and thus \pstable (see Corollary~\ref{cor:playerStable}).
    \end{itemize}
\end{itemize}
This game is fixed in this section.

We begin by defining a deterministic parity automaton that monitors the Rabin objective of the game $\mathcal G$ (see Definition~\ref{def:DetParityAut} below). This is made possible by complementing the following non-deterministic Streett automaton.

\begin{defi} \label{def:parityAut}
From the game $\mathcal G$, we define a \emph{non-deterministic} Streett automaton ${\automaton}' = (Q',q'_0,\alphabet,E',\Streett((E_j,F_j)_{j\in J})$ such that 
\begin{itemize}
    \item its set of states $Q'$ is equal to $V$ and its initial state $q'_0$ to $v_0$,
    \item its alphabet is equal to $\alphabet = \{\Obs(va) \mid v \in V, a \in A$, and $\delta(v,a)$ is defined$\}$
    \item its set $E'$ of edges is equal to $\{(v,o\bar{o},v') \mid o = \Obs(v), \bar{o} = \Obs(a)$ for some $a \in A$, and $v' = \delta(v,a)\}$
    \item its Streett acceptance condition $\Streett((E_j,F_j)_{j\in J})$ is the opposite of $\Rabin((E_j,F_j)_{j\in J})$. 
\end{itemize}
\end{defi}

Notice that each element $o \bar{o}$ of the alphabet $\alphabet$ is such that $\bar{o}$ is the observation of either a visible action of $A_P$ or an invisible action of $A_C$ (in which case $\bar{o} = \sharp)$. As the alphabet $\Sigma$ is defined from observations, ${\automaton}'$ is non-deterministic. 

\begin{lemm} \label{lem:AutParity}
A word $w \in \alphabet^\omega$ is accepted by ${\automaton}'$ if, and only if, there exists a play $\rho$ in $G$ such that $\Obs(\rho) = w$ and $\rho \not\in \Rabin((E_j,F_j)_{j\in J})$. 
\end{lemm}

\begin{proof}
We prove only one implication; the other one is proved similarly. Let $w = o_0 \bar{o}_0 o_1 \bar{o}_1 \ldots \in \alphabet^\omega$ be a word accepted by $\automaton'$. Then there exists a run $v_0 \xrightarrow{o_0 \bar{o}_0} v_1 \xrightarrow{o_1 \bar{o}_1} \ldots $ in $\automaton'$ that satisfies its Streett acceptance condition. By definition of $\automaton'$, there exists a play $\rho = v_0a_0v_1a_1 \ldots$ in $G$ using the same states as the run of $\mathcal{A}'$ labeled by $w$ and such that $\Obs(\rho) = w$. By definition of the Streett acceptance condition of $\automaton'$, we get that $\rho \not\in \Rabin((E_j,F_j)_{j\in J})$.
\end{proof}

It is known that any non-deterministic Streett automaton can be transformed into an equivalent deterministic parity automaton~\cite{DBLP:journals/lmcs/Piterman07}. The latter automaton is then easily complemented by replacing any of its priorities $d$ by $d+1$. These two constructions lead to the following automaton.

\begin{defi}[Monitoring automaton]\label{def:DetParityAut}
From the automaton ${\automaton}'$, we construct a \emph{deterministic} parity automaton ${\automaton} = (Q,q_0,\alphabet,E,\Parity(\beta))$ such that the language $\languageOf(\automaton)$ accepted by $\automaton$ is the opposite of the language accepted by $\automaton'$, i.e.,  $\languageOf(\automaton) = \alphabet^\omega \setminus \languageOf(\automaton')$.
\end{defi}

By Lemma~\ref{lem:AutParity}, we get the next corollary.

\begin{corol} \label{cor:automaton}
A word $w \in \alphabet^\omega$ is accepted by $\automaton$ if, and only if, for every play $\rho$ in $G$ such that $\Obs(\rho) = w$, we have $\rho \in \Rabin((E_j,F_j)_{j\in J})$.
\end{corol}

In the next lemma, we evaluate the size of the automaton $\automaton$ depending on the components of $\mathcal G$.

\begin{lemm} \label{lem:sizeParityAut}
Let $|V|$ and $|A|$ be the size of the game structure $G$ of $\mathcal G$, and $|J|$ be the size of its Rabin objective. Then the deterministic parity automaton $\automaton$ of Definition~\ref{def:DetParityAut} has 
\begin{itemize}
\item a size $|Q|$ exponential in $|V|$ and $|J|$, and $|\alphabet|$ linear in $|V|$ and $|A|$, 
\item a parity acceptance condition of size $|\beta|$ linear in $|V|$ and $|J|$. 
\end{itemize}
\end{lemm}

\begin{proof}
Consider the non-deterministic Streett of Definition~\ref{def:parityAut}. It has $|V|$ states, an alphabet size in $O(|V| |A|)$, and $|J|$ pairs for its Streett acceptance condition. Determinizing such an automaton and then complementing it into a deterministic parity automaton leads to an automaton with the same alphabet $\Sigma$, a number of states in $2^{O(|V||J| \log(|V||J|))}$, and a number of priorities in $O(|V|  |J|)$~\cite{DBLP:journals/lmcs/Piterman07}.
\end{proof}

We are now ready to define our end-point object, a two-player game with perfect information.
As mentioned at the beginning of this section~\ref{subsec:eliminating},
we proceed in one step to perform both the subset construction from the \pc game, and the product with the monitoring automaton from Definition~\ref{def:parityAut}. 
Note that the construction is \emph{deterministic} in both components, in the sense that, for any history $h$ in $\mathcal{G}$, there exists a unique corresponding history in the new game, and for any pair of action/state $a,v$ such that $hav$ is a valid history in $\mathcal{G}$, there is only one available action and successor state in the new game.

\begin{defi}[Two-player parity game with perfect information]\label{def:last_game}
From the game $\mathcal G$ and the parity automaton $\automaton$, we define a two-player zero-sum parity game $\mathcal{G}' = (G',\Parity(\beta'))$ with a game structure $G' = (V', A', \lbrace P, C\rbrace, \delta', v'_0)$ and a parity objective $\Parity(\beta')$ defined as follows:
\begin{itemize}
    \item $P$ and $C$ are the same two players,
    \item $V'$ is the set\footnote{For the sequel, we need to limit $V'$ to the states accessible from the initial state $v'_0$ (see Lemma~\ref{lem:coherent}).} of states $\{(U,q) \mid U \subseteq V, U \neq \varnothing, q \in Q\}$ such that $V' = V'_P \cup V'_C$ with $V'_P = \{(U,q) \in V' \mid U \subseteq V_P \}$ and $V'_C = \{(U,q) \in V' \mid U \subseteq V_C \}$, 
    \item $v'_0 = (\{v_0\},q_0)$ is the initial state,
    \item $A' = \alphabet$ is the set of actions, such that $A' = A'_P \cup A'_C$ with $A'_P = \{o \bar{o} \in \alphabet \mid \bar{o} \in A_P\}$ and $A'_C = \{o \bar{o} \in \alphabet \mid \bar{o} = \sharp \}$
    \item for the partial transition function $\delta'$, we have $\delta'((U,q),o \bar{o}) = (U',q')$ if
    \begin{itemize}
        \item $U' = \{ v' \in V \mid o = \Obs(v)$ for some $v \in U$, $\bar{o} = \Obs(a)$ for some $a \in A$, and $\delta(v,a) = v'\}$, and $U' \neq \varnothing$,
        \item $(q,o \bar{o},q') \in E$, 
    \end{itemize}
    \item the parity objective $\Parity(\beta')$ for player~$P$ uses the priority function $\beta'$ defined by $\beta'(U,q) = \beta(q)$ for all $(U,q) \in V'$.
\end{itemize}
\end{defi}

The new game $\mathcal{G}'$ has perfect information. Notice that $\delta'$ is a (partial) function since $\automaton$ is deterministic, and that it is deadlock-free. By definition, the sets $A'_P$ and $A'_C$ form a partition of $A'$. This is also the case  that $V'_P$ and $V'_C$ form a partition of $V'$, as shown in the next lemma. In addition, this lemma states that the states of $U$ when $(U,q) \in V'_P$ are all observed similarly by Prover in $\mathcal{G}$, thus making it safe to be perfectly observed by Prover in $\mathcal{G}'$.

\begin{lemm}[State-visibility in $\mathcal{G}'$] \label{lem:coherent}
For each state $(U,q)$ of $G'$ for which there exists a history in $\Hist{G'}$ whose last state is $(U,q)$, we have that either $U \subseteq V_P$ or $U \subseteq V_C$. Moreover, $\Obs(v) = \Obs(u)$ for all $v,u \in U$. 
\end{lemm}

\begin{proof}
Let $h' = v'_0o_0\bar{o}_0v'_1o_1\bar{o}_1 \ldots v'_k \in \Hist{G'}$ be a history such that $v'_k = (U,q)$. Let $v, u$ be two states in $U$. Then by definition of $G'$, from $h'$, one can construct two histories $h = h_1v$ and $g = g_1u$ in $\Hist{G}$ such that $\Obs(h_1) = \Obs(g_1) = o_0\bar{o}_0o_1\bar{o}_1 \ldots o_{k-1}\bar{o}_{k-1}$. As $\Obs$ is \StronlgyPstable (Lemma~\ref{lem:strongBis}), it follows that $v$ and $u$ are controlled by the same player and $\Obs(v) = \Obs(u)$. 
\end{proof}

Notice that the previous lemma motivates the restriction of $V'$ to states accessible from the initial state in Definition~\ref{def:last_game}. 

In the next lemma, we evaluate the size of the new game $\mathcal{G}'$.

\begin{lemm} \label{lem:sizeLastGame}
Let $|V|$ and $|A|$ be the size of the game structure $G$ of $\mathcal{G}$, and $|J|$ be the size of its Rabin objective. Then the  game $\mathcal{G}'$ has 
\begin{itemize}
\item a game structure $G'$ with size $|V'|$ exponential in $|V|$ and $|J|$, and $|A'|$ linear in $|V|$ and $|A|$, 
\item a parity objective $\Parity(\beta')$ with size $|\beta'|$ linear in $|V|$ and $|J|$. 
\end{itemize}
\end{lemm}

\begin{proof}
The proof of this lemma follows from the definition of $\mathcal{G}'$ and from Lemma~\ref{lem:sizeParityAut}. 
\end{proof}

In the next lemma, we further develop the relationship between histories (resp. plays) in the two games $\mathcal G$ and $\mathcal{G}'$.

\begin{lemm}
\label{lem:correspbis} 
\begin{enumerate} 
    \item Let $h = v_0a_0v_1a_1 \ldots v_k \in \Hist{G}$ be a history in $G$. Then it corresponds to $h$ a unique history $h' = v'_0o_0\bar{o}_0v'_1o_1\bar{o}_1 \ldots v'_k \in \Hist{G'}$, denoted by $h' = \correspbis{h}$, such that 
    \begin{itemize}
        \item $\Obs(h) = o_0\bar{o}_0 \ldots o_{k-1}\bar{o}_{k-1}\Obs(v_k)$, 
        \item $v'_k$ is equal to a unique state $(U,q)$ such that $v_k \in U$ and $q$ is the state of $\automaton$ reached by reading $o_0\bar{o}_0 \ldots o_{k-1}\bar{o}_{k-1}$ in $\automaton$ from $q_0$. 
    \end{itemize}
    \item Let $h, g \in \Hist{G}$ be two histories in $G$. If $\Obs(h) = \Obs(g)$, then $\correspbis{h} = \correspbis{g}$. Conversely, if $\correspbis{h} = \correspbis{g}$ and the last states of $h,g$ both belong to $V_P$, then $\Obs(h) = \Obs(g)$.
    \item Given a play $\rho = v_0a_0v_1a_1 \ldots \in \Plays{G}$, there exists a unique play $\rho' = v'_0o_0\bar{o}_0v'_1o_1\bar{o}_1 \ldots \in \Plays{G'}$, denoted by $\correspbis{\rho}$, such that $\Obs(\rho) = o_0\bar{o}_0o_1\bar{o}_1 \ldots$. 
    Moreover, $\correspbis{\rho}$ belongs to $\Parity(\beta')$ if, and only if,  $\Obs(\rho)$ belongs to $\languageOf(\automaton)$. \label{item:visibility}
\end{enumerate}
\end{lemm}

\begin{proof}
We begin with the first statement. Let $h = v_0a_0v_1a_1 \ldots v_k \in \Hist{G}$. Let us read the sequence of actions $o_0\bar{o}_0 \ldots o_{k-1}\bar{o}_{k-1} = \Obs(v_0a_0 \ldots v_{k-1}a_{k-1})$ from the initial state $v'_0$ in $G'$. By definition of $G'$, we thus obtain a unique history $h' = v'_0o_0\bar{o}_0v'_1o_1\bar{o}_1 \ldots v'_k \in \Hist{G'}$ whose last state $v'_k$ is of the form $(U,q)$ with $v_k \in U$ and $q$ being the state of $\automaton$ reached by reading $o_0\bar{o}_0 \ldots o_{k-1}\bar{o}_{k-1}$ in $\automaton$ from $q_0$.

Let us prove the second statement. Let $h = h_1v, g = g_1u \in \Hist{G}$. Suppose that $\Obs(g) = \Obs(h)$. From $g$ and $h$, the same sequence of actions $o_0\bar{o}_0 \ldots o_{k-1}\bar{o}_{k-1} = \Obs(h_1) = \Obs(g_1)$ is thus considered as in the first statement, showing that $\correspbis{h} = \correspbis{g}$. Conversely, suppose that $\correspbis{h} = \correspbis{g}$ and the last states $v,u$ of $h,g$ both belong to $V_P$. From $\correspbis{h} = \correspbis{g}$ and the first statement, we get $\Obs(h_1) = \Obs(g_1)$.  Moreover, from $v,u \in V_P$, by Lemma~\ref{lem:coherent}, we get that $\Obs(v) = \Obs(u)$. Therefore $\Obs(h) = \Obs(g)$.

Let $\rho = v_0a_0v_1a_1 \ldots \in \Plays{G}$. Let us read the sequence of actions $o_0\bar{o}_0 o_1 \bar{o}_1 \ldots = \Obs(\rho)$ from the initial state $v'_0$ in $G'$. By definition of $G'$, we thus obtain a unique play $\rho' = v'_0o_0\bar{o}_0v'_1o_1\bar{o}_1 \ldots \in \Plays{G'}$. Moreover, by definition of $\beta'$, the play $\rho'$ belongs to $\Parity(\beta')$ if, and only if, the word $\Obs(\rho)$ belongs to $\languageOf(\automaton)$.
\end{proof}

The next theorem states the equivalence between winning strategies of Prover in both games $\mathcal G$ and~$\mathcal{G}'$. 

\begin{theo}[Equivalence theorem between $\mathcal G$ and $\mathcal{G}'$]\label{thm:PerfectInfo}
In $\mathcal{G}$, there exists an observation-based strategy $\sigma_{P}$ of $P$ such that for all strategies $\sigma_C$ of $C$, we have  $\langle \sigma_P, \sigma_C\rangle_{v_0} \in \Rabin((E_j,F_j)_{j \in J})$ if, and only if, in $\mathcal{G}'$, there exists a strategy $\sigma'_{P}$ for $P$ such that for all strategies $\sigma'_C$ of $C$, we have $\langle \sigma'_P, \sigma'_C\rangle_{v'_0} \in \Parity(\beta')$.
\end{theo}  

\begin{proof}
We begin with the first direction of Theorem~\ref{thm:PerfectInfo}. Let us show how to define a strategy $\sigma'_P$ for player~$P$ in $\mathcal{G}'$ from the given observation-based strategy $\sigma_P$ in $\mathcal{G}$. Let $h' \in \Histi{G'}{P}$ ending with a state $v' = (U,q) \in V'_P$. Recall that $U \subseteq V_P$. Take some $v \in U$. By definition of $G'$, from $h'$ and $v$, we can construct a history $h \in \Hist{G}$ such that $\correspbis{h} = h'$ and whose last state is $v$. Notice that such a history $h$ is not unique since we could take another state $u \neq v$ in $U$. However, by Lemma~\ref{lem:correspbis} (second statement) and as $U \subseteq V_P$, for all $g$ such that $\correspbis{g} = \correspbis{h} = h'$, we have $\Obs(g) = \Obs(h)$. Let $a = \sigma_P(h)$, we then define $\sigma'_P(h') = o \bar{o}$ such that $o = \Obs(v)$ and $\bar{o} = \Obs(a) = a$. This strategy $\sigma'_P$ is well-defined since $\sigma_P$ is observation-based.

By hypothesis on $\sigma_P$, we have $\langle \sigma_P, \sigma_C\rangle_{v_0} \in \Rabin((E_j,F_j)_{j \in J})$ for all strategies $\sigma_C$ of $C$. Let us now prove that in $\mathcal{G}'$, for all strategies $\sigma'_C$ of $C$, the play $\langle \sigma'_P, \sigma'_C\rangle_{v'_0}$ belongs to $\Parity(\beta')$. Consider such a play $\rho'$ compatible with $\sigma'_P$. We consider the infinite tree rooted in $v_0$ and whose branches are all histories $h \in \Hist{G}$ such that $\correspbis{h}$ is prefix of $\rho'$. All those histories are compatible with $\sigma_P$ by definition of $\sigma'_P$. By K\"onig's lemma, there exists an infinite branch corresponding to a play $\rho \in \Plays{G}$ that is compatible with $\sigma_P$. Hence $\rho \in \Rabin((E_j,F_j)_{j \in J})$ and $\correspbis{\rho} = \rho'$. Notice that for all $\xi$ such that $\Obs(\rho) = \Obs(\xi)$, $\xi$ also belongs to $\Rabin((E_j,F_j)_{j \in J})$, as $\sigma_P$ is observation-based. By Corollary~\ref{cor:automaton}, it follows that the word $w = \Obs(\rho)$ is accepted by $\automaton$. Hence, by Lemma~\ref{lem:correspbis} (third statement), $\rho' \in \Parity(\beta')$.

We now turn to the second direction of Theorem~\ref{thm:PerfectInfo}.
Let us show how to define a observation-based strategy $\sigma_P$ for player~$P$ in $\mathcal G$ from the given strategy $\sigma'_P$ in $\mathcal{G}'$. Let $h \in \Histi{G}{P}$ ending with a state $v \in V_P$. By Lemma~\ref{lem:correspbis} (first statement), consider the corresponding history $\correspbis{h}$ ending with a state $(U,q)$ such that $v \in U$. Hence, we have $(U,q) \in V'_P$. Let $o \bar{o} = \sigma'_P(\correspbis{h})$. Recall (see Definition~\ref{def:parityAut}) that $\bar{o} = a$ for some visible action $a \in A_P$. We then define $\sigma_P(h) = a$. This strategy $\sigma_P$ is observation-based. Indeed, let $g \in \Hist{G}$ be such that $\Obs(g) = \Obs(h)$. Then, by Lemma~\ref{lem:correspbis} (second statement), we get $\correspbis{g} = \correspbis{h}$ which implies that  
$\sigma_P(g) = \bar o = a = \sigma_P(h)$.

By hypothesis on $\sigma'_P$, we have $\langle \sigma'_P, \sigma'_C\rangle_{v'_0} \in \Parity(\beta')$ for all strategies $\sigma'_C$ of $C$. Let us prove that in $\mathcal G$, for all strategies $\sigma_C$ of $C$, the play $\langle \sigma_P, \sigma_C\rangle_{v_0}$ belongs to $\Rabin((E_j,F_j)_{j \in J})$. Consider such a play $\rho$ compatible with $\sigma_P$. Then by definition of $\sigma_P$, the play $\correspbis{\rho}$ is compatible with $\sigma'_P$ and therefore satisfies $\Parity(\beta')$. Hence, by Lemma~\ref{lem:correspbis} (third statement), $\Obs(\rho)$ is accepted by $\automaton$. Therefore, by Corollary~\ref{cor:automaton}, $\rho \in \Rabin((E_j,F_j)_{j \in J})$. 
\end{proof}

\section{Complexity Analysis} \label{sec:complexity}

In this section, we finally prove the main theorem of our paper stating the complexity of solving the \NCRS problem for parity games (see Theorem~\ref{thm:main}). We also establish the complexity of this problem for reachability games when the number of players and the number of actions is fixed (see Theorem~\ref{thm:mainReach}).

\subsection{Complexity Analysis for Parity Games} 

The proof of the upper bound of Theorem~\ref{thm:main}, first part, aggregates the intermediary results of Proposition~\ref{prop:PCPParityRabin}, and Lemmas~\ref{lem:sizePCgame} and~\ref{lem:sizeLastGame}, on the sizes of the objects we successively constructed, and articulates them with the complexity of solving two-player zero-sum parity games~\cite{DBLP:journals/siamcomp/CaludeJKLS22}. For the PSpace lower bound, the QBF-reduction of~\cite{DBLP:conf/icalp/ConduracheFGR16} (Theorem 7) applies here; all NE responses to $\sigma_0$ in that proof are also SPEs.
Concerning the second part of Theorem~\ref{thm:main} (devoted to the case where the number of players is fixed), the complexity upper bound is obtained as a direct corollary of the first part of Theorem~\ref{thm:main}, and the complexity lower bounds are proved in Appendix~\ref{app:lowerBounds}.

\begin{proof}[Proof of Theorem~\ref{thm:main}, first part, upper bound.]
Figure~\ref{fig_structure} details the different steps for solving the \NCRS problem for reachability and parity games. It finally reduces to solve a two-player zero-sum parity game with perfect information. Solving such a game is in time $n^{O(\log(d))}$ where $n$ is its number of vertices and $d$ its number of priorities~\cite{DBLP:journals/siamcomp/CaludeJKLS22}. For each step of Figure~\ref{fig_structure} where one game is transformed into another game, we stated the time complexity of constructing the new game from the given game (see Proposition~\ref{prop:PCPParityRabin}, and Lemmas~\ref{lem:sizePCgame} and \ref{lem:sizeLastGame}). To prove Theorem~\ref{thm:main}, first part, it thus remains to compute the total complexity from the complexity results of the intermediate steps. Let $\mathcal{G}_0 = (G,(W_i)_{i \in \Pi})$ be a multi-player 
parity game, with a game structure with size $|V|, |A|$, and $|\Pi|$. As the $W_i$'s are parity objectives, we denote by $\alpha_i$ their priority function (with size $|\alpha_i|$). We describe hereafter the complexity of the current step in terms of the size of the \emph{initial} game $\mathcal{G}_0$.
\begin{itemize}
    \item \emph{Constructing the \pcp game with a Rabin objective for the two Provers:} 
     by Proposition~\ref{prop:PCPParityRabin}, the \pcp game $(G',\Rabin((E_j,F_j)_{j \in J}),\Obs)$ has a size $|V'|$ polynomial in $|V|, |A|$ and exponential in $|\Pi|$, a size $|A'|$ polynomial in $|A|$ and exponential in $|\Pi|$, and a size $|J|$ polynomial in $|\Pi|,|\alpha_i|, \forall i$.
    
    \item \emph{Constructing the game \pc game with a Rabin objective for Prover:}
      by Lemma~\ref{lem:sizePCgame} and the previous item, the \pc game $(G',\Rabin((E_j,F_j)_{j \in J}),\Obs)$ has a size $|V'|$ polynomial in $|V|, |A|$ and exponential in $|\Pi|$, a size $|A'|$ exponential in $|V|,|A|$ and double-exponential in $|\Pi|$, and a size $|J|$ polynomial in $|\Pi|,|\alpha_i|, \forall i$.
   
    \item \emph{Constructing the two-player parity game with perfect information:}
     by Lemma~\ref{lem:sizeLastGame} and the previous item, the parity game $(G',\Parity(\beta))$ has a size $|V'|$ exponential in $|V|, |A|$, $|\alpha_i|, \forall i$, and double-exponential in $|\Pi|$, a size $|A'|$ exponential in $|V|,|A|$ and double-exponential in $|\Pi|$, and a size $|\beta|$ polynomial in $|V|,|A|,|\alpha_i|, \forall i$ and exponential in $|\Pi|$.
   
    \item \emph{Solving this parity game:}
      by \cite{DBLP:journals/siamcomp/CaludeJKLS22}, solving the ultimate parity game, or equivalently the \NCRS problem, is in time exponential in $|V|,|A|,|\alpha_i|, \forall i$, and double-exponential in $|\Pi|$. As the game structure of $\mathcal{G}_0$ is action-unique (see Definition~\ref{def:game_structure}), it follows that \Ve{$|A| \leq |V|^2$}, thus leading to an algorithm in time exponential in $|V|$ and each $|\alpha_i|$, and double-exponential in $|\Pi|$.
   
\end{itemize}
\end{proof}

\begin{rema}
Suppose that our algorithm establishes the existence of a solution $\sigma_0$ to the \NCRS problem for parity games. As it is obtained from a memoryless winning strategy in the \emph{final} parity zero-sum game, we get a finite-memory solution $\sigma_0$ whose memory size is exponential when the number of players is fixed, doubly exponential otherwise. For a lower bound on the memory required, it is fairly straightforward to show that it requires exponential memory by reducing a Streett zero-sum game to our problem.
\end{rema}

\subsection{Complexity Analysis for Reachability Games} \label{sec:reach_case}

In the previous sections, we have shown a general solution to solve the problem \NCRS, since parity objectives subsume all $\omega$ regular objectives. In this section, for the particular  case of reachability games, we show that a careful analysis of this simpler case, when \emph{the number of players 
is fixed}, leads to a more fine-tuned solution and a better complexity result: Theorem~\ref{thm:mainReach} states a polynomial complexity instead of the exponential complexity of Theorem~\ref{thm:main}, second part.

This approach follows the same steps as presented in Section~\ref{sec:SolvingPCPGame}, but with a few adjustments. The key idea is that monitoring reachability objectives is simpler than parity ones, which means that the synchronized product (described in the proof of Proposition~\ref{prop:PCPParityRabin}) can omit some information from the original game structure.
Indeed, when monitoring a play in the \pcp game that stems from an original reachability game, one still needs to keep track of the actions of both Provers and Challenger, but the checking of the gain component is simplified compared to parity: only the information of whether each player has already visited his target set is to be remembered. 

This can be exploited to avoid the exponential blowup of the state space in the subset construction phase (with respect to the original state set $V$) to get rid of imperfect information (see Section~\ref{subsec:eliminating}) and thus obtain a better complexity, as we detail in the following subsections. 


In the sequel, we therefore suppose that $\mathcal{G}$ is a reachability game such that $|\Pi|$ is fixed. We can also suppose that $|A|$ is \emph{constant}, leading to the size of $\mathcal{G}$ depending only on $|V|$. Indeed, it can be easily shown that any game structure of a reachability game with $|V|$ vertices and $|A|$ actions can be transformed into one with $\mathcal{O}(|V|^2)$ vertices and two actions such that Player~$0$ has a solution to the \NCRS problem in the initial game if, and only if, he has a solution in the modified game. Such a similar construction can be found in the Appendix of \cite{DBLP:journals/tocl/BruyereFRT24} (the idea is to replace the outgoing transitions of any vertex $v$ by a complete binary tree whose $v$ is the root).

\subsubsection{The \texorpdfstring{\pcp} Game as a Rabin Game} 
\label{subsec:construction_solving_reach}
In this section, we come back to the construction of Section~\ref{subsec:pr-g-parity} and describe the adjustment made to define the \pcp game of a reachability game as a three-player Rabin game.

In Section~\ref{subsec:pr-g-parity}, we defined the game \pcpOG game for a parity game $\mathcal G$, which involves the construction of an observer automaton $\mathcal{O}$ monitoring the states of \pcpG as the plays progressed. In the reachability case, we make some \emph{adaptations} on this observer automaton: the states of $\mathcal O$ keep track of which player has already visited his target set, these states include additional information in a way to define the Rabin objective (translating the objective $W_P$ for the two Provers) directly on $\mathcal O$ and not on \pcpOG, and the alphabet of $\mathcal O$ includes both states and actions of the \pcp (for the need of future constructions).

The next proposition is the counterpart of Proposition~\ref{prop:PCPParityRabin}.

\begin{prop} \label{prop:PCPReachRabin}
Let $\mathcal{G}$ be a reachability game and $|V|$ be the size of its game structure $G$ (the sizes $|A|$ and $|\Pi|$ are constants).
Let $G'$ be the game structure of \pcpG (according to Section~\ref{subsec:DefPCPgame}). Then, 
\begin{itemize}
\item the game \pcpOG uses an observer automaton $\mathcal O = (Q,q_0,\Sigma,E,\Rabin((E_j,F_j)_{j \in J}))$ such that $|Q|$ and $|J|$ are constant, and $|\Sigma|$ is linear in $|V|$,
\item the game structure $G' \times \ObsAut$ of \pcpOG has a set of states of size linear in $|V|$, a constant set of actions and a constant number of Rabin pairs. 
\end{itemize}
\end{prop}

Let us first introduce a notation.  Given a state $v'$ of $G'$, we call \emph{\NegComponent} the tuple obtained from $v'$ by deleting its $G$-component (for the initial state $v'_0$, its \NegComponent is empty, denoted by $\varnothing$). The set of \NegComponents is denoted $\NegG$. Note that this set is of constant size as the \NegComponents only depend on $A$ and $\Pi$ which have constant size. In the synchronized product $G' \times \mathcal{O}$, the states are of the form $(v',q)$ where $v'$ is a state of $G'$ and $q$ is a state of $\mathcal{O}$ called the \emph{\OComponent} of $(v',q)$.

\begin{proof}[Proof of Proposition~\ref{prop:PCPReachRabin}]
We suppose that each reachability objective $W_i$, $i \in \Pi$, of the reachability game $\mathcal G$ uses the target set $T_i \subseteq V$. For the game structure $G'$ of the \pcpG, we use the notations of Definitions~\ref{def:pcp-game-structure-states}-\ref{def:pcp-game-structure-transitions}.

We are going to construct the observer automaton as a deterministic Rabin automaton $\ObsAut = (Q,q_0,\Sigma,E,\Rabin((E_j,F_j)_{j \in J}))$ that observes the states and actions of \pcpG and whose Rabin condition will encode the objective $W_P$ of the Provers. The automaton $\ObsAut$ is defined as follows (Figure~\ref{fig:prover-game-play} should again be helpful): 

\begin{itemize}
    \item The set $Q$ of states is composed of the tuples $(\nu,\bar{t},j,d,f) \in \NegG \times\{0,1\}^{|\Pi|} \times (\Pi \setminus \{0\} \cup \{\varnothing\}) \times (\Pi \setminus \{0\}) \times \{0,1\}$. In addition to the \NegComponent $\nu$, each component $t_i$ of $\bar t$ indicates whether or not $T_i$ has already been visited, $j$ is the last player seen in a player-state (including $j = \varnothing$), $d$ is the last deviating player (with $d \neq \varnothing$), and $f$ oscillates between $0$ and $1$ ($f=1$ meaning that the last two deviating players are different). 
    \item The initial state $q_0$ is the tuple $(\varnothing,\bar{0},0,0,0)$. (As this state will be synchronized with the initial state $v'_0$ of $G'$, we can arbitrarily choose the values of the tuple, except for $\nu = \varnothing$ being the \NegComponent of $v'_0$ and $\bar{t} = \bar{0}$ that indicates that no target set $T_i$ has been visited yet).
    \item The alphabet $\Sigma$ is equal to $\{v'a' \mid v' \in V', a' \in A', \text{ and } \delta'(v',a') \text{ is defined} \}$.
    \item The set $E$ of transitions of $\ObsAut$ if composed of triples $((\nu,\bar{t},j,d,f),v'a',(\mu,\bar{t}',j',d',f'))$ such that: let $u' = \delta'(v',a')$,
    \begin{itemize}
        \item $\nu$ is the \NegComponent of $v'$ and $\mu$ is the \NegComponent of $u'$,
        \item let $u$ be the $G$-component of $u'$; for all $i \in \Pi$, if $u \in T_i$, then $t'_i = 1$, otherwise $t'_i = t_i$,
        \item if $u'$ is not a player-state, then $(j',d',f') = (j,d,f)$.
        \item if $u' = (u,i,\bar{g})$ is a player-state with $i = \varnothing$, then $j' = \varnothing$ and $(d',f') = (d,f)$,
        \item if $u' = (u,i,\bar{g})$ is a player-state with $i \neq \varnothing$, then $j' = d' = i$ and $f' = 1 \Leftrightarrow d' \neq d$,
    \end{itemize} 
\end{itemize}

Let us now describe the Rabin condition $\Rabin((E_j,F_j)_{j \in J})$ of $\ObsAut$. It is composed of the following pairs that we explain below: 
\begin{itemize}
\item a pair $(E_0,F_0)$ such that $E_0 = \{(\nu,\bar{t},j,d,f) \mid j \neq \varnothing\}$ and $F_0 = \{(\nu,\bar{t},j,d,f) \mid \bar{g} \neq \bar{t} \text{ where $\bar g$ is the gain-component of $\nu$}\}$,
\item for each $j \in \Pi \setminus \{0\}$, a pair $(E_j,F_j)$ such that $E_j = \{(\nu,\bar{t},j,d,f) \mid f = 1 \}$ and $F_j = \{(\nu,\bar{t},j,d,f) \mid g_j < t_j \text{ where $\bar g$ is the gain-component of $\nu$}\}$.  
\end{itemize}

The pair $(E_0,F_0)$ corresponds to condition $(iP)$  of the Provers. Indeed, with $E_0$, we impose that eventually, no player deviates, and with $F_0$, as $\bar{g}$ and $\bar{t}$ eventually stabilize, we impose that the gain $\bar{g}$ predicted by $C$ is not correct. Given $j \in \Pi \setminus \{0\}$, with the pair $(E_j,F_j)$, we impose with $E_j$ that there is at most one infinitely deviating player, and with $F_j$, that $j \in \Pi$ is the deviating player and that this deviation is profitable for $j$ (as $g_j = 0 < 1 = t_j$).

Notice that $\ObsAut$ has a constant number $|Q|$ of states, a constant number $|J|$ of Rabin pairs, and an alphabet size $|\Sigma|$ linear in $|V|$ (see Definitions~\ref{def:pcp-game-structure-states}-\ref{def:pcp-game-structure-transitions} with $A$ and $\Pi$ of constant size). 

\begin{figure} 
\begin{center}
	\begin{tikzpicture}[->, >=latex,shorten >=1pt, scale=0.7, every node/.style={scale=1, align=center}]

    \node[] (game) at (-2, 11) {$G'$};
	\node[draw,circle, dotted, minimum height = 0.8cm, minimum width = 0.8cm] (v'0) at (1, 11) {$v'_0$};
	\node[draw, circle, dotted] (v'1) at (5, 11) {$v'_1$};
	\node[draw, circle, dotted, minimum height = 0.8cm] (v'2) at (9, 11) {$v'_2$};
    \node[] (s) at (13, 11) {$\dots$};

	\path[->] (v'0) edge node[above] {$a'_0$} (v'1);
	\path[->] (v'1) edge node[above] {$a'_1$} (v'2);
	\path[->] (v'2) edge node[above] {$a'_2$} (s);

    \node[] (obs) at (-2, 9) {$\mathcal O$};
	\node[draw, circle, dotted,minimum height = 0.8cm, minimum width = 0.8cm] (q0) at (1, 9) {$q_0$};
	\node[draw, circle, dotted] (q1) at (5, 9) {$q_1$};
	\node[draw, circle, dotted, minimum height = 0.8cm] (q2) at (9, 9) {$q_2$};
    \node[] (s') at (13, 9) {$\dots$};

	\path[->] (q0) edge node[above] {$v'_0a'_0$} (q1);
	\path[->] (q1) edge node[above] {$v'_1a'_1$} (q2);
	\path[->] (q2) edge node[above] {$v'_2a'_2$} (s');

    \node[] (gameobs) at (-2, 7) {$G' \times {\mathcal O}$};
	\node[draw,circle, dotted,minimum height = 0.8cm, minimum width = 0.8cm] (vq0) at (1, 7) {$(v'_0,q_0)$};
	\node[draw, circle, dotted] (vq1) at (5, 7) {$(v'_1,q_1)$};
	\node[draw, circle, dotted, minimum height = 0.8cm] (vq2) at (9, 7) {$(v'_2,q_2)$};
    \node[] (vs') at (13, 7) {$\dots$};

	\path[->] (vq0) edge node[above] {$a'_0$} (vq1);
	\path[->] (vq1) edge node[above] {$a'_1$} (vq2);
	\path[->] (vq2) edge node[above] {$a'_2$} (vs');

	\end{tikzpicture}
\end{center}
\caption{The synchronized product of $G'$ and $\ObsAut$}
\label{fig:sync}
\end{figure}
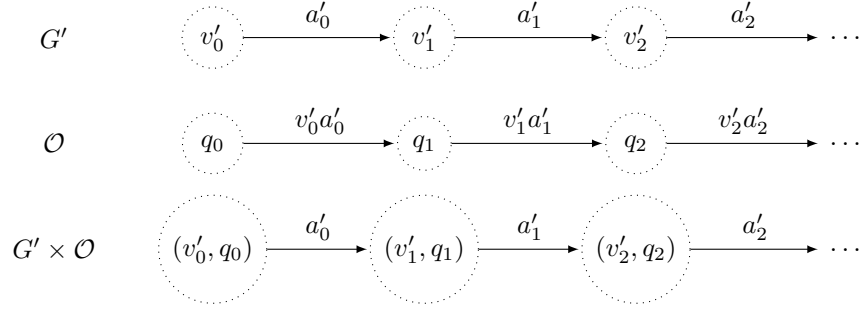

We then construct the synchronized product of the game structure $G'$ of \pcpG with the observer automaton $\ObsAut$ (see Figure~\ref{fig:sync}). The resulting game structure is composed of the states $(v',\bar{t},j,d,f)$ resulting from the states $v'$ of $G'$ extended with the observed information $(\nu,\bar{t},j,d,f)$ of $\ObsAut$.\footnote{As the \NegComponent $\nu$ of $v'$ is already part of $v'$, it is not duplicated in the synchronized product.} 
By definition of the observer $\mathcal O$ and by Definitions~\ref{def:pcp-game-structure-states}-\ref{def:pcp-game-structure-transitions}, the set of states of \pcpOG has a size linear in $|V|$, a set of actions of constant size, and a constant number of Rabin pairs. 

Finally, the observation function of the extended states is naturally defined such that the observation of $(v',\bar{t},j,d,f)$ is equal to $\Obs(v')$, while we keep the same observation function for the actions. The resulting function is still denoted by $\Obs$. The Rabin condition is also naturally extended to the states of the synchronized product and still denoted $\Rabin((E_j,F_j)_{i \in J})$. In this way, given a reachability game $\mathcal{G} = (G,(W_i)_{j \in \Pi})$, we get the \pcpOG $= (G'\times \mathcal{O},\Rabin((E_j,F_j)_{i \in J}),\Obs)$ with a Rabin objective for the Provers. 
\end{proof}


From the player-stability of $\Obs$ (see Lemma~\ref{lem:playerStable}), we get the next technical lemma that will be useful later in the solving process. Figure~\ref{fig:sync} could be helpful to understand this lemma and its proof given in Appendix~\ref{App:dernierlem}. 

\begin{lemm} \label{lem:crucial}
Let $\rho = (v'_0,q_0)a'_0(v'_1,q_1)a'_1 \ldots$ be a play in the game structure $G' \times \mathcal{O}$ and
$$p_0 \xrightarrow{u'_0b'_0} p_1 \xrightarrow{u'_1b'_1} \ldots$$ be an infinite path in the observer automaton $\ObsAut$ with label $\pi = u'_0b'_0u'_1b'_1 \ldots$. If $\Obs(\rho) = \Obs(\pi)$, then 
$(u'_0,p_0)b'_0(u'_1,p_1)b'_1 \ldots$
is also a play in $G' \times \mathcal{O}$.   
\end{lemm}

\subsubsection{Solving the game \pcpOG}

Now let us move on to solving the game \pcpOG. We will proceed as in the parity case in Section~\ref{sec:SolvingPCPGame}.

First, we start by \emph{merging the two provers} into one to get a two-player Rabin game with imperfect information, similarly to the general case (see Section~\ref{subsec:2-pl-pr-g-ii} for a detailed account).
Let us simply remark that, while in the parity case we start from a somewhat generic three-player Rabin game, for the reachability game, it is good to keep in mind that this Rabin game stems from specific construction. Namely, the synchronized product of Proposition~\ref{prop:PCPReachRabin} between an observer automaton and a game structure described with a Rabin condition that only concerns the observer automaton. The resulting two-player Rabin game is denoted by \pcOG as it is still seen as the synchronized product of the \pcp game and its observer automaton $\mathcal O$ where the actions of $A_{P_1}$ and $A_{P_2}$ of Provers $P_1$ and $P_2$ have been \emph{replaced by the actions of $A_{P_1} \cup \ARondePd$} of Prover $P$ (see Definition~\ref{def:OneProver}). In the sequel, we keep the same notation $\mathcal O$ for this modified observer automaton. Notice that Lemma~\ref{lem:crucial} still holds when applied to the game \pcOG.

Second, we show how to \emph{get rid of the imperfect information} of the game \pcOG. We proceed as in the general case (see Section~\ref{subsec:eliminating}), however with some adaptations. The first one is the construction of a Streett automaton, which builds directly on the observer automaton $\mathcal O$ for the reachability case, unlike the one constructed from \pcOG for the parity case (see Definition~\ref{def:parityAut}).

\begin{defi} 
From a reachability game and its (modified) observer $\ObsAut = (Q,q_0,\alphabet,E,\Rabin((E_j,F_j)_{j\in J})$, we define a non-deterministic Streett automaton ${\automaton}' = (Q,q_0,\alphabet',E',\Streett((E_j,F_j)_{j\in J})$ such that 
\begin{itemize}
    \item its set of states is equal to $Q$ and its initial state to $q_0$,
    \item its alphabet is equal to $\alphabet' = \{\Obs(va) \mid v a \in \Sigma\}$
    \item its set $E'$ of edges is equal to $\{(v,o\bar{o},v') \mid o \bar{o} = \Obs(va)$ for some $va \in \Sigma$, and $(q,va,q') \in E \}$
    \item its Streett acceptance condition $\Streett((E_j,F_j)_{j\in J})$ is the opposite of $\Rabin((E_j,F_j)_{j\in J})$. 
\end{itemize}
\end{defi}

We can then continue as in Section~\ref{subsec:eliminating}, and complement this Streett automaton ${\automaton}'$ to obtain a deterministic parity automaton $\automaton$ (as in Definition~\ref{def:DetParityAut}). The automaton $\automaton$ satisfies the following two properties (the second one corresponds to Corollary~\ref{cor:automaton} in the parity case).

\begin{lemm} \label{cor:automatonBis} 
\begin{itemize}
    \item The deterministic parity automaton $\automaton$ has a constant number of states and a parity acceptance condition of constant size. Its alphabet size is linear in $|V|$.
    \item Let $w = \Obs(\rho) \in \alphabet^\omega$ be the observation of a play $\rho$ in \pcOG. Then, $w$ is accepted by the parity automaton $\automaton$ if, and only if, for every play $\pi$ in \pcOG such that $\Obs(\pi) = w$, we have $\pi \in \Rabin((E_j,F_j)_{j\in J})$.
\end{itemize}
\end{lemm}

\begin{proof}
The first part of the lemma is proved as for Lemma~\ref{lem:sizeParityAut} given the size of $\mathcal O$ given in Proposition~\ref{prop:PCPReachRabin}. 

For the second part of the lemma, we prove the following equivalent property on the Streett automaton~$\automaton'$: Let $w = \Obs(\rho) \in \alphabet^\omega$ be  the observation of a play $\rho$ in \pcOG. Then, $w$ is accepted by ${\automaton}'$ if, and only if, there exists a play $\pi$ in \pcOG such that $\Obs(\pi) = w$ and $\pi \not\in \Rabin((E_j,F_j)_{j\in J})$.

Suppose that $w = o_0 \bar{o}_0 o_1 \bar{o}_1 \ldots$ is accepted by $\automaton'$. Then there exists a run $q_0 \xrightarrow{o_0 \bar{o}_0} q_1 \xrightarrow{o_1 \bar{o}_1} \ldots $ in $\automaton'$ that satisfies its Streett acceptance condition. By definition of $\automaton'$, there exists an infinite path $q_0 \xrightarrow{v_0 a_0} q_1 \xrightarrow{v_1 a_1} \ldots $ in $\ObsAut$ using the same states as the run of $\mathcal{A}'$ such that $\Obs(v_0a_0v_1a_1 \ldots) = w$. By Lemma~\ref{lem:crucial}, as $w = \Obs(\rho)$, it follows that $\pi = (v_0,q_0)a_0 (v_1,q_1) a_1 \ldots$ is a play in \pcOG. By definition of the Streett acceptance condition of $\automaton'$, we get that $\pi \not\in \Rabin((E_j,F_j)_{j\in J})$. The other implication is easily proved from the definition of $\automaton'$. 
\end{proof}

Finally, we can finish the solving process and provide a proof of Theorem~\ref{thm:mainReach}.

\begin{proof}[Proof of Theorem~\ref{thm:mainReach}]
The last step is to construct a two-player parity game $\mathcal{G}'$ with perfect information as in Definition~\ref{def:last_game}. Let us study the number of states of $\mathcal{G}'$ and the size of its parity objective. 

The parity objective is of constant size as it is defined from the parity acceptance condition of $\automaton$ (see Lemma~\ref{cor:automatonBis}).

Recall the form $(U,q)$ of the states of $\mathcal{G}'$, where $U$ is a subset of states of \pcOG and $q$ is a state of $\automaton$. Recall also that by Lemma~\ref{lem:coherent}, all states of $U$ are observed by $\Obs$ as the same state $v \in V$ of the original game. Thus, the subset construction used in Definition~\ref{def:last_game} does not involve the set $V$, but only on the \NegComponents and the \OComponents of \pcOG. As those components do not depend on $V$, the number of states of $\mathcal{G}$ is thus polynomial in $|V|$. This is where the gain in complexity is made compared to the parity case studied in Section~\ref{sec:SolvingPCPGame}.

To solve the \NCRS problem, it remains to solve the game $\mathcal{G'}$ that can be done in time $n^{O(\log(d))}$ where $n$ is its number of vertices and $d$ its number of priorities~\cite{DBLP:journals/siamcomp/CaludeJKLS22}. From the previous arguments, it follows that \NCRS problem is solvable in time polynomial in $|V|$.
\end{proof}

\section{Conclusion}
\label{sec:conclusion}

In this work, we introduce a novel algorithm to solve the \NCRS problem for parity objectives. Unlike previous methods that converted the problem into a model-checking problem for Strategy Logic, our algorithm reduces the \NCRS problem to a three-player zero-sum game with imperfect information, framed as a Prover-Challenger game. This new angle yields improved complexity upper bounds: exponential time in the number of vertices of the game structure and the number of priorities of the parity objectives, doubly exponential time in the number of players. In particular, our algorithm runs within exponential time for a fixed number of players, which is particularly relevant since the number of players is typically small in practical scenarios. Moreover, we establish a lower bound that indicates the impossibility of solving the \NCRS problem in polynomial time unless ${\sf P=NP}$ even for a fixed number of players. For the particular case of reachability objectives, when the number of players is fixed, we prove polynomial complexity like for the NE-NCRS problem \cite{DBLP:conf/icalp/ConduracheFGR16}.

We believe that the Prover-Challenger framework, based on a three-player model with imperfect information, may be applicable for other synthesis challenges beyond our current application.

\bibliographystyle{plain}
\bibliography{prover-game}

\appendix

\section{Proofs of Lemmas~\ref{lem:playerStable} and \ref{lem:actionStable}} \label{app:Action-stability}

We begin with the proof of Lemma~\ref{lem:actionStable}.

\begin{proof}[Proof of Lemma~\ref{lem:actionStable}] 
First, suppose that $a'_1 = a'_2$ and let us show that $\Obs(u'_1) = \Obs(u'_2)$. If $\Obs(v'_1) = \Obs(v'_2)$ is the initial state $v'_0$ of the \pcpG, then $v'_1 = v'_2 = v'_0$ (by definition of $\Obs$) and we trivially get that $\Obs(u'_1) = \Obs(u'_2)$ is the initial state $v_0$ of $\mathcal G$. Suppose now that $v = \Obs(v'_1) = \Obs(v'_2)$ with $v \in V$. If $a'_1 = a'_2 \in A'_{P_1}$, then $v'_1, v'_2 \in V'_1$. Then in $G$, $\delta(v, a'_1) = u = \delta(v, a'_2)$. Hence by definition of $\delta'$ in \pcpG, $\Obs(u'_1) = u = \Obs(u'_2)$. We get the same conclusion if $a'_1 = a'_2 \in A'_{P_2}$. In the last case where $a'_1 = a'_2 \in A'_{C}$, by definition of $\delta'$, we get that $\Obs(u'_1) = \Obs(u'_2) = v$. 

Second, suppose that $\Obs(u'_1) = \Obs(u'_2)$ and $a'_1, a'_2$ are visible actions, that is, $a'_1, a'_2 \in \APun \cup \APd$. Let us show that $a'_1 = a'_2$. Let $v,u \in V$ be such that $v = \Obs(v'_1) = \Obs(v'_2)$ and $u = \Obs(u'_1) = \Obs(u'_2)$. By definition of \pcpG and since $\delta'(v'_1, a'_1) = u'_1$, $\delta'(v'_2, a'_2) = u'_2$, we have in $G$ that $\delta(v, a'_1) = u = \delta(v, a'_2)$. As $G$ is \actionstable (see Definition~\ref{def:game_structure}), it follows that $a'_1 = a'_2$. 
\end{proof}

From the definition of $\delta'$ and $\Obs$ and thanks to Lemma~\ref{lem:actionStable}, the next lemma is straightforward. It states what $P_1$ knows about the next state after observing an action.

\begin{lemm} \label{lem:remark} 
Let $v', u' \in V'$ be two states and $a' \in A'$ be an action such that $\delta'(v', a') = u'$.
\begin{itemize}
    \item Suppose that $a' \in A_{P_2}$. Then $u'$ is owned by player~$C$.
    \item Suppose that $a' \in A_C$. If $v' \neq v'_0$, then $\Obs'(u') = \Obs'(v')$, otherwise $\Obs'(u') = v_0$.
    \item Suppose that $a' \in A_{P_1}$. Then $u'$ is owned by either player~$P_1$ or player~$C$. Moreover, let $v'_1, u'_1 \in V'$ and $a'_1 \in A'$ be such that $\delta'(v'_1, a'_1) = u'_1$ and $\Obs'(v'a') = \Obs'(v'_1a'_1)$. Then $\Obs(u') = \Obs(u'_1)$ and $u',u'_1$ are owned by the same player.
\end{itemize}
\end{lemm}

Let us now prove Lemma~\ref{lem:playerStable}. 

\begin{proof}[Proof of Lemma~\ref{lem:playerStable}] 
Consider  in $G'$ two histories $h' = v'_0a'_0v'_1a'_1 \ldots v'_k$ and $g' = u'_0b'_0u'_1b'_1 \ldots u'_k$ such that $\Obs(h') = \Obs(g')$. We have to prove that $v'_\ell,u'_\ell$ belong to the same player, for all $\ell \in \{0,\ldots, k\}$. By definition of \pcpG, from the observation of $\Obs(a'_\ell)$ of each action $a'_\ell$ of $h'$, player~$P_1$ knows who controls $v'_\ell$: this is player~$C$ if $\Obs(a'_\ell) = \sharp$; player~$P_1$ if $\Obs(a'_\ell) = a'_\ell \in A_0$; player~$P_2$ if $\Obs(a'_\ell) = a'_\ell \in A \setminus A_0$. (Recall that $A_0$ and $A \setminus A_0$ are distinct sets, see Definition~\ref{def:game_structure}.) Therefore, as the actions are observed similarly in both histories $h'$ and $g'$, that is, $\Obs(a'_\ell) = \Obs(b'_\ell)$ for all $\ell$, it follows that $v'_\ell,u'_\ell$ belong to the same player, for all $\ell \in \{0,\ldots, k-1\}$. It remains to prove that this is also the case for the last pair of states $v'_k,u'_k$. Several cases have to be considered, by considering the definition of \pcpG and $\Obs$, and Lemma~\ref{lem:remark}.

Suppose first that $k = 0$. Then $h', g'$ are both equal to the initial state $v'_0$ owned by Challenger. Suppose now that $k>0$. If $v'_{k-1}, u'_{k-1}$ are owned by Prover~$1$ (resp. Prover~$2$), then $v'_k,u'_k$ are owned by the same player by Lemma~\ref{lem:remark}. We can thus suppose that $v'_{k-1}, u'_{k-1}$ are owned by Challenger. If $k = 1$, then $v'_{k-1}, u'_{k-1}$ are both the initial state $v'_0$, showing that $v'_k,u'_k$ have the same $G$-component $v_0$. It follows that $v'_k,u'_k$ are again owned by the same player. It remains to consider the case $k \geq 2$: \begin{itemize} 
\item If $v'_{k-2}, u'_{k-2}$ are both owned by either Prover~$1$ or Challenger, recalling that $v'_{k-1}, u'_{k-1}$ are owned by Challenger, we have that $v'_k,u'_k$ are owned by Prover~$2$. 
\item If $v'_{k-2}, u'_{k-2}$ are both owned by Prover~$2$, we can check that $v'_k,u'_k$ are owned by the same player (by definition of $\delta$ and $\Obs$). Indeed, $v'_{k-2}, u'_{k-2}$ are necessarily action-states (see Definition~\ref{def:pcp-game-structure-states} and Figure~\ref{fig_structure}). Thus, $v'_{k-1}, u'_{k-1}$ are owned by Challenger. Furthermore, the states $v'_k,u'_k$ are $G$-states with the same $G$-component, as $\Obs(v'_k) = \Obs(u'_k)$. Depending whether this $G$-component is owned in $G$ by player $0$ or another one, both $v'_k$ and $u'_k$ belong either to $P_1$ or $C$.
\end{itemize}
Therefore we have proved that $\Obs$ is \pstable.
\end{proof}

\begin{rema} \label{rem:playerStable} 
In the previous proof, for all $\ell \in \{0,\ldots k-1\}$, it was straightforward to derive that $v'_\ell,u'_\ell$ belong to the same player, from the hypothesis that the pair of actions $a'_\ell, b'_\ell$ is observed similarly. 
We need the stronger hypothesis $\Obs(h') = \Obs(g')$ to get that $v'_k,u'_k$ also belong to the same player.
\end{rema}

\section{Proof of Lemma~\ref{lem:strongBis}} \label{app:StrongPStable}

\begin{proof}[Proof of Lemma~\ref{lem:strongBis}] 
Let $h' = h'_1v'_k = v'_0a'_0v'_1a'_1 \ldots v'_k$ and $g' = g'_1u'_k = u'_0b'_0u'_1b'_1 \ldots u'_k$ be two histories in $G'$ such that $\Obs'(h'_1)=\Obs'(g'_1)$. Recall the certain regularity of the plays’ shape in the \pcpG, see Figure~\ref{fig:prover-game-play}.
Let us prove that $v'_k, u'_k$ are owned by the same player and that $\Obs(v'_k) = \Obs(u'_k)$. If $k = 0$, then $v'_k, u'_k$ are both equal to and observed as the initial state $v'_0$ of $G'$ that is controlled by player~$C$. Suppose that $k > 0$. Let us consider the pair of actions $a'_{k-1}, b'_{k-1}$ (which both belong to the same player since $\Obs'(h'_1) = \Obs'(g'_1)$ and $\Obs'$ is player-stable by Corollary~\ref{cor:playerStable}).
\begin{itemize}
\item If $a'_{k-1}, b'_{k-1} \in \ARondePd$, then $v'_{k-1}, u'_{k-1}$ are owned by player~$P$ and $v'_k, u'_k$ are owned by player~$C$ by definition of $\delta'$ and by Lemma~\ref{lem:remark}.
It follows that $\Obs(v'_k) = \Obs(v'_{k-1}) = \Obs(u'_{k-1}) = \Obs(u'_k)$ by definition of $\delta'$ and as $\Obs(h'_1) = \Obs(h'_2)$.
\item If $a'_{k-1}, b'_{k-1} \in A_{P_1}$, as $\Obs(v'_{k-1}a'_{k-1}) = \Obs(u'_{k-1}b'_{k-1})$, then $v'_k, u'_k$ are owned by the same player and $\Obs'(v'_k) = \Obs'(u'_k)$ by Lemma~\ref{lem:remark}.
\item If $a'_{k-1}, b'_{k-1} \in A_{C}$, as $\Obs(v'_{k-1}) = \Obs(u'_{k-1})$, then $\Obs(v'_k) = \Obs(u'_k)$ by Lemma~\ref{lem:remark}. Therefore $\Obs'(h') = \Obs'(g')$, and thus $\Obs(\cor{h'}) = \Obs(\cor{g'})$ by Lemma~\ref{lem:corresp}. As $\Obs$ is \pstable (see Lemma~\ref{lem:playerStable}), it follows that $v'_k, u'_k$ are owned by the same player.
\end{itemize} 
\end{proof}

\section{Proof of the complexity lower bounds of Theorem~\ref{thm:main}}  \label{app:lowerBounds}

In this appendix, we prove the complexity lower bounds of Theorem~\ref{thm:main}, second part. This proof is inspired by the one proposed for the NE-NCRS problem in \cite{DBLP:conf/icalp/ConduracheFGR16}, that we adapt in a way to deal with SPEs instead of NEs. 

\begin{proof}[Proof of Theorem~\ref{thm:main}, second part, hardness.]
We first prove that the \NCRS problem for parity games with three players is NP-hard. We use a reduction from solving a two-player zero-sum game with an intersection of two parity objectives. Deciding whether one player has a winning strategy for this kind of objective is coNP-hard~\cite{DBLP:conf/fossacs/ChatterjeeHP07}.

\begin{figure} 
\begin{center}
	\begin{tikzpicture}[->, >=latex,shorten >=1pt, scale=1, every node/.style={scale=1, align=center}]
	
	\node[draw,minimum height = 0.8cm, minimum width = 0.8cm] (v) at (1, 11) {$v$};
    \node[above, rotate=45] at (1.6,11.6) {$\in V'_0$};
	\node[draw, diamond] (v'-bar) at (5, 11) {$\bar{v}'$};
    \node[above, rotate=45] at (5.6,11.6) {$\in V'_1$};
	\node[draw, circle, minimum height = 0.8cm] (v') at (9, 11) {$v'$};
    \node[above, rotate=45] at (9.6,11.6) {$\in V'_2$};
	\node[draw, minimum height = 0.8cm, minimum width = 0.8cm] (s) at (7, 8.5) {$s$} ;
	\node[below] at (7,8) {$(0, 0, 0)$};
	
	\path[->] (v) edge node[right] {} (v'-bar);
    \path[->] (v) edge[dotted] (3,12);
    \path[->] (v) edge[dotted] (3,10);
	\path[->] (v'-bar) edge node[above] {$\bar{a}$} (v');
	
	\path[->] (v'-bar) edge node[right, pos = 0.3] {$b_1$} (s);
	\path[->] (v') edge node[right, pos = 0.3] {$b_2$} (s);
    \path[->] (v') edge[dotted] (11,12);
    \path[->] (v') edge[dotted] (11,11);
    \path[->] (v') edge[dotted] (11,10);
    \path[->] (s)  edge[loop right] node[right] {$b_0$} (s); 
    \draw[dashed] (0, 12.5) rectangle (12,9.5);
	\node (coin) at (1, 9.8) {$(0,\alpha'_1, \alpha'_2)$};

	\end{tikzpicture}
\end{center}
\caption{Reduction for the NP-hardness}
\label{fig:reduction1}
\end{figure}
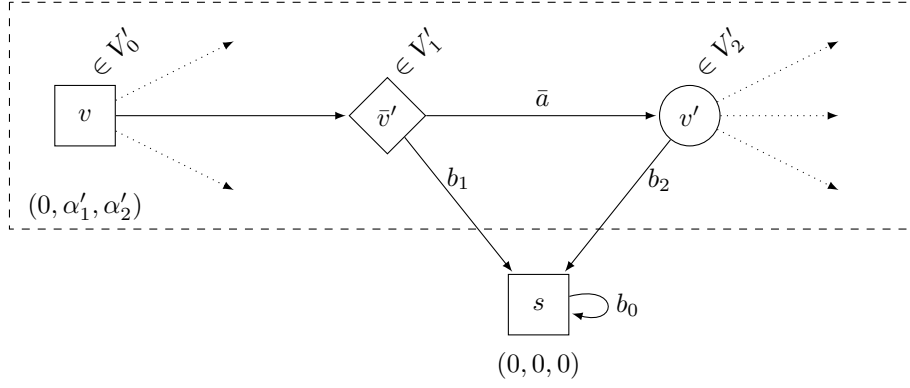

Consider a two-player zero-sum game with an intersection of two parity objectives  ${\mathcal G} = (G,\Parity(\alpha_1) \cap \Parity(\alpha_2))$ such that $G = (V,A,\{J_1,J_2\},\delta,v_0)$. We assume that player~$J_1$ wants to satisfy the objective $\Parity(\alpha_1) \cap \Parity(\alpha_2)$. W.l.o.g., we can suppose that the initial state $v_0$ is owned by player~$J_1$, that $A_{J_1} \cap A_{J_2} = \varnothing$,  and that $\delta(v,a) \in V_{J_2}$ for all $(v,a) \in V_{J_1} \times A_{J_1}$ and $\delta(v,a) \in V_{J_1}$ for all $(v,a) \in V_{J_2} \times A_{J_2}$ (the states of $V_{J_1}$ and $V_{J_2}$ alternate along the plays). From $\mathcal G$, we construct a parity game $\mathcal{G}' = (G', (\Parity(\alpha'_i))_{i\in\{0,1,2\}})$ with three players as follows. Its game structure $G'$ consists in a modified copy of $G$ where the states of $J_1$ are duplicated and where a new sink state $s$ is added with a self-loop labeled by a new action $b_0$.  

Formally, see also Figure~\ref{fig:reduction1}, we define: 
\begin{itemize}
    \item the set of states $V' = V'_0 \cup V'_1 \cup V'_2$ such that $V'_0 = V_{J_2} \cup \{s\}$, $V'_1 = \{\bar v \mid v \in V_{J_1}\}$, and $V'_2 = V_{J_1}$, 
    \item the set of actions $A' = A'_0 \cup A'_1 \cup A'_2$ such that $A'_0 = A_{J_2} \cup \{b_0\}$, $A'_1 = \{\bar a, b_1\}$ and $A'_2 = A_{J_1} \cup \{b_2\}$,
    \item the initial state $v_0$,
    \item the transition function $\delta' = \bigcup_{i=0,1,2} (V'_i \times A'_i) \rightarrow V'$ such that 
    \begin{itemize}
        \item $\delta'_0(v,a) = \overline{\delta(v,a)}$ for all $(v,a) \in V_{J_2} \times A_{J_2}$ and $\delta'_0(s,b_0)=s$,
        \item $\delta'_1(\bar v,\bar a) = v$ and $\delta'_1(\bar v,b_1) = s$ for all $\bar v \in V'_1$,
        \item $\delta'_2(v,a) = \delta(v,a)$ for all $(v,a) \in V_{J_1} \times A_{J_1}$ and $\delta'_2(v,b_2) = s$ for all $v \in V'_2$,
    \end{itemize} 
    \item the priority functions of the parity objectives $(\Parity(\alpha'_i))_{i\in\{0,1,2\}}$ such that 
    \begin{itemize}
        \item $\alpha'_0(v) = 1$ for all $v \in V' \setminus \{s\}$ and $\alpha'_0(s) = 0$,
        \item $\alpha'_1(v) = \alpha_1(v)$ for all $v \in V$, $\alpha'_1(\bar v) = \alpha_1(v)$ for all $\bar v \in V'_1$, and $\alpha'_1(s) = 0$, 
        \item $\alpha'_2(v) = \alpha_2(v)$ for all $v \in V$, $\alpha'_2(\bar v) = \alpha_2(v)$ for all $\bar v \in V'_1$, and $\alpha'_2(s) = 0$.
    \end{itemize}
\end{itemize}
Thus, the priority functions are such that all plays eventually looping in $s$ are winning for all players. The plays not visiting $s$ are losing for player~$0$ and have priorities depending on $\alpha_i$ for each player~$i \in \{1,2\}$.  

    We claim that player $J_1$ has a winning strategy from $v_0$ in $\mathcal G$ if, and only if, there is no solution to the \NCRS problem in $\mathcal G'$. As the first problem is coNP-hard, we will get that the \NCRS problem for parity games with three players is NP-hard. 
    \begin{itemize}
    \item Suppose first that player $J_1$ has a winning strategy $\sigma_{J_1}$ from $v_0$ in $\mathcal G$ for the objective $\Parity(\alpha_1) \cap \Parity(\alpha_2)$. Let $\sigma'_0$ be any strategy of player~$0$ in $\mathcal G'$. Let us show that $\sigma'_0$ is not solution to the \NCRS problem. The strategy $\sigma'_0$ can be viewed as a strategy $\sigma_{J_2}$ of player $J_2$ in $\mathcal G$. 
    By hypothesis, $\rho = \langle \sigma_{J_1},\sigma_{J_2} \rangle_{v_0}$ belongs to $\Parity(\alpha_1) \cap \Parity(\alpha_2)$. Let us show that there exists a subgame-perfect response $\bar \sigma'_{-0}$ to $\sigma'_0$ such that the play $\langle (\sigma'_0, \bar{\sigma}'_{-0}) \rangle_{v_0}$ is losing for player~$0$. From $\rho = v_0v_1v_2 \dots \in \Plays{G}$, we define $\rho' = v_0 \bar v_1 v_1 v_2 \bar v_3 v_3  \dots v_{2k} \bar v_{2k+1} v_{2k+1} \dots \in \Plays{G'}$. Notice that $\rho'$ is winning for both players~$1$ and~$2$. We define two strategies $\sigma'_1$ and $\sigma'_2$ respectively for player~$1$ and player~$2$ in $\mathcal G'$ such that 
    \begin{itemize}
        \item $\rho'  =  \langle \bar \sigma'\rangle_{v_0} = \langle \sigma'_0,\sigma'_1,\sigma'_2 \rangle_{v_0}$ ($\sigma'_2$ mimics $\sigma_{J_1}$ and for all $k \in \N$, $\sigma'_1(v_0 \bar v_1 v_1 \dots v_{2k} \bar v_{2k+1}) = v_{2k+1}$),
        \item for any history $h'$ not prefix of $\rho'$, $\sigma'_i(h') = s$ if $h'$ ends with a state owned by player~$i \in \{1,2\}$.
    \end{itemize}
    The outcome $\rho'$ of $\bar \sigma'$ is losing for player~$0$ as $\rho'$ never visits the sink state~$s$. Moreover, $\bar \sigma'$ is a 0-fixed SPE. Indeed, for any history $h'$, $\bar{\sigma}'_{\upharpoonright h'}$ is a 0-fixed NE in the subgame $G_{\upharpoonright h'}$ as $\rho'$ and any play eventually looping in $s$ is winning for both players~$1$ and~$2$ (they have no incentive to deviate).
    \item Suppose now that player $J_2$ has a winning strategy $\sigma_{J_2}$ from $v_0$ in $\mathcal G$ for the objective opposite to $\Parity(\alpha_1) \cap \Parity(\alpha_2)$. The strategy $\sigma_{J_2}$ can be viewed as a strategy $\sigma'_0$ for player~$0$ in $\mathcal G'$. Let us show that $\sigma'_0$ is a solution to the \NCRS problem. Consider any play $\rho' = v_0 \bar v_1 v_1 v_2 \bar v_3 v_3  \dots v_{2k} \bar v_{2k+1} v_{2k+1} \dots \in \Plays{G'}$ compatible with $\sigma'_0$ that does not visit $s$. By eliminating the copies $\bar{v}_{2k+1}$ in $\rho'$, we get a play $\rho = v_0v_1v_2 \dots $ compatible with $\sigma_{J_2}$, thus losing for $\Parity(\alpha_1) \cap \Parity(\alpha_2)$. Hence, $\rho'$ is not a $0$-fixed NE outcome (and thus not a 0-fixed SPE outcome), as at least one player $i \in \{1,2\}$ will has an incentive to deviate in the sink state~$s$. It follows that for all subgame-perfect responses $\bar \sigma'_{-0}$ to $\sigma'_0$, the play $\langle \sigma'_0, \bar \sigma'_{-0}\rangle_{v_0}$ eventually loops in $s$ and thus is winning for player~$0$. 
    \end{itemize}

\begin{figure} 
\begin{center}
	\begin{tikzpicture}[->, >=latex,shorten >=1pt, scale=1, every node/.style={scale=1, align=center}]
	
	\node[draw,minimum height = 0.8cm, minimum width = 0.8cm] (v) at (1, 11) {$v$};
    \node[above, rotate=45] at (1.6,11.6) {$\in V'_0$};
	\node[draw, diamond] (v'-bar-bar) at (5, 11) {$\bar{\bar{v}}'$};
    \node[above, rotate=45] at (5.6,11.6) {$\in V'_1$};
	\node[draw, regular polygon,regular polygon sides=6,inner sep=1.5pt, minimum height = 0.8cm] (v-bar') at (9, 11) {$\bar{v}'$};
    \node[above, rotate=45] at (9.6,11.6) {$\in V'_2$};
    \node[draw, circle, minimum height = 0.8cm] (v') at (13, 11) {$v'$};
    \node[above, rotate=45] at (13.6,11.6) {$\in V'_3$};
	\node[draw, circle, minimum height = 0.8cm, minimum width = 0.8cm] (s1) at (5, 8.5) {$s_1$} ;
	\node[below] at (5,8) {$(1, 1, 0, 0)$};
 \node[draw, circle, minimum height = 0.8cm, minimum width = 0.8cm] (s2) at (9, 8.5) {$s_2$} ;
	\node[below] at (9,8) {$(1, 0, 1, 0)$};
 \node[draw, circle, minimum height = 0.8cm, minimum width = 0.8cm] (s3) at (13, 8.5) {$s_3$} ;
	\node[below] at (13,8) {$(0, 0, 0, 0)$};
	
	\path[->] (v) edge node[right] {} (v'-bar);
    \path[->] (v) edge[dotted] (3,12);
    \path[->] (v) edge[dotted] (3,10);
	\path[->] (v'-bar-bar) edge node[above] {$\bar{\bar{a}}$} (v-bar');
    \path[->] (v-bar') edge node[above] {$\bar{a}$} (v');
	
	\path[->] (v'-bar-bar) edge node[right, pos = 0.3] {$b_1$} (s1);
    \path[->] (v-bar') edge node[right, pos = 0.3] {$b_2$} (s2);
	\path[->] (v') edge node[right, pos = 0.3] {$b_3$} (s3);
    \path[->] (v') edge[dotted] (15,12);
    \path[->] (v') edge[dotted] (15,11);
    \path[->] (v') edge[dotted] (15,10);
    \path[->] (s1)  edge[loop right] node[right] {$b_0$} (s1); 
    \path[->] (s2)  edge[loop right] node[right] {$b_0$} (s2); 
    \path[->] (s3)  edge[loop right] node[right] {$b_0$} (s3); 
    \draw[dashed] (0, 12.5) rectangle (16,9.5);
	\node (coin) at (1, 9.8) {$(0,\alpha_1', \alpha_2', 0)$};

	\end{tikzpicture}
\end{center}
\caption{Reduction for the coNP-hardness}
\label{fig:reduction2}
\end{figure}
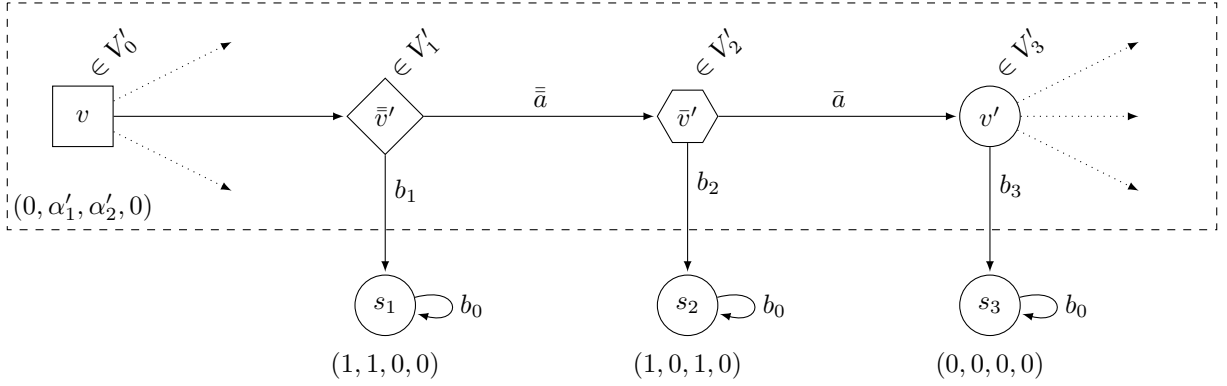

\medskip
Let us now prove that the \NCRS problem for parity games with four players is coNP-hard. We consider a two-player zero-sum game with an intersection of two parity objectives ${\mathcal G} = (G,\Parity(\alpha_1) \cap \Parity(\alpha_2))$ such that $G = (V,A,\{J_1,J_2\},\delta,v_0)$, with the same assumptions as in the first part of the proof. From $\mathcal G$, we construct a parity game $\mathcal{G}' = (G', (\Parity(\alpha'_i))_{i\in\{0,1,2,3\}})$ with four players as follows. Its game structure $G'$ consists in a modified copy of $G$ obtained by triplicating the states of $J_2$ and by adding three new sink states $s_i$, $i \in \{1,2,3\}$, with a self-loop. Formally, see also Figure~\ref{fig:reduction2}, we define: 
\begin{itemize}
    \item the set of states $V' = V'_0 \cup V'_1 \cup V'_2 \cup V'_3$ such that $V'_0 = V_{J_1} \cup \{s_1,s_2,s_3\}$, $V'_1 = \{\bar{\bar v} \mid v \in V_{J_2}\}$, $V'_2 = \{\bar v \mid v \in V_{J_2}\}$, and $V'_3 = V_{J_2}$,
    \item the set of actions $A' = A'_0 \cup A'_1 \cup A'_2 \cup A'_3$ such that $A'_0 = A_{J_1} \cup \{b_0\}$, $A'_1 = \{\bar{\bar a}, b_1\}$, $A'_2 = \{\bar a, b_2\}$, and $A'_3 = A_{J_2} \cup \{b_3\}$,
    \item the initial state $v_0$,
    \item the transition function $\delta' = \bigcup_{i=0,1,2,3} (V'_i \times A'_i) \rightarrow V'$ such that 
    \begin{itemize}
        \item $\delta'_0(v,a) = \overline{\overline{\delta(v,a)}}$ for all $(v,a) \in V_{J_1} \times A_{J_1}$, and $\delta'_0(s_i,b_0)=s_i$, for all $i \in \{1,2,3\}$,
        \item $\delta'_1(\bar{\bar v},\bar{\bar a}) = \bar{v}$ and $\delta'_1(\bar{\bar v},b_1) = s_1$ for all $\bar{\bar v} \in V'_1$,
        \item $\delta'_2(\bar v,\bar a) = v$ and $\delta'_1(\bar v,b_2) = s_2$ for all $\bar v \in V'_2$,
        \item $\delta'_3(v,a) = \delta(v,a)$ for all $(v,a) \in V_{J_2} \times A_{J_2}$, $\delta'_3(v,b_3) = s_3$ for all $v \in V_{J_2}$,
    \end{itemize} 
    \item the priority functions of the parity objectives $(\Parity(\alpha'_i))_{i\in\{0,1,2,3\}}$ such that 
    \begin{itemize}
        \item $\alpha'_0(v) = 0$ for all $v \in V' \setminus \{s_1,s_2,s_3\}$,  $\alpha'_0(s_1) = \alpha'_0(s_2) = 1$, and $\alpha'_0(s_3) = 0$,
        \item $\alpha'_1(v) = \alpha_1(v)$ for all $v \in V$, $\alpha'_1(\bar{\bar v}) = \alpha'_1(\bar v) = \alpha_1(v)$ for all $v \in V_{J_2}$, $\alpha'_1(s_1) = 1$, and $\alpha'_1(s_2) = \alpha'_1(s_3) = 0$,
        \item $\alpha'_2(v) = \alpha_2(v)$ for all $v \in V$, $\alpha'_2(\bar{\bar v}) = \alpha'_2(\bar v) = \alpha_2(v)$ for all $v \in V_{J_2}$, $\alpha'_2(s_1) = \alpha'_2(s_3) = 0$, and $\alpha'_2(s_2) = 1$,
        \item $\alpha'_3(v) = 0$ for all $v \in V'$.
    \end{itemize}
\end{itemize}
Thus, all the plays are winning for player~$3$. The plays winning for player~$0$ are those not visiting the sink states $s_1$ and $s_2$. And for each player~$i \in \{1,2\}$, it depends on the priority function~$\alpha_i$ and on the visited sink state.

    We claim that player $J_1$ has a winning strategy from $v_0$ in $\mathcal G$ if, and only if, there is a solution to the \NCRS problem in $\mathcal G'$. As the first problem is coNP-hard, we will get that the \NCRS problem for parity games with four players is coNP-hard. 
    \begin{itemize}
    \item Suppose first that player $J_1$ has a winning strategy $\sigma_{J_1}$ from $v_0$ in $\mathcal G$ for the objective $\Parity(\alpha_1) \cap \Parity(\alpha_2)$. The strategy $\sigma_{J_1}$ can be viewed as a strategy $\sigma'_0$ for player~$0$ in $\mathcal G'$. Let us show that $\sigma'_0$ is a solution to the \NCRS problem. Let us consider the different plays $\rho' \in \Plays{G'}$ compatible with $\sigma'_0$:
    \begin{itemize}
        \item Suppose that $\rho'$ visits no sink state, i.e., $$\rho' = v_0 \bar{\bar{v}}_1 \bar v_1 v_1 v_2 \bar{\bar{v}}_3 \bar v_3 v_3  \dots v_{2k} \bar{\bar{v}}_{2k+1} \bar v_{2k+1} v_{2k+1} \dots$$ Notice that it is winning for both players~$0$ and~$3$. By eliminating the copies $\bar{\bar{v}}_{2k+1}$ and $\bar{v}_{2k+1}$  in $\rho'$ for all $k \in \N$, we get a play $\rho = v_0v_1v_2 \dots \in \Plays{G}$ compatible with $\sigma_{J_1}$, thus winning for $\Parity(\alpha_1) \cap \Parity(\alpha_2)$. It follows that $\rho'$ is also winning for both players~$1$ and~$2$, thus for all players. Hence, $\rho'$ is a 0-fixed NE outcome (no player has an incentive to deviate).
        \item Suppose that $\rho'$ eventually loops in $s_3$. Then it is also winning for all players and thus a 0-fixed NE outcome.
        \item  Suppose that $\rho'$ eventually loops in $s_1$. Then, it is losing for player~$1$ and not a 0-fixed NE outcome. Indeed, player~$1$ has a profitable deviation by always playing the action $\bar{\bar a}$ (the resulting deviating play either loops in $s_2$, or loops in $s_3$, or visits no sink state).
        \item Suppose that $\rho'$ eventually loops in $s_2$. With an argument similar to the previous one, we get that $\rho'$ is not a $0$-fixed NE outcome. 
    \end{itemize}
        Therefore, as all subgame-perfect responses $\bar \sigma'_{-0}$ to $\sigma'_0$ are  in particular 0-fixed NEs, their outcome either visits no sink state or eventually loops in $s_3$. In both cases, this outcome is winning for player~$0$.
   
    \item  Suppose now that player $J_2$ has a winning strategy 
    from $v_0$ in $\mathcal G$ for the objective $W$ which is the opposite of $\Parity(\alpha_1) \cap \Parity(\alpha_2)$. Let $\sigma'_0$ be any strategy of player~$0$ in $\mathcal G'$. Let us show that $\sigma'_0$ is not solution to the \NCRS problem. The strategy $\sigma'_0$ can be viewed as a strategy $\sigma_{J_1}$ of player $J_1$ in $\mathcal G$. 
    Let us show that there exists a subgame-perfect response $\bar \sigma'_{-0}$ to $\sigma'_0$ such that the play $\langle (\sigma'_0, \bar{\sigma}'_{-0}) \rangle_{v_0} = \langle \bar{\sigma}' \rangle_{v_0}$ is losing for player~$0$.

\begin{figure} 
\centering

\begin{subfigure}[b]{\textwidth}

	\begin{tikzpicture}[->, >=latex,shorten >=1pt, scale=1, every node/.style={scale=1, align=center}]
	
	\node[draw, diamond] (v'-bar-bar) at (5, 11) {$\bar{\bar{v}}$};
	\node[draw, regular polygon,regular polygon sides=6,inner sep=1.5pt, minimum height = 0.8cm] (v-bar') at (9, 11) {$\bar{v}$};
    \node[draw, circle, minimum height = 0.8cm] (v') at (13, 11) {$v$};
	\node[draw, circle, minimum height = 0.8cm, minimum width = 0.8cm] (s1) at (5, 8.5) {$s_1$} ;
 
    \path[->] (v-bar') edge node[above] {$\bar{a}$} (v');
	
	\path[->] (v'-bar-bar) edge node[right, pos = 0.3] {$b_1$} (s1);

    \path[->] (v') edge[dotted] (15,11);
    
    \path[->] (s1)  edge[loop right] node[right] {$b_0$} (s1); 

	\end{tikzpicture}

\caption{Case $1(a)$}
\vspace{1cm}
\label{fig_case1a}
\end{subfigure}

\begin{subfigure}[b]{\textwidth}

	\begin{tikzpicture}[->, >=latex,shorten >=1pt, scale=1, every node/.style={scale=1, align=center}]
	
	\node[draw, diamond] (v'-bar-bar) at (5, 11) {$\bar{\bar{v}}$};
    
	\node[draw, regular polygon,regular polygon sides=6,inner sep=1.5pt, minimum height = 0.8cm] (v-bar') at (9, 11) {$\bar{v}$};
    
    \node[draw, circle, minimum height = 0.8cm] (v') at (13, 11) {$v$};
    
    \node[draw, circle, minimum height = 0.8cm, minimum width = 0.8cm] (s2) at (9, 8.5) {$s_2$} ;
	
	\path[->] (v'-bar-bar) edge node[above] {$\bar{\bar{a}}$} (v-bar');
    
    \path[->] (v-bar') edge node[right, pos = 0.3] {$b_2$} (s2);
	
    \path[->] (v') edge[dotted] (15,11);
    
    \path[->] (s2)  edge[loop right] node[right] {$b_0$} (s2); 

	\end{tikzpicture}

\caption{Case $1(b)$}

\label{fig_case1b}
\vspace{1cm}
\end{subfigure}

\begin{subfigure}[b]{\textwidth}

	\begin{tikzpicture}[->, >=latex,shorten >=1pt, scale=1, every node/.style={scale=1, align=center}]
	
	\node[draw, diamond] (v'-bar-bar) at (5, 11) {$\bar{\bar{v}}$};
    
	\node[draw, regular polygon,regular polygon sides=6,inner sep=1.5pt, minimum height = 0.8cm] (v-bar') at (9, 11) {$\bar{v}$};
   
    \node[draw, circle, minimum height = 0.8cm] (v') at (13, 11) {$v$};
    
 \node[draw, circle, minimum height = 0.8cm, minimum width = 0.8cm] (s3) at (13, 8.5) {$s_3$} ;
	
	\path[->] (v'-bar-bar) edge node[above] {$\bar{\bar{a}}$} (v-bar');
    \path[->] (v-bar') edge node[above] {$\bar{a}$} (v');

	\path[->] (v') edge node[right, pos = 0.3] {$b_3$} (s3);
   
    \path[->] (s3)  edge[loop right] node[right] {$b_0$} (s3); 
   
	\end{tikzpicture}
\caption{Case $2$}

\label{fig_case2}
\end{subfigure}

\caption{A subgame-perfect response to $\sigma'_0$}
\label{fig:subgameResponse}
\end{figure}
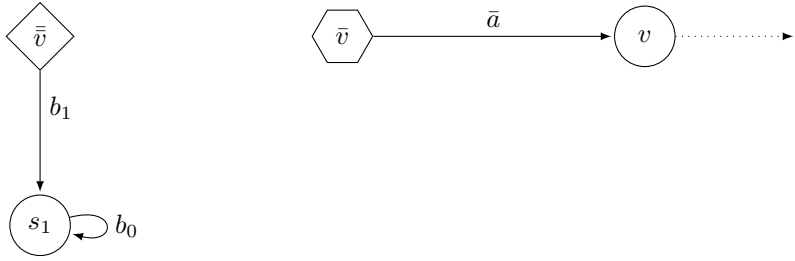
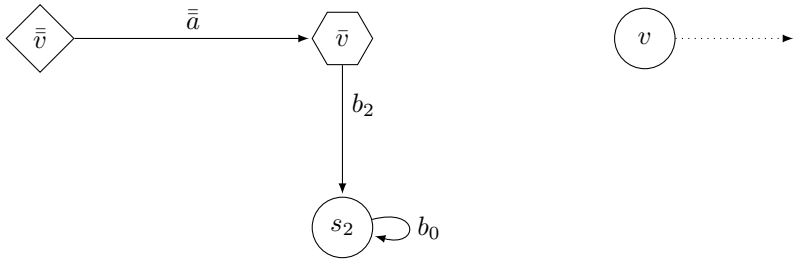
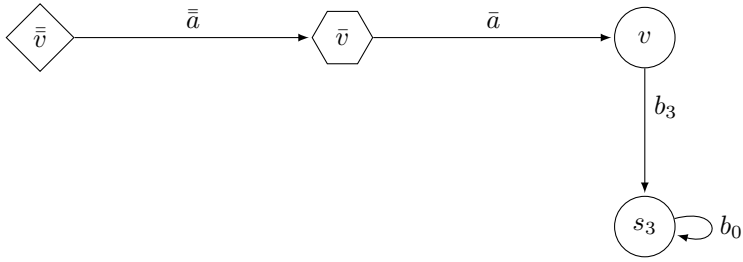

    We define $\bar \sigma'_{-0}$ as follows (see Figure~\ref{fig:subgameResponse}). Let $h'v'$ be any history compatible with $\sigma'_0$, such that its last state $v'$ belongs to $V'_i$, for some $i \in \{1,2,3\}$. This state is either a state $v$ in $G$ or one of its copies $\bar v$ or $\bar{\bar v}$. We denote by $hv$ the history in $G$ that corresponds to $h'v'$ (any state $u$ in $V_2$ and its copies $\bar{\bar u}$ and $\bar u$ are merged into one state $u$)
    \begin{enumerate}
        \item Suppose first that $J_2$ has a winning strategy $\sigma_{J_2,v}$ from $v$ in $\mathcal G$ for the objective $W$. Then the outcome $\pi' = \langle \sigma_{J_1,\upharpoonright hv}, \sigma_{J_2,v}\rangle_{v}$ belongs to $W$.
        (a) If $\pi'$ does not satisfy $\Parity(\alpha_1)$,  then 
        \begin{itemize}
            \item if $v' \in V'_1$, then player~$1$ chooses action~$b_1$,
            \item if $v' \in V'_2$, then player~$2$ chooses action~$\bar{a}$,
            \item if $v' \in V'_3$, then player~$3$ chooses the action dictated by $\sigma_{J_2,v}$.
        \end{itemize}
        (b) Otherwise, $\pi'$ does not satisfy $\Parity(\alpha_2)$. Then \begin{itemize}
            \item if $v' \in V'_1$, then player~$1$ chooses action~$\bar{\bar a}$,
            \item if $v' \in V'_2$, then player~$2$ chooses action~$b_2$,
            \item if $v' \in V'_3$, then player~$3$ chooses the action dictated by $\sigma_{J_2,v}$.
        \end{itemize}
        \item Suppose now that $J_2$ has no winning strategy from $v$ in $\mathcal G$ for the objective $W$. Then
        \begin{itemize}
            \item if $v' \in V'_1$, then player~$1$ chooses action~$\bar{\bar a}$,
            \item if $v' \in V'_2$, then player~$2$ chooses action~$\bar a$,
            \item if $v' \in V'_3$, then player~$3$ chooses action $b_3$. 
        \end{itemize}
    \end{enumerate}
    Notice that in case 1, when $v' \in V '_3$, player 3 plays an action leading to a state from which $J_2$ has again a winning stategy, thus to a situation where case 1 still holds. 
    Let us show that $\bar \sigma'_{-0}$ so defined is a subgame-perfect response to $\sigma'_0$. Consider any subgame $\mathcal{G}'_{\upharpoonright h'v'}$ of $\mathcal{G}'$. Recall that $\sigma'_0$ is fixed. Hence, there are essentially three cases to discuss: 
    \begin{itemize}
        \item Suppose that the outcome $\langle \bar \sigma'_{\upharpoonright h'v'}\rangle_{v'}$ eventually loops in $s_1$ (thus case 1(a)). Then, this outcome is an NE outcome. Indeed, it is losing for player~$1$, who has no profitable deviation as $\pi'$ does not satisfy $\Parity(\alpha_1)$ in case 1(a). Player~$2$ has no incentive to deviate as looping in $s_2$ is losing for him, and player~$3$ has also no incentive to deviate as he always wins in $\mathcal G'$.
        \item Suppose that  $\langle \bar \sigma'_{\upharpoonright h'v'}\rangle_{v'}$ eventually loops in $s_2$  (thus case 1(a)). This outcome is also an NE outcome. Indeed, it is losing for player~$2$, who has no profitable deviation as $\pi'$ does not satisfy $\Parity(\alpha_2)$ in case 1(b). Player~$1$ has no incentive to deviate as looping in $s_1$ is losing for him, and player~$3$ always wins in $\mathcal G'$. 
        \item Suppose finally that $\langle \bar \sigma'_{\upharpoonright h'v'}\rangle_{v'}$ eventually loops in $s_3$ (thus case 2). It is again an NE outcome as it is winning for all players. 
    \end{itemize}
    Therefore, $\sigma'_{-0}$ is a subgame-perfect response to $\sigma'_0$. Finally, notice that at the initial state $v_0$, the outcome of $\bar \sigma'$ loops in either $s_1$ (case 1(a)) or $s_2$ (case 1(b)), because $J_2$ has a winning strategy from $v_0$ by hypothesis. In both cases, this outcome is losing for player~$0$.
    \end{itemize}
\end{proof}

\section{Proof of Lemma~\ref{lem:crucial}} \label{App:dernierlem}

\begin{proof}[Proof of Lemma~\ref{lem:crucial}]\footnote{In this proof, we use the notations of Definitions~\ref{def:pcp-game-structure-states}-\ref{def:pcp-game-structure-transitions} of the observer definition.} 
We are going to prove that $\pi = u'_0b'_0u'_1b'_1 \ldots$ is a play in $G'$. From the synchronized product of this play with the path $p_0 \xrightarrow{u'_0b'_0} p_1 \xrightarrow{u'_1b'_1} \ldots$ in $\mathcal O$, it will follow that $(u'_0,p_0)b'_0(u'_1,p_1)b'_1 \ldots$ is a play in $G' \times \mathcal{O}$.

Assume that $\pi$ is not a play in $G'$. Consider the play $\rho_{G'} = v'_0a'_0v'_1a'_1 \ldots$ being the projection of $\rho$ on the game structure $G'$ (we remove the \OComponent of all states of $\rho$). Recall that $\Obs(\rho) = \Obs(v'_0a'_0v'_1a'_1 \ldots)$ and that it is equal to $\Obs(\pi)$ by hypothesis. Consider the longest prefix $g'$ of $\pi$ that is a history in $G'$. Note that $g$ exists as $\Obs(v'_0) = \Obs(u'_0)$ is the initial state of $G'$. Let $h'$ be the prefix of $\rho_{G'}$ such that $|h'| = |g'|$ and let $v'_{k}$ and $u'_{k}$ be the last state of $h'$ and $g'$ respectively. As $\Obs(h') = \Obs(g')$ by hypothesis and $\Obs$ is \pstable by Lemma~\ref{lem:playerStable}, it follows that $v'_{k}$ and $u'_{k}$ both belong to the same player.

To get a contradiction, we show that $g'b'_ku'_{k+1}$ is still a history of $G'$. We proceed case by case (recall the certain regularity of the plays’ shape in the game \pcpG, see Figure~\ref{fig:prover-game-play}): 
\begin{itemize}
\item If $v'_k, u'_k$ belong to Prover~1, as $\Obs(v'_ka'_k) =  \Obs(u'_kb'_k)$ by hypothesis, then $v'_k = (v,\bar g)$, $u'_k = (v,\bar \gamma)$ for some $v \in V_0$ and $\bar g, \bar \gamma \in \{0,1\}^{|\Pi|}$, and $a'_k, b'_k$ are both visible actions and thus equal to some $a \in A_0$. By definition of the transition function $\delta'$ of $G'$, we have $\delta'(v'_k,a'_k) = (\delta(v,a),\bar g)$ and $\delta'(u'_k,b'_k) = (\delta(v,a),\bar \gamma)$. So we have to prove that $u'_{k+1} = (\delta(v,a),\bar \gamma)$. For the $G$-component, this follows from $\delta(v,a) = \Obs(v'_{k+1}) = \Obs(u'_{k+1})$. For the \NegComponent, this follows from the definition of the set $E$ of $\ObsAut$: the existence of the two consecutive edges $p_k \xrightarrow{u'_kb'_k} p_{k+1} \xrightarrow{u'_{k+1}b'_{k+1}}p_{k+2}$ impose that the \NegComponent of $p_{k+1}$ is equal to the \NegComponent of both $\delta'(u'_k,b'_k)$ and $u'_{k+1}$, that is, to $\bar \gamma$. It follows that $u'_{k+1} = (\delta(v,a),\bar \gamma)$.

The following cases are treated with the same approach.

\item If $v'_k, u'_k$ belong to Challenger such that $v'_k = (v,\bar g)$ for some $v \in V_i$ with $i \neq 0$ and some $\bar g \in \{0,1\}^{|\Pi|}$, then $u'_k = (v, \bar \gamma)$ for some $\bar \gamma \in \{0,1\}^{|\Pi|}$, and both actions $a'_k, b'_k$ are invisible (they are observed as \#). We have $\delta'(v'_k,a'_k) = (v,a,\bar g)$ and $\delta'(u'_k,b'_k) = (v,b,\bar \gamma)$ for some $a, b \in A_i$. So we have to prove that $u'_{k+1} = (v,b,\bar \gamma)$. For the $G$-component, this follows from $v = \Obs(v'_{k+1}) = \Obs(u'_{k+1})$. For the \NegComponent, this follows from the  existence of the two consecutive edges $p_k \xrightarrow{u'_kb'_k}p_{k+1} \xrightarrow{u'_{k+1}b'_{k+1}}p_{k+2}$ in $E$ imposing that the \NegComponent of $p_{k+1}$ is equal to the \NegComponent of both $\delta'(u'_k,b'_k)$ and $u'_{k+1}$, that is, to $(b, \bar \gamma)$.

\item If $v'_k, u'_k$ belong to Challenger such that $v'_k = (v,i,\bar g)$ for some $v \in V$, $i \in \Pi \cup \{\varnothing\}$ and $\bar g \in \{0,1\}^{|\Pi|}$, then $u'_k = (v, j, \bar \gamma)$ for some $j \in \Pi \cup \{\varnothing\}$ and $\bar \gamma \in \{0,1\}^{|\Pi|}$, and both actions $a'_k, b'_k$ are observed as \#. We have $\delta'(v'_k,a'_k) = (v,\bar{g}')$ with $\bar{g}' = \bar g$ if $i = \varnothing$ and $\delta'(u'_k,b'_k) = (v,\bar{\gamma}')$ with $\bar{\gamma}' = \bar \gamma$ if $j = \varnothing$. So we have to prove that $u'_{k+1} = (v,\bar{ \gamma}')$. For the $G$-component, this follows from $v = \Obs(v'_{k+1}) = \Obs(u'_{k+1})$. For the \NegComponent, this follows from $p_k \xrightarrow{u'_kb'_k} p_{k+1} \xrightarrow{u'_{k+1}b'_{k+1}}p_{k+2}$ imposing that the \NegComponent of $\delta'(u'_k,b'_k)$ and $u'_{k+1}$ are both equal to $\bar{\gamma}'$. 

\item If $v'_k, u'_k$ belong to Prover~2 such that $v'_k = (v,a,\bar g)$ for some $v \in V_i$, $a \in A_i$ (with $i \neq 0$), and $\bar g \in \{0,1\}^{|\Pi|}$, then $u'_k = (v, b, \bar \gamma)$ for some $b \in A_i$ and $\bar \gamma \in \{0,1\}^{|\Pi|}$, and both actions $a'_k, b'_k$ are visible and equal to some $a' \in A_i$. We have $\delta'(v'_k,a'_k) = (\delta(v,a'),i,\bar g)$ and $\delta'(u'_k,b'_k) = (\delta(v,a'),j,\bar \gamma)$ for some $i, j \in \Pi \cup \{\varnothing\}$. So we have to prove that $u'_{k+1} = (\delta(v,a'),j,\bar \gamma)$. For the $G$-component, this follows from $\delta(v,a') = \Obs(v'_{k+1}) = \Obs(u'_{k+1})$. For the \NegComponent, this follows from $p_k \xrightarrow{u'_kb'_k} p_{k+1} \xrightarrow{u'_{k+1}b'_{k+1}}p_{k+2}$ imposing that the \NegComponent of $\delta'(u'_k,b'_k)$ and $u'_{k+1}$ are both equal to $(j, \bar \gamma)$.
\end{itemize}

We thus obtain the announced contradiction. It follows that $\pi = u'_0b'_0u'_1b'_1 \ldots$ is a play in $G'$ and then $(u'_0,p_0)b'_0(u'_1,p_1)b'_1 \ldots$ is a play in $G' \times \mathcal{O}$.
\end{proof}

\end{document}